\documentclass[a4paper,10pt,reqno]{amsart}
\usepackage[foot]{amsaddr}

\usepackage[a4paper, margin=2.5cm]{geometry}

\usepackage{amssymb}
\usepackage{amsmath}
\numberwithin{equation}{section}

\usepackage{slashed}
\usepackage{tensor}
\usepackage{float}
\usepackage{nicefrac}
\usepackage{enumerate}
\usepackage{cancel}

\allowdisplaybreaks

\usepackage{soul}

\usepackage{graphicx}
\usepackage{caption}
\usepackage{subcaption}
\usepackage{amsthm}
\usepackage{latexsym}
\usepackage[dvips]{epsfig}

\usepackage{wasysym}
\usepackage{mathrsfs}
\usepackage{bm}
\usepackage{tikz}
\usepackage{slashed}
\usepackage{yhmath} 
\usepackage{stmaryrd}

\usepackage[utf8]{inputenc}
\usepackage[english]{babel}
\usepackage{csquotes}

\usepackage[
backend=biber,
style=alphabetic,
sorting=ynt,
maxbibnames=10,
maxalphanames=4
]{biblatex}

\usepackage{xpatch}
 
\usepackage{multicol}
\usepackage{color}
 \usepackage{soul}
\usepackage{comment}
 
\usepackage{hyperref}
\hypersetup{
    colorlinks,
    linkcolor={blue!60!black},
    citecolor={red!60!black},
    urlcolor={blue!80!black}
}
\usepackage{cleveref}

\def\bme{{\bm e}}

\def\bmzero{{\bm 0}}
\def\bmone{{\bm 1}}

\def\bmA{{\bm A}}
\def\bmB{{\bm B}}
\def\bmC{{\bm C}}

\def\bmK{{\bm K}}

\def\bmQ{{\bm Q}}

\def\bmS{{\bm S}}

\def\bmX{{\bm X}}
\def\bmZ{{\bm Z}}

\def\bmeta{{\bm \eta}}

\def\bmsigma{{\bm \sigma}}

\def\bmpartial{{\bm \partial}}

\def\scri{{\mathscr I}}

\newcounter{mnotecount}

\newcommand{\mnotex}[1]
{\protect{\stepcounter{mnotecount}}$^{\mbox{\footnotesize $\bullet$\themnotecount}}$ 
\marginpar{
\raggedright\tiny\em
$\!\!\!\!\!\!\,\bullet$\themnotecount: #1} }

\renewcommand{\leq}{\leqslant}
\renewcommand{\geq}{\geqslant}
\renewcommand{\d}{\mathrm{d}}

\theoremstyle{plain}
\newtheorem{proposition}{Proposition}
\newtheorem{lemma}{Lemma}
\newtheorem{theorem}{Theorem}

\newtheorem{remark}{Remark}

\numberwithin{proposition}{section}
\numberwithin{lemma}{section}
\numberwithin{theorem}{section}
\numberwithin{assumption}{section}
\numberwithin{corollary}{section}
\numberwithin{remark}{section}

\begin{document}

\title[Regularity at future null infinity from past asymptotic data]{Controlled regularity at future null infinity from past asymptotic initial data: the wave equation}

\author[J. Marajh]{Jordan Marajh$^1$}
\address[1]{School of Mathematical Sciences, Queen Mary, University of London, Mile End Road, London E1 4NS, UK.}
\email{j.marajh@qmul.ac.uk}

\author[G. Taujanskas]{Grigalius Taujanskas$^2$}
\address[2]{Department of Pure Mathematics and Mathematical Statistics, University of Cambridge, CB3 0WB, UK.}
\email{taujanskas@dpmms.cam.ac.uk}

\author[J. A. Valiente Kroon]{Juan A. Valiente Kroon$^1$}
\email{j.a.valiente-kroon@qmul.ac.uk}



\begin{abstract}
We study the relationship between asymptotic characteristic initial data for the wave equation at past null infinity and the regularity of the solution at future null infinity on the Minkowski spacetime. By constructing estimates on a causal rectangle reaching the conformal boundary, we prove that the solution admits an asymptotic expansion near {future} null and spatial infinity whose regularity is controlled quantitatively in terms of the regularity of the data at past null infinity. In particular, our method gives rise to solutions to the wave equation in a neighbourhood of spatial infinity satisfying the peeling behaviour, for data on past null infinity with non-compact support.  Our approach makes use of Friedrich's conformal representation of spatial infinity in which we prove delicate non-degenerate Gr\"onwall estimates. We describe the relationship between the solution and the data both in terms of Friedrich's conformal coordinates and the usual physical coordinates on Minkowski space.
\end{abstract}

\maketitle



\setcounter{tocdepth}{1}
\tableofcontents

\section{Introduction}

This paper is a companion paper to \cite{TauVal23}, in which the second and third author studied the regularity at future null infinity of linear massless fields of spin $s > 0$ in a fixed Minkowski space background. In particular, the analogous study for the case of the wave equation has remained open, which we address here. A rough description of our result is provided in \Cref{sec:main_result} and a precise description later in \Cref{Subsection:MainResult}.

\medskip
Building on the founding works of Cotton, Schouten and Weyl \cite{Sch21,Cot99,Wey18}, in the early and mid-sixties Penrose introduced \cite{Pen63,Pen65a} what is now known as the \emph{conformal method} to tackle asymptotic questions in general relativity. The key idea was the observation that conformal transformations allow one to study the asymptotic aspects of massless fields by instead understanding the local behaviour of appropriately rescaled---unphysical---fields near a \emph{conformal boundary} representing physical infinity.
Penrose termed spacetimes admitting such a conformal compactification \emph{asymptotically simple}, a notion which has been instrumental in the development of the present understanding of gravitational radiation \cite{Ger76,PenRin86,Fri04}.

The introduction of a conformal boundary arises as a natural framework for the study of \emph{massless scattering\footnote{See e.g. \cite{Tau19,NicTau24,Nicolas2024} for a more complete literature review of the history of conformal formulations of scattering.} problems} \cite{LaxPhillips1964}, i.e. the development of massless fields with prescribed asymptotic initial data \cite{Friedlander1962,Friedlander1964,Friedlander1967}, and in particular the precise relationship  \cite{KadKeh25} between past and future asymptotic data. In the case of an asymptotically flat background, the conformal boundary---called \emph{null infinity} and denoted $\scri$---is a null hypersurface, and a scattering problem, phrased in the conformally rescaled spacetime, takes the form of a characteristic initial value problem with data on null infinity \cite{Penrose1980,Fri80}. An important question in such scattering problems, intimately related to the structure of gravitational radiation emitted by real physical objects, is to understand how initial data prescribed at past null infinity $\scri^-$ affects the structure of solutions at future null infinity $\scri^+$, {and in particular when \emph{peeling} holds, a property which can be seen as sufficiently high conformal regularity at, or of, null infinity.} The literature on this subject is very large and a complete list of citations is not possible; we mention here {the classical proof of Friedrich \cite{Fri86b} of the so-called semi-global stability of Minkowski space and the monumental proof by Christodoulou--Klainerman \cite{ChristodoulouKlainerman93} of the global nonlinear stability of Minkowski space, as well as the works of} Klainerman--Nicol\'o and Chru\'sciel--Delay \cite{KlainermanNicolo1999,KlainermanNicolo2003,ChruscielDelay2002,ChruscielDelay2003} {producing non-trivial classes of global\footnote{We note that while Friedrich's proof \cite{Fri86b} produces spacetimes close to Minkowski space with a smooth null infinity, they are not \emph{global} in the sense that the proof uses a hyperboloidal foliation which stays away from spatial infinity. The proof of Christodoulou--Klainerman \cite{ChristodoulouKlainerman93}, on the other hand, produces global spacetimes close to Minkowski space which do not necessarily exhibit peeling. The question of the existence of spacetimes with (an appropriately) smooth null infinity therefore had remained open until the constructions of \cite{KlainermanNicolo1999, KlainermanNicolo2003,ChruscielDelay2002,ChruscielDelay2003}.} spacetimes exhibiting the peeling property.}

{It has been recognised for some time that, though elegant, full conformal smoothness is insufficient to capture all physically interesting scenarios (see for instance the recent works of Kehrberger et al. \cite{Kehrberger2021a,Kehrberger2021b,Kehrberger2021c,KehrbergerGajic2022,Kehrberger2024,KehrbergerMasaood2024,KadKeh25}). Instead, polyhomogeneous expansions might appear to} provide a sufficiently general setting to describe the asymptotics of weak gravitational fields  \cite{Fri98b,ChristodoulouKlainerman93,LindbladRodnianski2010,HinVas17,Lindblad2017}. {On the other hand, a complete understanding of the conditions in the asymptotic past that \emph{do} result in peeling in the future, particularly in the absence of compactly supported data, is still absent, a question which we aim to study here.}


In \cite{Fri98a}, Friedrich set up a programme to investigate the {appearance} of polyhomogeneous asymptotics {in what we call the \emph{F-gauge}}. This approach has {led} to the identification of at least two distinct mechanisms which {generate} polyhomogeneous asymptotics, namely the nonlinearities in Einstein's equations, and the fact that, in {a certain} sense, spatial infinity is a corner of the conformal boundary of spacetime \cite{Fri98b,Val04a,Val04b,GasVal17b}. The latter survives at the linearised level, and may be seen in the usual conformal picture of Penrose as a consequence of the degeneracy of the generators of $\scri^+$ and $\scri^-$ as they intersect at $i^0$. A typical approach to study the structure of this degeneracy employs a geometric blow-up of $i^0$. Such a construction has appeared in the literature in at least three related incarnations going back to the work of Ashtekar--Hansen \cite{AshHan78}, Friedrich \cite{Fri98b} (cf. \cite{MagVal21}), and more recently of Hintz--Vasy \cite{HinVas17} and collaborators (cf. \cite{Mel95}). We use in this paper the conformal blow-up of Friedrich, the so-called cylinder\footnote{Friedrich's cylinder at spatial infinity, i.e. the specific conformal scale and the coordinates $(\tau, \rho)$ (cf. \Cref{sec:F_gauge}), is what we call the F-gauge.} at spatial infinity. {We note, in passing, that in the present case of the wave equation on a fixed Minkowski space background, similar results can probably be obtained by using instead the inverted double null coordinates $(U,V) = (u^{-1}, v^{-1})$ (where $(u,v)$ are the standard Minkowskian double null coordinates) and the Morawetz and scaling vector fields $\mathbf{K}$ and $\mathbf{S}$ (see \cref{footnote_Morawetz_scaling_vector_fields}), together with standard propagation of regularity results. Our method is subtly different, however, and produces differently weighted regularity requirements on the data: see in particular the function space given by \eqref{Sobolev_norm_past_null_infinity_transverse_derivative_physical_coords_first_part} and \eqref{Sobolev_norm_past_null_infinity_transverse_derivative_physical_coords_second_part}. Moreover, the advantage of our approach is that the F-gauge is more geometric (see below), and that the methods are expected to be generalizable to the nonlinear conformal Einstein field equations \cite{Fri03a}, which cannot be handled by a simple coordinate inversion.}

Constructed from conformal geodesics, the F-gauge is particularly useful in the formulation of a regular initial value problem in a neighbourhood of spatial infinity for the conformal Einstein field equations \cite{Fri98a}. In the case of the Minkowski spacetime, the F-gauge can be obtained by employing an explicit \emph{ad hoc} coordinate transformation and conformal rescaling, which we review briefly in \Cref{Background}. In this representation the topology of spatial infinity, now denoted $\mathcal{I}$, is blown up to\footnote{In fact, to avoid the usual coordinate singularity on the $2$-sphere, one also later lifts the topology to $(-1,1)\times \mathbb{S}^3$ via a Hopf fibration.} {$(-1,1) \times \mathbb{S}^2$}, with the so-called \emph{critical sets} at the endpoints $\{\pm 1\} \times \mathbb{S}^2$, where $\scri^\pm$ meets $\mathcal{I}$, denoted by $\mathcal{I}^\pm$. The intrinsic (inverse) metric on $\mathcal{I}$ induced by the conformally rescaled (inverse) metric is regular for $\tau \in (-1,1)$, but degenerates at $\tau = \pm 1$.

In this paper we show how to construct solutions to the asymptotic characteristic initial value problem for the wave equation in a neighbourhood of spatial infinity on Minkowski space that have suitably regular behaviour (in the conformal picture) towards both past and future null infinity.
Specifically, by writing the wave equation as a first order symmetric hyperbolic system, we adapt the methods in \cite{Fri03b} and \cite{TauVal23} to derive estimates  in a neighbourhood of $i^0$ reaching $\mathscr{I}^-$ and $\mathscr{I}^+$ which are insensitive to the degeneracies at $\mathcal{I}^\pm$. We emphasise that our entire analysis is carried out in a neighbourhood of spatial infinity, and that as a result here we treat a \emph{semi-global} scattering problem (like \cite{KadKeh25,TauVal23}, but unlike \cite{NicTau24} or \cite{Mas22,Mas24}, for example) in the sense that data is prescribed on a portion of past null infinity and a finite incoming null hypersurface, which we call $\underline{\mathcal{B}}_\varepsilon$. This is in contrast to the global problem in which data is prescribed on the whole of $\scri^-$. Due to the total characteristic nature of the cylinder $\mathcal{I}$, the wave equation reduces to a system of transport equations on $\mathcal{I}$.
Potential logarithmic divergences in the solutions then arise as a direct consequence of solving Jacobi differential equations on the cylinder. In turn, these explicit solutions allow us to identify conditions on the past radiation field (the freely specifiable data at past null infinity) which {eliminate} the logarithmic terms. We point out here that our results show an important distinction between the behaviour of the wave equation near $i^0$ and the behaviour of higher spin-$s$ fields: while both possess analogous hierarchical structures at spatial infinity, non-zero spin fields (and in particular linearised gravity) can in general exhibit asymmetric regularity at $\scri^-$ and $\scri^+$, which moreover becomes more pronounced for higher values of $s$. For the wave equation, there is exact symmetry in a certain sense (see \Cref{Appendix:InteriorEquations}).

We finally mention briefly that the problem of the smooth matching of solutions to the wave equation at spatial infinity is not only of relevance in the analysis of scattering problems but also for the identification of asymptotic BMS charges at past and future null infinity as discussed in, for example, \cite{FueHen24}.

\subsection{Main result}
\label{sec:main_result}

For conceptual clarity, we provide a rough statement of the main theorem of the paper below. A detailed formulation of our results can be found in \Cref{Subsection:MainResult}.

\begin{theorem}[Main result, rough version]
    {On a fixed Minkowski space background,} solutions to the wave equation with sufficiently regular {and appropriately decaying} asymptotic characteristic initial data at past null infinity $\mathscr{I}^-$ {and a short incoming null hypersurface $\underline{\mathcal{B}}_\varepsilon$} possess {correspondingly } regular asymptotic expansions in a neighbourhood of spatial infinity $i^0$, and in particular exhibit peeling {to arbitrary fixed order\footnote{Determined by the regularity of the asymptotic data.}} at future null infinity $\mathscr{I}^+$. 
\end{theorem}

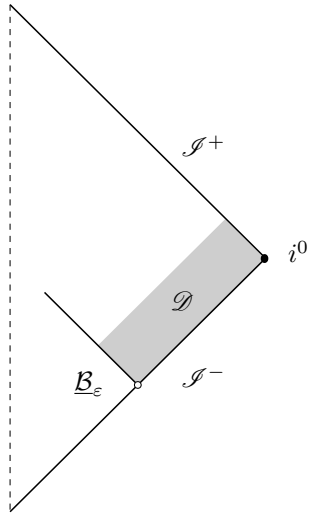
\begin{figure}[h]
\begin{center}
\includegraphics[width=0.475\textwidth]{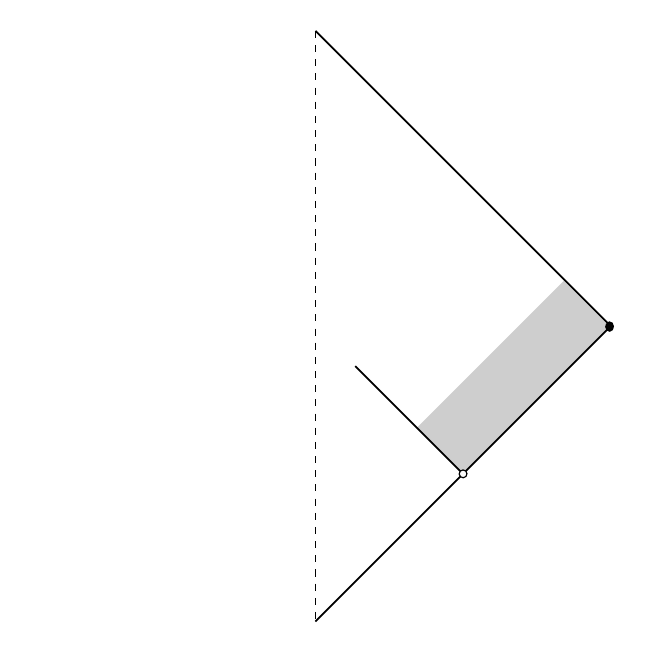}
\put(-50,143){$\mathscr{I}^+$}
\put(-50,56){$\mathscr{I}^-$}
\put(-10,103){$i^0$}
\put(-90,55){$\underline{\mathcal{B}}_\varepsilon$}
\put(-54,85){$\mathscr{D}$}
\end{center}
\caption{The Penrose diagram of the Minkowski spacetime showing the domain of existence $\mathscr{D}$ obtained in \Cref{thm:main_theorem_transverse_derivatives}.}
\label{fig:DOE}
\end{figure}

\begin{remark}
    {\em It is, of course, trivial to extend the domain of existence $\mathscr{D}$ to the entire Minkowski spacetime (provided the incoming null hypersurface $\underline{\mathcal{B}}_\varepsilon$ is extended to reach the axis of symmmetry $r = 0$, and therefore a complete Cauchy hypersurface can be put in) in the case of the linear wave equation considered in this paper. For the purpose of focussing the discussion on the region where our techniques are expected to apply (with appropriate modifications) to the full conformal Einstein equations, we nevertheless speak only of the (semiglobal) domain $\mathscr{D}$.}
\end{remark}

\begin{remark}
    \label{rmk:intro_regularity_vector_fields}
    {\em {Due to the conformal nature of our proof,} the regularity of our characteristic initial data is expressed in terms of Sobolev spaces in conformal coordinates (i.e. measuring regularity with respect to $\bmpartial_\tau$ and $\bmpartial_\rho$, see \Cref{Subsection:FGauge} for definitions), on $\scri^-$ {and $\underline{\mathcal{B}}_\varepsilon$.}   On the other hand, the regularity at future null infinity is obtained as an  asymptotic expansion\footnote{{Paradoxically, it will turn out that this asymptotic expansion is \emph{more} regular than the remainder, and in fact in our case will contain little information. See \Cref{rmk:main_theorem_remark_on_expansion}.}} {with} a remainder whose regularity, {in H\"older spaces}, is controlled quantitatively in terms of the regularity at $\scri^-$ {and $\underline{\mathcal{B}}_\varepsilon$}. In terms of physical {Bondi} coordinates {$(v=t+r, r)$}, the regularity of the data required on $\underline{\mathcal{B}}_\varepsilon$ sufficient to ensure that there exist solutions with the {stated behaviour}  is measured {essentially} with respect to $r^2 \bmpartial_r$. On $\scri^-$, the regularity is measured {essentially} with respect to $v^2 \bmpartial_v$. {For the reader interested in polyhomogeneity, we note that} polyhomogeneous data and asymptotic expansions with polyhomogeneous coefficients, on the other hand, {would} correspond to regularity measured with respect to $r \bmpartial_r$ and $v\bmpartial_v$, {which we make no attempt to control}.}
\end{remark}

\begin{remark}
    {\em As usual, the regularity, in conformal coordinates, referred to in \Cref{rmk:intro_regularity_vector_fields} can be translated in terms of physical coordinates, and implies a corresponding asymptotic behaviour in the physical spacetime.  It is immediate that under our regularity assumptions on the asymptotic characteristic data  the scalar field  satisfies the peeling behaviour, i.e. that the leading orders in its asymptotic expansion towards $\scri^+$ do not contain logarithms in $r$. This behaviour can be extracted to arbitrary order, provided sufficient regularity of the data in the past is demanded. Moreover, this is shown for data that is \emph{not} compactly supported. In particular, our method constructs solutions $\tilde{\phi}$ which for a fixed value of the retarded time {$u = t -r \ll -1$},  for large values of $r$, have an expansion  of the form
    \[ 
        \tilde{\phi}(u,r,\omega) = \frac{1}{r} \sum_{k = 0}^m \tilde{\phi}^{(k)}(u,\omega) \frac{1}{r^k} + o\left( r^{-(m+1)} \right)
    \]
    where the index $m$ measures the regularity in the asymptotic past, and the functions $\tilde{\phi}^{(k)}(u,\omega)$ have the regularity $C^{m-k,\alpha}(\scri^+)$ for some $\alpha > 0$. See \Cref{Subsection:Peeling} for slightly more details.
    }
\end{remark}

\subsection{Strategy of proof}
Our strategy here is similar to that of \cite{TauVal23}. For convenience, we briefly recall the outline of the argument here. The domain $\mathscr{D}$ shown in \Cref{fig:DOE}, in the ``Penrose gauge" corresponds to the grey area in \Cref{fig:FDOE}, in the ``Friedrich gauge". The domain in \Cref{fig:FDOE} has been split into two subdomains, the \emph{lower domain} $\mathcal{\underline{N}}_\varepsilon$ in light grey, and the \emph{upper domain} $\mathcal{N}_1$ in dark grey, which are separated by a spacelike hypersurface $\mathcal{S}_{-1+\varepsilon}$ terminating at the cylinder at
spatial infinity $\mathcal{I}$. \textbf{On the upper domain} $\mathcal{N}_1$ we { adapt the estimates of Friedrich \cite{Fri03b} all the way up to $\mathscr{I}^+$ in terms of the Cauchy data given on $\mathcal{S}_{-1+\varepsilon}$, {and subsequently construct an asymptotic expansion for the solution near $\mathcal{I}$ by repeated integration. The remainder in this expansion is then automatically controlled by the estimates on derivatives of the solution.} \textbf{On the lower domain} $\mathcal{\underline{N}}_\varepsilon$, we {similarly control derivatives of the solution in the boundary defining function of $\scri^-$ by deriving} new Gr\"onwall-type estimates for the wave equation which are uniform all the way down to $\scri^-$, similar to the ones for spin-$s$ fields in \cite{TauVal23}. {We then construct asymptotic expansions near $\scri^-$ as before, by repeated integration.} \textbf{Finally}, we stitch together the solutions on the lower and upper domains by ensuring that the lower solution has enough control at the spatial hypersurface $\mathcal{S}_{-1+\varepsilon}$ to be able to apply our estimates in the upper domain. We thus obtain a statement controlling the solution up
to future null infinity in terms of the asymptotic characteristic initial data on past null infinity {and $\underline{\mathcal{B}}_\varepsilon$}. 

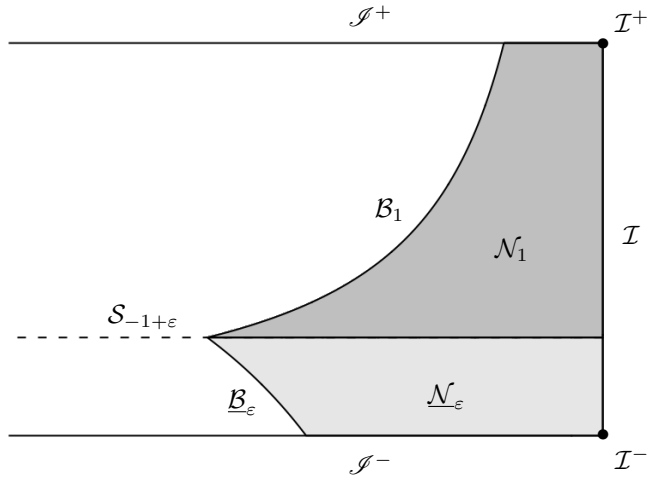
\begin{figure}[h]
\begin{center}
\includegraphics[width=0.507\textwidth]{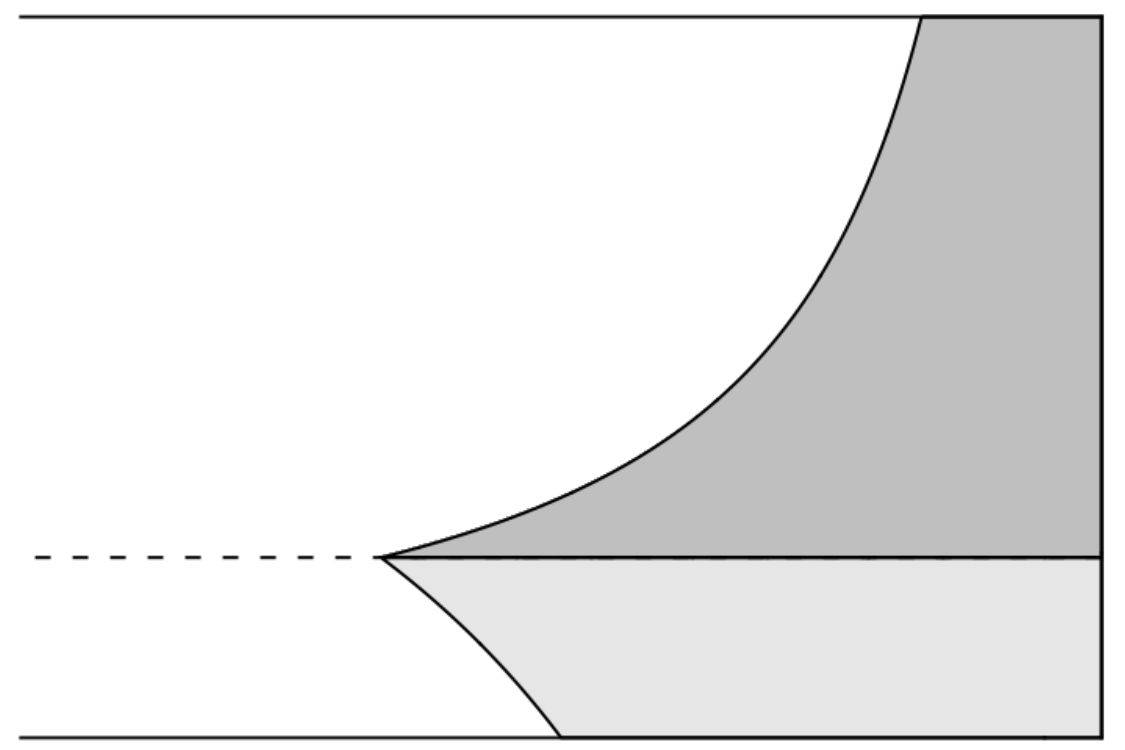}
\put(-100,157){$\mathscr{I}^+$}
\put(-100,-12){$\mathscr{I}^-$}
\put(3,75){$\mathcal{I}$}
\put(0,155){$\mathcal{I}^+$}
\put(-7,148){$\bullet$}
\put(0,-10){$\mathcal{I}^-$}
\put(-7,0.5){$\bullet$}
\put(-190,45){$\mathcal{S}_{-1+\varepsilon}$}
\put(-145,13){$\underline{\mathcal{B}}_\varepsilon$}
\put(-70,15){$\underline{\mathcal{N}}_\varepsilon$}
\put(-45,70){$\mathcal{N}_1$}
\put(-90,85){$\mathcal{B}_1$}

\end{center}
\caption{A depiction of the existence domain $\mathscr{D} = \underline{\mathcal{N}}_\varepsilon \cup \mathcal{N}_1$ in the F-gauge.}
\label{fig:FDOE}
\end{figure}

\subsection{Outline of the paper}
In \Cref{Background} we provide a brief summary of the geometric set-up used in the analysis: the F-gauge representation of spatial infinity on the Minkowski spacetime, and the reduction of the wave equation to a first order symmetric hyperbolic system using the space-spinor formalism. In Section \ref{Section3} we discuss the construction of estimates on the upper domain {$\mathcal{N}_1$} and the construction of asymptotic expansions near the cylinder at spatial infinity. \Cref{Section:EstimatesLowerDomain} {contains} the construction of estimates  in the lower domain {$\underline{\mathcal{N}}_\varepsilon$}. In \Cref{Section:AsymptoticCharactreistic} we {briefly discuss the construction of our past asymptotic initial data from a reduced set of freely specifiable data at $\scri^- \cup \underline{\mathcal{B}}_\varepsilon$}. In \Cref{Section:Stitching} we state {and prove} the main theorem {in full detail.  {The theorem is proven by combining the estimates in the upper and lower domains.} \Cref{Section:Interpretation} discusses the behaviour of solutions to the wave equation towards the critical sets {$\mathcal{I}^\pm$}, and frames our main results in the physical spacetime. The article has an appendix which contains two sections. \Cref{Appendix:FLemma} contains a technical lemma due to H. Friedrich used in the construction of estimates. \Cref{Appendix:InteriorEquations} gives {an}  analysis of the interior equations on the cylinder $\mathcal{I}$.

\subsection{Conventions and notation}

Our conventions follow those in \cite{CFEBook,TauVal23}. In particular, our metric signature is $(+,-,-,-)$, and the Riemann curvature tensor associated with the Levi-Civita connection of a metric $g_{ab}$ is defined by $\left[\nabla_{a}, \nabla_{b}\right]u^{d} = R\indices{^d_{cab}} u^{c}$. For a given spin dyad $\{o^A, \iota^A\}$ with $o_A\iota^A = 1$, we write $\epsilon\indices{_\bmA^A}, \bmA \in \{0,1\}$, to denote $\epsilon\indices{_0^A} = o^A$ and $\epsilon\indices{_1^A} = \iota^A$. Spinorial indices are raised and lowered using the antisymmetric $\epsilon$-spinor $\epsilon_{AB} = o_A \iota_B - \iota_A o_B$ (with inverse $\epsilon^{AB} = o^A \iota^B - \iota^A o^B$), e.g. $\xi_B = \xi^A \epsilon_{AB}$, using the convention that contracted indices should be “adjacent, descending to the right”. As usual, the spacetime metric $g_{ab}$ decomposes as $g_{ab} = \epsilon_{AB}\epsilon_{A'B'}$, where $\epsilon_{A'B'} = \overline{\epsilon_{AB}}$. The spin dyad $\epsilon\indices{_\bmA^A}$ gives rise to a tetrad of null vectors $\bme_{\bmA\bmA'} = e\indices{_{\bmA\bmA'}^{AA'}}\bmpartial_{AA'} = \epsilon\indices{_\bmA^A}\epsilon\indices{_{\bmA'}^{A'}}\bmpartial_{AA'}$.  When integrating over $\text{SU}(2)$, $\mu$ denotes the normalized Haar measure on $\text{SU}(2)$. Throughout the majority of the paper, we will be working on re-scaled Minkowski space (in the F-gauge), and we shall denote objects (fields, connections, spin dyad) on this spacetime plainly. When working on physical Minkowski space, we will denote objects with a tilde, e.g. $\Tilde{\phi}$. For a function space $X$ and functions $f$, $g$ (which may or may not be in $X$), we will use the notation $f = g + X$ to mean $f - g \in X$, i.e. that there exists $h \in X$ such that $f = g + h$.

\section{Preliminaries}\label{Background}

\subsection{Geometric setup}
\label{sec:F_gauge}

We start by briefly sketching the representation of spatial infinity to be considered and compare it to the standard Penrose compactification.

\subsubsection{Cylinder representation of spatial infinity: the F-gauge}
\label{Subsection:FGauge}

In the standard Penrose compactification of Minkowski space, the endpoint $i^0$ of all spacelike geodesics is a point. An alternative description, more suitable for our purposes, is obtained by blowing up $i^0$ to a cylinder. For $(t,r, \theta, \varphi)$ the standard polar coordinates on Minkowski space {and $\tilde{\bmeta}$ the standard metric on Minkowski space}, we define new coordinates $(\rho,\tau)$ in the exterior of the lightcone at the origin, which we denote by $\mathcal{N}$, by the relations
\begin{equation}
    \label{tau_rho_definitions}
    \rho = \frac{r}{r^2 - t^2} \quad \text{and} \quad \tau = \frac{t}{r}.
\end{equation}
Then defining the conformal factor
\begin{equation} \label{conformal_factor} 
\Theta \equiv \frac{1}{r} = \rho \left( 1 - \tau^2 \right), 
\end{equation}
we obtain an unphysical, rescaled, metric
\begin{align*}
    \bmeta & \equiv \Theta^2 \tilde{\bmeta} \\
    & = \mathbf{d} \tau\otimes \mathbf{d}\tau + \frac{\tau}{\rho} (\mathbf{d} \tau \otimes \mathbf{d} \rho + \mathbf{d} \rho \otimes \mathbf{d} \tau )- \frac{(1-\tau^2)}{\rho^2} \mathbf{d} \rho\otimes \mathbf{d} \rho - \bmsigma
\end{align*}
with inverse
\[ \bmeta^\sharp = (1-\tau^2) \bmpartial_\tau \otimes \bmpartial_\tau +  \tau \! \rho \, (\bmpartial_\tau \otimes \bmpartial_\rho+ \bmpartial_\rho \otimes \bmpartial_\tau) - \rho^2 \bmpartial_\rho \otimes \bmpartial_\rho - \bmsigma^\sharp,
\]
where $\bmsigma$ and $\bmsigma^\sharp$ are the standard covariant and inverse metric on the unit $2$-sphere, $\bmsigma = \d \theta \otimes \d\theta  + \sin^2 \! \theta \d \varphi\otimes \d\varphi$. Note here that $\tilde{\bmeta}$ is the physical metric on Minkowski {space . We call the coordinates $\tau$ and $\rho$ the \emph{F-coordinates} on the Minkowski spacetime. In terms of $\tau$ and $\rho$, the exterior region of the lightcone at the origin, $\mathcal{N}$, is given by
\[
\mathcal{N} =\left\{ (\tau, \rho) \times \mathbb{S}^2 \, | \, -1 <\tau < 1, \; \rho>0 \right\} \simeq (-1, 1)_\tau \times (0, \infty)_\rho \times \mathbb{S}^2
\]
and past and future null infinities ({intersected with } $\mathcal{N}$) are given by
\[
\mathscr{I}^\pm \equiv  \left\{ (\tau, \rho) \times \mathbb{S}^2 \, | \, \tau = \pm 1, \; \rho>0 \right\} \simeq (0, \infty)_\rho \times \mathbb{S}^2. 
\]
In these coordinates spatial infinity  is blown up to the set $i^0 \equiv \{ \rho = 0, \, \tau \in (-1, 1)\}$. We introduce the sets
\begin{align*}
&\mathcal{I} \equiv \left\{ (\tau, \rho) \times \mathbb{S}^2 \, | \,  |\tau| <1, \; \rho=0  \right\} \simeq (-1,1)_\tau \times \mathbb{S}^2, \\
&\mathcal{I}^\pm \equiv \left\{(\tau, \rho) \times \mathbb{S}^2 \, | \, \tau =\pm 1, \; \rho=0  \right\} \simeq \mathbb{S}^2,
\end{align*}
and call $\mathcal{I}$ the \emph{cylinder at spatial infinity} and $\mathcal{I}^\pm$ the \emph{critical sets}. Note that although $\mathcal{I}$ is at a finite $\rho$-coordinate, it is still at infinity with respect to the rescaled metric $\bmeta$ {(as well as, of course, the physical metric $\tilde{\bmeta}$)}.

We define the conformal completion \cite{Fri98a, Fri98b} of $\mathcal{N}$ to be the union
\begin{align*}
    \overline{\mathcal{N}} \equiv \mathcal{N} \cup \mathscr{I}^+\cup
\mathscr{I}^-\cup \mathcal{I} \cup \mathcal{I}^+ \cup \mathcal{I}^- \simeq [-1,1]_\tau \times [0, \infty)_\rho \times \mathbb{S}^2_\bmsigma.
\end{align*}
The vector fields $\bm\partial_\tau$ and $\bm\partial_\rho$ have norms with respect to the rescaled metric
\begin{align*}
    \bmeta(\bm\partial_\tau, \bm\partial_\tau)=1 \quad \text{and} \quad \bmeta(\bm\partial_\rho,\bm\partial_\rho) = - \frac{(1-\tau^2)}{\rho^2}, 
\end{align*}
so that $\bm\partial_\tau$ is uniformly timelike on all of $\overline{\mathcal{N}}$ with respect to $\bmeta$, and $\bm\partial_\rho$ is spacelike in the interior of $\overline{\mathcal{N}}$, null on $\scri$, and degenerates at $\mathcal{I}$. {Moreover, we have
\[ 
\bmeta(\bmpartial_\tau, \bmpartial_\rho) = \frac{\tau}{\rho}. 
\]
}
Note that the metric $\bmeta$ degenerates at $\rho = 0$, but the inverse metric $\bmeta^\sharp$ extends smoothly to all points of $\overline{\mathcal{N}}$. 

\subsubsection{Null frame near $\mathcal{I}$}
\label{Section:NullFrameNearI0}

In order to perform estimates, we will write the field equations in terms of a null frame 
\[
 \bme_{\bmA\bmA'} = \bme_{\bmA\bmA'}^\mu \bmpartial_\mu
\]
satisfying
\[
\bmeta(\bme_{\bmA\bmA'},\bme_{\bmB\bmB'}) =
\epsilon_{\bmA\bmB}\epsilon_{\bmA'\bmB'}.
\]
Specifically, we choose
\begin{equation}
\bme_{\bmzero\bmzero'} = \frac{1}{\sqrt{2}}\left(
(1-\tau)\bmpartial_\tau + \rho \bmpartial_\rho  \right), \qquad \bme_{\bmone\bmone'} =\frac{1}{\sqrt{2}}\left(
(1+\tau)\bmpartial_\tau - \rho \bmpartial_\rho  \right)
\label{FramesMinkowski}
\end{equation}
and complex vector fields $\bme_{\bmzero\bmone'}$ and $\bme_{\bmone\bmzero'}$ which are tangent to the spheres $\mathbb{S}^2_{\tau, \rho}$ of constant $\tau$ and $\rho$. {The vector field\footnote{In Newman--Penrose notation, $\bme_{\bmzero\bmzero'}$ would be called $n$, $\bme_{\bmone\bmone'}$ would be called $l$, and $(\bme_{\bmzero\bmone'}, \bme_{\bmone\bmzero'})$ would be called $(m,\bar{m})$.} $\bme_{\bmzero\bmzero'}$ is null, future-pointing and generates incoming surfaces like $\underline{\mathcal{B}}_\varepsilon$, and $\bme_{\bmone\bmone'}$ is similarly null, future-pointing, and generates outgoing surfaces like $\mathcal{B}_1$. In particular, $\bme_{\bmzero\bmzero'}|_{\scri^+}$ becomes a generator of $\scri^+$, and $\bme_{\bmone\bmone'}|_{\scri^-}$ becomes a generator of $\scri^-$.} To avoid the usual degeneracy of smooth vector fields on $\mathbb{S}^2$, {our main estimates will be carried out after lifting the 2-spheres to 3-spheres following the construction of Friedrich \cite[\S3]{Fri86a}, \cite[\S3]{Fri98a}.} For each pair $(\tau, \rho)$ we consider the Hopf map $\pi:\text{SU}(2)\cong \mathbb{S}^3 \mapsto \mathbb{S}^2 \cong\text{SU}(2)/\text{U}(1)$, which defines a principal $\text{U}(1)$-bundle over $\mathbb{S}^2$. By lifting each $2$-sphere to a copy of $\mathrm{SU}(2)$ (with coordinates denoted $t\indices{^{\bmA}_{\bmB}}$) along the Hopf map, we obtain a 5 dimensional manifold which we still denote 
\begin{align}
    \overline{\mathcal{N}} \equiv [-1,1]_\tau \times [0,\infty)_\rho \times \text{SU}(2)_{t\indices{^{\bmA}_{\bmB}}}.
    \label{Definition:OverlineN}
\end{align}
In this setting we now have
\[
\mathcal{I} = (-1,1)_\tau \times \mathrm{SU}(2), \quad
\mathcal{I}^\pm =
\mathrm{SU}(2), \quad \text {and} \quad \mathscr{I}^\pm = (0, \infty)_\rho \times
\mathrm{SU}(2)
\]
as subsets of $\overline{\mathcal{N}}$. To complete this construction, it remains to choose a frame of complex vector fields on $\mathrm{SU}(2)$.  We set
\[
\bme_{\bmzero\bmone'} = - \frac{1}{\sqrt{2}} \bmX_+ \quad \text{and} \quad
\bme_{\bmone\bmzero'} =- \frac{1}{\sqrt{2}} \bmX_-,
\]
where $\bmX_\pm$ are given in terms of two left-invariant vector fields $\bmZ_1, \bmZ_2$ on $\mathrm{SU}(2)$ by
\begin{align*}
    \bmX_+ \equiv -(\bmZ_2+i\bmZ_1), \qquad \bmX_- \equiv -(\bmZ_2-i\bmZ_1),
\end{align*}
the third left-invariant vector field $\bmZ_3$ chosen to point along the $\mathrm{U}(1)$-fibres. Specifically, the vector fields $\bmZ_i$ are generated by $\sqrt{-1}\sigma_i \in \mathfrak{su}(2)$, where $\sigma_i$ are the standard Pauli matrices, and obey the commutation relations $[\bmZ_i, \bmZ_j] = \textcolor{red}{2i} \epsilon_{ijk} \bmZ_k$. For sufficiently smooth (e.g. $C^1$) functions $f:\text{SU}(2)\to \mathbb{C}$, the vector fields $\bmX_+$ and $\bmX_-$ are complex conjugate in the sense that $\overline{\bmX_+f}=\bmX_-\overline{f}$, and the set
\[ \{ \bmX_+, \bmX_-, \bmZ_3 \}
\]
forms a basis (see Appendix A.3 of \cite{TauVal23} for details).

\begin{remark}
{\em {Throughout the paper we will freely switch between discussing the 5-dimensional manifold $\overline{\mathcal{N}}$ and the 4-dimensional manifold (also denoted by the same letter), depending on context.}  Recall that $\mathbb{S}^3\simeq \text{SU}(2)$ and that it admits a local trivialisation of the form $\mathbb{S}^2\times \mathrm{U}(1)$ with $\mathrm{U}(1)$ corresponding to the fibres. Spinor fields on $ \overline{\mathcal{N}}$ are obtained as the lift of spinor fields on the 4-dimensional base space in such a way that they transform homogeneously along fibres of $\mathbb{S}^3$---the resulting lifts are $\mathrm{U}(1)$-equivariant spinor fields, and in the case of scalar fields, they are $\mathrm{U}(1)$-invariant. The 5 and 4-dimensional points of view are related to each other by means of the Hopf map $\pi:\mathbb{S}^3\rightarrow \mathbb{S}^2$.  The Hopf map is a Riemannian submersion: its differential $\mathbf{d}\pi$ is surjective and its restriction to horizontal vectors is an isometry. Thus, given suitably smooth $f:\mathbb{S}^2\rightarrow \mathbb{C}$ it follows then that the norms
\[
\lVert f\circ\pi \rVert_{H^{k}(\mathbb{S}^3)}, \qquad 
\lVert  f\rVert_{H^k(\mathbb{S}^2)}
\]
are equivalent for a non-negative integer $k$, where $H^k(\mathbb{S}^2)$ and $H^k(\mathbb{S}^3)$ denote the standard $L^2$-Sobolev norms for complex valued functions over $\mathbb{S}^2$ and $\mathbb{S}^3$, respectively.}
\end{remark}

\subsection{Sobolev spaces}
We shall use Sobolev spaces on the domains and hypersurfaces referred to in \Cref{fig:FDOE} defined with respect to the conformal coordinates $(\tau, \rho, t^\bmA{}_\bmB)$. Specifically, on surfaces $\mathcal{S}_\tau$ of constant $\tau \in [-1,1]$, we define the $H^m$ norms
\[ \| f \|^2_{H^m(\mathcal{S}_\tau)} \equiv \int_{\mathcal{S}_\tau} \sum_{p' + |\alpha| \leq m} |\partial_\rho^{p'} \bmZ^\alpha f|^2 \, \d \rho \, \d \mu. \]
Observe that this in particular includes $\scri^\pm = \mathcal{S}_{\pm 1}$, since $\bmpartial_\rho$ is spacelike for $\tau \in (-1,1)$ and becomes a null generator of $\scri^\pm$ at $\tau = \pm 1$. On $\underline{\mathcal{B}}_\varepsilon$ and $\mathcal{B}_1$, we define Sobolev spaces with respect to the intrinsic null generators $\bme_{\bmzero \bmzero'}$ and $\bme_{\bmone \bmone'}$. For instance,
\[ \| f \|^2_{H^m(\underline{\mathcal{B}}_\varepsilon)} \equiv \int_{\underline{\mathcal{B}}_\varepsilon} \sum_{q' + |\alpha| \leq m} |(\bme_{\bmzero \bmzero'})^{q'} \bmZ^\alpha f|^2 \, \d \underline{\mathcal{B}}, \]
where $\d \underline{\mathcal{B}} = \bme_{\bmone \bmone'} \lrcorner \, \mathrm{dv}$ is the induced measure on $\underline{\mathcal{B}}_\varepsilon$, where $\mathrm{dv} = \d \tau \wedge \d \rho \wedge \d \mu$. On subdomains $\mathcal{N}$ of the full spacetime, we define
\[ \| f \|^2_{H^m(\mathcal{N})} \equiv \int_{\mathcal{N}} \sum_{p' + q' + |\alpha| \leq m} |\partial_\tau^{q'} \partial_\rho^{p'} \bmZ^\alpha f|^2 \, \mathrm{dv}. \]

\subsection{Translation between conformal and physical coordinates}
\label{Section:TransformationCoordinates}

Recall that the F-coordinates $(\tau, \rho)$ are given in terms of the standard physical coordinates $(t,r)$ by \eqref{tau_rho_definitions}. The standard double null coordinates $(u,v) \equiv (t-r, \, t+r)$ are then given in terms of $(\tau, \rho)$ by the expressions
\[
u = - \frac{1}{\rho(1+\tau)}, \quad v = \frac{1}{\rho(1-\tau)} \iff \tau = \frac{v+u}{v-u}, \quad \rho = \frac{u-v}{2uv}. 
\]
In order to translate our regularity assumptions on the past asymptotic initial data in terms of the more familiar physical coordinates, we shall use the Bondi coordinate $v = t + r$ (on past null infinity) and the inverted radial coordinate $R \equiv 1/r$. In terms of $(v,R)$, we have
\[ 
v = \frac{1}{\rho(1-\tau)}, \quad R = \rho (1-\tau^2) \iff \tau = -1 + Rv, \quad \rho = \frac{1}{v(2-Rv)}. 
\]
For coordinates near $\scri^+$, in terms of $(u,R)$, one similarly has
\[ u = \frac{-1}{\rho(1+\tau)}, \quad R = \rho(1-\tau^2) \iff \tau = 1+Ru, \quad \rho = \frac{-1}{u(2+Ru)}. \]
The vector fields $\bmpartial_\tau$ and $\bmpartial_\rho$ are then given in terms of $\bmpartial_v$ and $\bmpartial_R$\footnote{\label{footnote_Morawetz_scaling_vector_fields}{ We note that} $\bmpartial_\tau$ is, up to a factor of $\rho$, the conformal Killing (Morawetz) vector field $\bmK \equiv u^2\bmpartial_u+v^2\bmpartial_v$, i.e. $\bmpartial_\tau = \rho\bmK$. Similarly, $\bmpartial_\rho$ is almost exactly the scaling vector field $\bmS\equiv t\bmpartial_t+r\bmpartial_r$, $\rho\bmpartial_\rho = -\bmS$.}
\[
\bmpartial_\tau = \frac{v}{2-Rv}\bmpartial_v + \frac{2(1-Rv)}{v(2-Rv)}\bmpartial_R, \qquad \bmpartial_\rho = v(2-Rv)(-v\bmpartial_v + R \bmpartial_R).
\]
Analogously, using $(u,R)$ instead of $(v,R)$,
\[ \bmpartial_\tau = -\left( \frac{u}{2+Ru} \right) \bmpartial_u + \frac{2(1+Ru)}{u(2+Ru)} \bmpartial_R, \qquad \bmpartial_\rho = - u(2+Ru)(-u\bmpartial_u + R \bmpartial_R). \]

In particular, on $\scri^-$ one has that 
\begin{equation} \label{Friedrich_vectors_past_scri}
\bmpartial_\tau |_{\scri^-} = \frac{v}{2} \bmpartial_v + \frac{1}{v} \bmpartial_R , \qquad \bmpartial_\rho |_{\scri^-} = -2v^2 \bmpartial_v ,
\end{equation}
and on $\scri^+$
\begin{equation} \label{Friedrich_vectors_future_scri}
\bmpartial_\tau|_{\scri^+} = - \frac{u}{2} \bmpartial_u + \frac{1}{u} \bmpartial_R, \qquad \bmpartial_\rho|_{\scri^+} = 2u^2 \bmpartial_u.
\end{equation}
Moreover,
\begin{equation} \label{Friedrich_form_past_scri}
\mathbf{d} \rho |_{\scri^-} = - \frac{1}{2v^2} \mathbf{d} v.
\end{equation}
For future use we also record that
\begin{equation}
\bmpartial_v = \rho(1-\tau^2) \bmpartial_\tau + 2\rho^2\tau \bmpartial_\rho.
\label{Partial_v}
\end{equation}

\subsection{The wave equation near spatial infinity}
With reference to the conformal wave equation
\begin{align}\label{waveEqn}
    \square \phi + \frac{1}{6} R \phi = 0,
\end{align}
we note that the scalar curvature of \emph{both} $\bmeta$ and $\tilde{\bmeta}$ vanishes, $R[\bmeta] = 0 = R[\tilde{\bmeta}]$. The physical wave equation $\square_{\tilde{\bmeta}} \tilde{\phi} = 0$ is therefore equivalent to $\square \phi \equiv \square_\bmeta \phi = 0$, where $\phi = \Theta^{-1} \tilde{\phi} = r \tilde{\phi}$ {(cf. \eqref{conformal_factor})}. Thus, the wave equation $\Box_{\tilde{\bmeta}} \tilde{\phi} = 0$ in F-coordinates $(\tau, \rho, \theta, \varphi)$ is given by
\begin{align}\label{spatInfWaveEqn}
    \square \phi = (1-\tau^2)\partial_\tau^2 \phi +2\tau \! \rho \partial_\tau \partial_\rho \phi -\rho^2 \partial_\rho^2 \phi-2\tau\partial_\tau \phi-\slashed{\Delta}_{\mathbb{S}^2}\phi=0, 
\end{align}
where $\slashed{\Delta}_{\mathbb{S}^2}$ denotes the Laplacian on $\mathbb{S}^2$. This expression is clearly degenerate at $\tau = \pm 1$. In fact, this degeneracy is {only a consequence of the characteristic nature of} null infinity $\scri^\pm = \{ \tau = \pm 1, ~ \rho > 0\}$, but a true degeneracy in the hyperbolicity at the critical sets $\mathcal{I}^\pm = \{ \tau = \pm 1, ~ \rho = 0 \}$. Deriving estimates which are robust {to} this degeneracy is a key part of this work.

\subsubsection{First order reduction}

In order to adapt the methods of \cite{TauVal23} we consider a first-order hyperbolic reduction of the wave equation \eqref{spatInfWaveEqn}. This formulation follows the discussion in \cite[\S12]{CFEBook} and is elaborated below for the sake of completeness.

Starting from the wave equation in spinorial form
\begin{align}\label{spinWaveEqn}
    \square \phi\equiv\nabla^{A A'} \nabla_{A A'} \phi = 0,
\end{align}
define the auxiliary variable $\phi_{A A'} \equiv \nabla_{A A'} \phi$. By contracting this with the Hermitian spinor
\[ \tau^a = \tau^{A A'} = \sqrt{2} (\bmpartial_\tau)^a, \]
one obtains an evolution equation for $\phi$ reading
\begin{align}\label{defPsi}
    \mathcal{D} \phi = \psi,
\end{align}
where $\psi \equiv \tau^{AA'}\phi_{AA'} {= \sqrt{2} \partial_\tau \phi}$ and $\mathcal{D} \equiv \tau^{AA'}\nabla_{AA'} {=\sqrt{2} \nabla_{\partial_\tau}}$ (the ``Fermi derivative''). An equation for the auxiliary field $\phi_{AA'}$ is then given by the no-torsion condition
\begin{align*}
    \nabla_{AA'}\nabla_{BB'} \phi - \nabla_{BB'}\nabla_{AA'}\phi=0,
\end{align*}
which can be recast as $\nabla_{AA'}\phi_{BB'} - \nabla_{BB'}\phi_{AA'}=0$.  Contracting the primed indices and symmetrising results in
\begin{align*}
    \nabla\indices{_{(A}^{Q'}}\phi\indices{_{B)Q'}}=0,
\end{align*}
{where the contravariant index $Q'$ is understood to be out of the symmetrisation. This convention is used throughout the rest of the article.} This is {in fact} equivalent to the no-torsion property as a consequence of the Hermiticity of the 2-spinor $\phi_{AA'}$. Observing the identity
\begin{align*}
    \nabla\indices{_{A}^{Q'}}\phi\indices{_{BQ'}} \equiv \nabla\indices{_{(A}^{Q'}}\phi\indices{_{B)Q'}} -\frac{1}{2}\epsilon_{AB}\nabla\indices{^{QQ'}}\phi_{QQ'}
\end{align*}
one concludes, using \eqref{spinWaveEqn}, that
\begin{align}
    \label{nabla_derivative_phi_equation}
    \nabla\indices{_{A}^{Q'}}\phi\indices{_{BQ'}} = 0 .
\end{align}
Now define {the symmetric spinor}
\begin{equation}
    \label{psi_AB_definition}
    \psi_{AB} \equiv \tau\indices{_{(B}^{A'}}\phi_{A)A'}.
\end{equation}
A decomposition of the spinor $\phi_{AA'}$ into timelike and spacelike parts gives
\begin{align}\label{decomp}
    \phi_{AA'} = \frac{1}{2}\psi\tau_{AA'}-\tau\indices{^Q_{A'}}\psi_{AQ},
\end{align}
which can now be used to perform a space-spinor split{ting} of \eqref{nabla_derivative_phi_equation} \cite[\S4]{CFEBook}. This {amounts to ``space-spinorising'' the derivative operator $\nabla_{AA'}$ using $\tau^{BA'}$ and} yields, after a calculation taking into account the decomposition of the connection
\begin{align*}
    \nabla_{AB} = \frac{1}{2}\epsilon_{AB}\mathcal{D}+\mathcal{D}_{AB}
\end{align*}
into timelike and spacelike parts, the pair of equations
\begin{align*}
    &\mathcal{D}\psi +2 \mathcal{D}^{AB}\psi_{AB} = {\psi_{AB}\chi^{AB}-2\sqrt{2}\psi^{AB}\chi_{ACB}{}^C-\sqrt{2}\psi\chi_{AB}{}^{AB}}, \\
    &\mathcal{D}\psi_{AB} - \mathcal{D}_{AB}\psi+2\mathcal{D}\indices{_{(A}^{Q}}\psi\indices{_{B)Q}} = {\chi_{(A}{}^C \psi_{B)C}-\sqrt{2}\psi \chi_{(A}{}^C{}_{B)C}-2\sqrt{2} \chi_{(A}{}^C{}_{B)}{}^D\psi_{CD}},
\end{align*}
{where $\chi_{AB}$ and $\chi_{ABCD}$, the acceleration and  Weingarten spinors, encode information about the derivatives of $\tau^{AA'}$}. More precisely, one has that 
\[
{\chi_{AB} \equiv \frac{1}{\sqrt{2}}\tau_B{}^{A'} \mathcal{D}\tau_{AA'}, \qquad \chi_{ABCD} \equiv \frac{1}{\sqrt{2}}\tau_{D}{}^{C'}\mathcal{D}_{AB}\tau_{CC'}.}
\]
The derivative $\mathcal{D}_{AB}\equiv \tau\indices{_{(A}^{A'}}\nabla\indices{_{B)A'}}$ is the ``Sen connection''. These equations, together with (\ref{defPsi}), expressed in a specific adapted spin dyad (see below), give the system   
\begin{subequations}
\label{waveSHS}
\begin{align}
    &\mathcal{D}\phi=\psi, \label{waveSHS1}\\
    &\mathcal{D}\psi +2 \mathcal{D}^{\bmA\bmB}\psi_{\bmA\bmB} =  {\psi_{\bmA\bmB}\chi^{\bmA\bmB}-2\sqrt{2}\psi^{\bmA\bmB}\chi_{\bmA\bmC\bmB}{}^\bmC-\sqrt{2}\psi\chi_{\bmA\bmB}{}^{\bmA\bmB}}, \label{waveSHS2}\\
    &\frac{4}{(\bmA+\bmB)!(2-\bmA-\bmB)!}\left(\mathcal{D}\psi_{\bmA\bmB} -\mathcal{D}_{\bmA\bmB}\psi+2\mathcal{D}\indices{_{(\bmA}^{\bmQ}}\psi\indices{_{\bmB)\bmQ}}\right)  \nonumber \\
    & \hspace{2cm} = {\frac{4}{(\bmA+\bmB)!(2-\bmA-\bmB)!}\left(\chi_{(A}{}^C \psi_{B)C}-\sqrt{2}\psi \chi_{(A}{}^C{}_{B)C}-2\sqrt{2} \chi_{(A}{}^C{}_{B)}{}^D\psi_{CD}\right)}& \label{waveSHS3}
\end{align}
\end{subequations}
where $\bmA,\,\bmB$ take the values $0$ and $1$. {The system \eqref{waveSHS}
 can be checked to be symmetric hyperbolic by a direct computation.

From now on, let $\psi_k$, $k \in \{0, 1, 2\}$, denote the components of the spinor $\psi_{AB}$ with respect to a spinor dyad $\{ o^A,\, \iota^A \}$ adapted to $\sqrt{2} (\bmpartial_\tau)^a = \tau^{AA'}$. That is, we decompose
\[
\tau^{AA'} =o^A\bar{o}^{A'}+\iota^A \bar{\iota}^{A'},
\]
and denote the contractions of $\psi_{AB}$ with $\{ o^A, \iota^A \}$ by
\begin{align*}
    \psi_0 \equiv \psi_{AB}o^A o^B , \qquad
    \psi_1 \equiv \psi_{AB}\iota^A o^B, \qquad
    \psi_2 \equiv \psi_{AB}\iota^A \iota^B.
\end{align*}
Explicitly in terms of the vector fields $\{\partial_\tau, \partial_\rho, \bmX_\pm \}$,
\begin{equation}
    \label{phi_Components}
    \psi_0 =\frac{1}{\sqrt{2}}\bmX_-\phi, \qquad {\psi_1 =\frac{1}{\sqrt{2}} (\tau\partial_\tau\phi-\rho \partial_\rho \phi)}, \qquad \psi_2 = -\frac{1}{\sqrt{2}}\bmX_+\phi, 
\end{equation}

together with the earlier expression
\[ \psi = \sqrt{2} \partial_\tau \phi. \]

In the coordinates described in \Cref{Subsection:FGauge}, the system \eqref{waveSHS} takes the form
\begin{subequations}
\label{waveEqnDecomposed}
\begin{align}
    \psi &= \sqrt{2}\partial_\tau \phi, \label{waveEqnDecomposed1} \\
    E_0&\equiv (1-\tau) \partial_\tau\psi_0+\rho \partial_\rho \psi_0+\frac{1}{2} \bmX_+\psi-\bmX_+\psi_1 = 0, \label{waveEqnDecomposed2}\\
    E_2&\equiv (1+\tau) \partial_\tau\psi_2-\rho \partial_\rho \psi_2-\frac{1}{2} \bmX_-\psi-\bmX_-\psi_1 = 0, \label{waveEqnDecomposed3}\\
    C_a&\equiv \tau \partial_\tau\psi_1 - \rho \partial_\rho\psi_1 -\frac{1}{2}\partial_\tau \psi+\frac{1}{2}\bmX_+\psi_2-\frac{1}{2}\bmX_-\psi_0+ \psi_1 = 0, \label{waveEqnDecomposed4} \\
    C_b&\equiv \tau \partial_\tau\psi- \rho \partial_\rho\psi-2\partial_\tau\psi_1+ \bmX_+\psi_2+ \bmX_-\psi_0+\psi = 0,\label{waveEqnDecomposed5}
\end{align}
\end{subequations}
equations $E_0 = 0$ and $E_2 = 0$ being evolution equations (the indices {of $E_i$} agreeing with the indices of the components {$\psi_i$} which they transport), and $C_a = 0$ and $C_b = 0$ being constraints ({the indices $a$ and $b$ being mere labels, not tensor indices}). A hierarchical system of transport equations is then given by certain linear combinations of the system \eqref{waveEqnDecomposed}, namely by defining
\begin{align*}
    A_0 \equiv E_2, \qquad
    B_0 \equiv -C_a-\frac{1}{2}C_b, \qquad
    A_1 \equiv C_a-\frac{1}{2}C_b, \qquad
    B_1 \equiv E_0.
\end{align*}
Explicitly, we then have
\begin{subequations}
\label{transport_hierarchy}
\begin{align}
    A_0 &\equiv (1+\tau) \partial_\tau\psi_2-\rho \partial_\rho \psi_2-\frac{1}{2} \bmX_-\psi-\bmX_-\psi_1 = 0, \label{transport1} \\
    B_0 &\equiv (1-\tau)\partial_\tau\psi_1+\rho \partial_\rho\psi_1+\frac{1}{2}\left((1-\tau)\partial_\tau\psi+\rho \partial_\rho\psi\right)-\bmX_+\psi_2-\frac{1}{2}\psi-\psi_1 = 0, \label{transport2}\\
    A_1 &\equiv (1+\tau)\partial_\tau\psi_1-\rho \partial_\rho\psi_1-\frac{1}{2}\left((1+\tau)\partial_\tau\psi-\rho \partial_\rho\psi\right)-\bmX_-\psi_0-\frac{1}{2}\psi+\psi_1 = 0, \label{transport3} \\
    B_1 &\equiv (1-\tau) \partial_\tau\psi_0+\rho \partial_\rho \psi_0+\frac{1}{2} \bmX_+\psi-\bmX_+\psi_1 = 0.  \label{transport4}
\end{align}
\end{subequations}

\begin{remark}
    \label{rmk:linear_combinations_in_system}
   {\em Equations \eqref{transport_hierarchy} exhibit a hierarchical structure similar to that of the spin-$s$ equations, as in \cite{TauVal23}, and is needed in order to exploit the structure of the $(1\pm \tau)$ and $\rho$ factors in our estimates. A difference, however, is the coupling between $\psi$ and $\psi_1$. These components appear as a ``double" transport equation. In fact, \eqref{transport_hierarchy} may be better thought of as transport equations for the linear combinations $\mathbf{\Psi} = \psi_1 + \frac{1}{2} \psi$ and $\mathbf{\Phi} = \psi_1 - \frac{1}{2} \psi$, which will also play primary roles in the estimates. In terms of $\phi$, the combinations $\mathbf{\Psi}$ and $\mathbf{\Phi}$ are given by 
   \[
 {\mathbf{\Phi}= \bme_{\bmone\bmone'}\phi, \qquad \mathbf{\Psi}=-\bme_{\bmzero\bmzero'}\phi}.
   \]
   respectively. {Thus, the pair $(\mathbf{\Phi}, \mathbf{\Psi})$ corresponds to a pair of (linearly independent) null derivatives of $\phi$}}.
\end{remark}

\section{Estimates near $\scri^+$}\label{Section3}

In this section we show how to adapt the original construction of Friedrich's estimates \cite{Fri03b} for the system \eqref{transport_hierarchy} to control the solution to the wave equation in the upper domain $\mathcal{N}_1$ in \Cref{fig:FDOE}. Given $\varepsilon \in (0,1)$, $t \in (-1,1]$, $t>\tau_\star \in (-1,1)$ and $\rho_\star >0$, consider the following hypersurfaces:
\begin{align*}
    \mathcal{N}_t &\equiv \left\{ (\tau,\rho, t^\bmA{}_{\bmB}) \, | \, \tau_\star \leq \tau \leq t, ~ 0 \leq \rho \leq \left(\frac{2\varepsilon}{2-\varepsilon}\right) \frac{\rho_\star}{1+\tau}, ~ t^\bmA{}_\bmB \in \mathrm{SU}(2) \right\} \subset [-1,1]_\tau \times [0,\infty)_\rho \times \mathrm{SU}(2), \\
    \mathcal{S}_t &\equiv \left\{ (\tau,\rho, t^\bmA{}_{\bmB}) \, | \, \tau = t, ~ 0 \leq \rho \leq \left(\frac{2\varepsilon}{2-\varepsilon}\right) \frac{\rho_\star}{1+\tau}, ~ t^\bmA{}_\bmB \in \mathrm{SU}(2) \right\} , \\
    \mathcal{B}_t &\equiv \left\{ (\tau,\rho, t^\bmA{}_{\bmB}) \, | \, \tau_\star \leq \tau \leq t, ~ \rho = \left(\frac{2\varepsilon}{2-\varepsilon}\right) \frac{\rho_\star}{1+\tau}, ~ t^\bmA{}_\bmB \in \mathrm{SU}(2) \right\} , \\
    \mathcal{I}_t &\equiv \left\{ (\tau,\rho, t^\bmA{}_{\bmB}) \, | \, \tau_\star \leq \tau \leq t, ~ \rho = 0, ~ t^\bmA{}_\bmB \in \mathrm{SU}(2) \right\} ,
\end{align*}
and from now on denote by $\mathcal{S}_\star \equiv \mathcal{S}_{\tau_\star}$ the Cauchy hypersurface on which we consider initial data.

{
\begin{remark}
    {\em The factor $2\varepsilon / (2-\varepsilon )$ is chosen so that the geometry of the sets defined above matches up with \Cref{fig:FDOE} when we assume that the null hypersurface $\underline{\mathcal{B}}_\varepsilon$ intersects past null infinity at some fixed $\rho_\star > 0$.}
\end{remark}
}

 Further, given non-negative integers $p$, $p'$, $q$, $q'$ and a multi-index $\alpha \equiv (\alpha_1, \alpha_2, \alpha_3)$, we consider the differential operators 
\begin{align*}
    D \equiv D^{p,q,\alpha} \equiv \partial_{\rho}^{p} \partial_{\tau}^{q} \bmZ^{\alpha} \quad \text{and} \quad D' \equiv D^{p',q',\alpha} \equiv \partial_{\rho}^{p'} \partial_{\tau}^{q'} \bmZ^{\alpha},
\end{align*}
where $\bmZ^{\alpha}$ denotes the left-invariant vector fields $\{ \bmZ_1, \bmZ_2, \bmZ_3\}$ on $\mathrm{SU}(2)$ {(see \Cref{Appendix:FLemma})}. 
We then have the following \emph{a priori} estimates controlling the solutions to \eqref{transport_hierarchy} up to $\scri^+$.

\begin{proposition}\label{Proposition1}
    Let $\tau_\star \in (-1,1)$, $t>\tau_\star$, and consider the field $\psi$ and the components $\psi_k$, $k \in \{ 0, 1, 2 \}$, satisfying equations \eqref{transport_hierarchy} in the region $\mathcal{N}_t$. Let $(p,m) \in \mathbb{N} \times \mathbb{N}$ be such that
    \begin{align*}
        p > m + 1,
    \end{align*}
    and suppose that
    \begin{align}
    \label{Proposition1_hypothesis}
      {N_m(\psi, \psi_0, \psi_1, \psi_2) \equiv} \int_{\mathcal{S_\star}}\sum_{p'+q'+\lvert \alpha \rvert \leq m} \left( \sum_{k=0}^2 \lvert D'(\partial_{\rho}^{p}\psi_k) \rvert^2 + |D'(\partial_\rho^p \psi)|^2 \right) \d\rho \, \d\mu  < \infty.
    \end{align}
    Then there exists a constant $C_{p,m}>0$ which is independent of $t$ such that
    \begin{equation}
        \label{upperResult}
        \sum_{k=0}^2 \lVert \partial_{\rho}^{p} \psi_k \rVert_{H^m(\mathcal{N}_t)} ^2 + \lVert \partial_{\rho}^{p} \psi \rVert_{H^m(\mathcal{N}_t)} ^2 \leq C_{p,m} {N_m(\psi, \psi_0, \psi_1, \psi_2),}
    \end{equation}
    where the norm $H^m(\mathcal{N}_t)$ is given by
    \begin{align*}
        \lVert f \rVert_{H^m(\mathcal{N}_t)} ^2 \equiv \int_{\mathcal{N}_t} \sum_{p'+q'+\lvert \alpha \rvert \leq m} \lvert D' f  \rvert^2 \, \mathrm{dv}.
    \end{align*}
\end{proposition}

\begin{remark}
    {\em A crucial feature of estimate \eqref{upperResult} is the uniformity of the constant $C_{p,m}$ with respect to $t$. This allows us to push the estimate all the way to future null infinity in the limit $t \to 1$.}
\end{remark}

\begin{remark}
    \label{rmk:time_reversibility}
   {\em Note that the system \eqref{transport_hierarchy} is time-reversible: it is invariant under $\tau \mapsto - \tau$, $\psi \mapsto - \psi$, $\psi_1 \mapsto \psi_1$, $\psi_0 \mapsto \psi_2$, $\psi_2 \mapsto \psi_0$, $\bmX_+ \mapsto - \bmX_-$, and $\bmX_- \mapsto - \bmX_+$, which effects the transformation of the equations $A_0 \mapsto - B_1$, $B_1 \mapsto - A_0$, $B_0 \mapsto -A_1$, $A_1 \mapsto -B_0$. Using this fact, we deduce that \Cref{Proposition1} also holds for $t < \tau_*$, with the time-reversed domain $\underline{\mathcal{N}}_t$ being defined in the obvious way.}
\end{remark}

\begin{proof}
Applying the operator $D$ to each equation in \eqref{transport_hierarchy} and multiplying by $\overline{D\psi_{k}}$ and $\overline{D\psi}$ in an appropriate combination yields the following higher-order currents,
\begin{subequations}
\begin{align}
    \label{current0}
    0 &= 2\operatorname{Re}\left(\overline{D\psi_{2}} DA_{0}+\overline{D\psi_{1}} DB_{0}\right) + \operatorname{Re}\left(\overline{D\psi}DB_{0}\right), \\
    \label{current2} 
    0 &= 2\operatorname{Re}\left(\overline{D\psi_{1}} DA_{1}+\overline{D\psi_{0}} DB_{1}\right) - \operatorname{Re}\left(\overline{D\psi}DA_{1}\right).
\end{align}
\end{subequations}
When expanded, equation \eqref{current0} takes the form\footnote{This structure, for the spin-2 equations, was first observed in \cite{Fri03b}.}
\begin{align}\label{expandedEstimate0}
    0 &= \binom{\partial_\tau}{\partial_\rho}\cdot\binom{(1+\tau)\lvert D\psi_{2} \rvert ^2+(1-\tau) \lvert D\psi_{1} \rvert^2+\frac{1}{4}(1-\tau)\lvert D\psi \rvert^2+(1-\tau)\operatorname{Re}\left( \overline{D\psi} D\psi_1 \right)}{-\rho \lvert D\psi_{2} \vert^2+\rho \lvert D\psi_{1}\rvert^2+\frac{1}{4}\rho\lvert D\psi \rvert ^2+\rho\operatorname{Re}\left( \overline{D\psi} D\psi_1 \right)} \nonumber\\
    &-\bmZ^{\alpha}\bmX_{+}(\partial_{\rho}^{p}\partial_{\tau}^{q}\psi_{2})\bmZ^{\alpha}(\partial_{\rho}^{p}\partial_{\tau}^{q}\overline{\psi}_{1}) -\bmZ^{\alpha}(\partial_{\rho}^{p}\partial_{\tau}^{q}\psi_{2})\bmZ^{\alpha}\bmX_{+}(\partial_{\rho}^{p}\partial_{\tau}^{q}\overline{\psi}_{1}) \nonumber \\
    &-\bmZ^{\alpha}\bmX_{-}(\partial_{\rho}^{p}\partial_{\tau}^{q}\psi_{1})\bmZ^{\alpha}(\partial_{\rho}^{p}\partial_{\tau}^{q}\overline{\psi}_{2}) -\bmZ^{\alpha}(\partial_{\rho}^{p}\partial_{\tau}^{q}\psi_{1})\bmZ^{\alpha}\bmX_{-}(\partial_{\rho}^{p}\partial_{\tau}^{q}\overline{\psi}_{2}) \\
    &-\frac{1}{2}\left( \bmZ^{\alpha}\bmX_{+}(\partial_{\rho}^{p}\partial_{\tau}^{q}\psi_{2})\bmZ^{\alpha}(\partial_{\rho}^{p}\partial_{\tau}^{q}\overline{\psi}) +\bmZ^{\alpha}(\partial_{\rho}^{p}\partial_{\tau}^{q}\psi_{2})\bmZ^{\alpha}\bmX_{+}(\partial_{\rho}^{p}\partial_{\tau}^{q}\overline{\psi}) \right) \nonumber \\
    &-\frac{1}{2}\left( \bmZ^{\alpha}\bmX_{-}(\partial_{\rho}^{p}\partial_{\tau}^{q}\psi)\bmZ^{\alpha}(\partial_{\rho}^{p}\partial_{\tau}^{q}\overline{\psi}_{2}) +\bmZ^{\alpha}(\partial_{\rho}^{p}\partial_{\tau}^{q}\psi)\bmZ^{\alpha}\bmX_{-}(\partial_{\rho}^{p}\partial_{\tau}^{q}\overline{\psi}_{2}) \right) \nonumber \\
    &-2(p-q)\lvert D\psi_{2} \rvert ^2 +2(p-q-1)\lvert D\psi_{1} \rvert ^2 +\frac{1}{2}(p-q-1)\lvert D\psi \rvert ^2 +2(p-q-1)\operatorname{Re}\left( \overline{D\psi} D\psi_1 \right)\nonumber , 
\end{align}
while equation \eqref{current2} yields
\begin{align}\label{expandedEstimate1}
    0 &= \binom{\partial_\tau}{\partial_\rho}\cdot\binom{(1+\tau)\lvert D\psi_{1} \rvert ^2+(1-\tau) \lvert D\psi_{0} \rvert^2+\frac{1}{4}(1+\tau)\lvert D\psi \rvert^2-(1+\tau)\operatorname{Re}\left( \overline{D\psi} D\psi_1 \right)}{-\rho \lvert D\psi_{1} \vert^2+\rho \lvert D\psi_{0}\rvert^2-\frac{1}{4}\rho\lvert D\psi \rvert ^2+\rho\operatorname{Re}\left( \overline{D\psi} D\psi_1 \right)} \nonumber\\
    &-\bmZ^{\alpha}\bmX_{+}(\partial_{\rho}^{p}\partial_{\tau}^{q}\psi_{1})\bmZ^{\alpha}(\partial_{\rho}^{p}\partial_{\tau}^{q}\overline{\psi}_{0}) -\bmZ^{\alpha}(\partial_{\rho}^{p}\partial_{\tau}^{q}\psi_{1})\bmZ^{\alpha}\bmX_{+}(\partial_{\rho}^{p}\partial_{\tau}^{q}\overline{\psi}_{0}) \nonumber \\
    &-\bmZ^{\alpha}\bmX_{-}(\partial_{\rho}^{p}\partial_{\tau}^{q}\psi_{0})\bmZ^{\alpha}(\partial_{\rho}^{p}\partial_{\tau}^{q}\overline{\psi}_{1}) -\bmZ^{\alpha}(\partial_{\rho}^{p}\partial_{\tau}^{q}\psi_{0})\bmZ^{\alpha}\bmX_{-}(\partial_{\rho}^{p}\partial_{\tau}^{q}\overline{\psi}_{1}) \\
    &+\frac{1}{2}\left( \bmZ^{\alpha}\bmX_{+}(\partial_{\rho}^{p}\partial_{\tau}^{q}\psi)\bmZ^{\alpha}(\partial_{\rho}^{p}\partial_{\tau}^{q}\overline{\psi}_{0}) +\bmZ^{\alpha}(\partial_{\rho}^{p}\partial_{\tau}^{q}\psi)\bmZ^{\alpha}\bmX_{+}(\partial_{\rho}^{p}\partial_{\tau}^{q}\overline{\psi}_{0}) \right) \nonumber \\
    &+\frac{1}{2}\left( \bmZ^{\alpha}\bmX_{-}(\partial_{\rho}^{p}\partial_{\tau}^{q}\psi_{0})\bmZ^{\alpha}(\partial_{\rho}^{p}\partial_{\tau}^{q}\overline{\psi}) +\bmZ^{\alpha}(\partial_{\rho}^{p}\partial_{\tau}^{q}\psi_{0})\bmZ^{\alpha}\bmX_{-}(\partial_{\rho}^{p}\partial_{\tau}^{q}\overline{\psi}) \right) \nonumber \\
    &-2(p-q-1)\lvert D\psi_{1} \rvert ^2 +2(p-q)\lvert D\psi_{0} \rvert ^2 -\frac{1}{2}(p-q-1)\lvert D\psi \rvert ^2 +2(p-q-1)\operatorname{Re}\left( \overline{D\psi} D\psi_1 \right)\nonumber.
\end{align}
We now integrate equations \eqref{expandedEstimate0} and \eqref{expandedEstimate1} over $\mathcal{N}_t$ with respect to the volume element $\mathrm{dv}=\d\tau \wedge \d\rho \wedge \d\mu$, where $\d\mu$ is the Haar measure on $\text{SU}(2)$. Note at this stage that any terms involving the angular derivatives on $\text{SU}(2)$ are in the right form to apply Friedrich's technical lemma (\Cref{FTL}). This allows us to see that summing the integrated angular derivatives will give zero for any order $\lvert\alpha\rvert$. We will make use of this property shortly.
In light of this, for integrals of the angular derivatives, we will simply write angular($\alpha$). Hence, we have
\begin{align*}
    0 &= \int_{\mathcal{N}_t}\binom{\partial_\tau}{\partial_\rho}\cdot\binom{(1+\tau)\lvert D\psi_{2} \rvert ^2+(1-\tau) \lvert D\psi_{1} \rvert^2+\frac{1}{4}(1-\tau)\lvert D\psi \rvert^2+(1-\tau)\operatorname{Re}\left( \overline{D\psi} D\psi_1 \right)}{-\rho \lvert D\psi_{2} \vert^2+\rho \lvert D\psi_{1}\rvert^2+\frac{1}{4}\rho\lvert D\psi \rvert ^2+\rho\operatorname{Re}\left( \overline{D\psi} D\psi_1 \right)} \mathrm{dv}  \nonumber \\
    &-2(p-q)\int_{\mathcal{N}_t}\lvert D\psi_{2} \rvert ^2\mathrm{dv}  +2(p-q-1)\int_{\mathcal{N}_t}\lvert D\psi_{1} \rvert ^2\mathrm{dv}  \\
    &+\frac{1}{2}(p-q-1)\int_{\mathcal{N}_t}\lvert D\psi \rvert ^2\mathrm{dv}  +(p-q-1)\int_{\mathcal{N}_t}2\operatorname{Re}\left( \overline{D\psi} D\psi_1 \right)\mathrm{dv}  \nonumber\\
    &-\text{angular}(\alpha), \nonumber
\end{align*}
and
\begin{align*}
    0 &= \int_{\mathcal{N}_t}\binom{\partial_\tau}{\partial_\rho}\cdot\binom{(1+\tau)\lvert D\psi_{1} \rvert ^2+(1-\tau) \lvert D\psi_{0} \rvert^2+\frac{1}{4}(1+\tau)\lvert D\psi \rvert^2-(1+\tau)\operatorname{Re}\left( \overline{D\psi} D\psi_1 \right)}{-\rho \lvert D\psi_{1} \vert^2+\rho \lvert D\psi_{0}\rvert^2-\frac{1}{4}\rho\lvert D\psi \rvert ^2+\rho\operatorname{Re}\left( \overline{D\psi} D\psi_1 \right)}\mathrm{dv}  \nonumber \\
    &-2(p-q-1)\int_{\mathcal{N}_t}\lvert D\psi_{1} \rvert ^2\mathrm{dv}  +2(p-q)\int_{\mathcal{N}_t}\lvert D\psi_{0} \rvert ^2\mathrm{dv}  \\
    &-\frac{1}{2}(p-q-1)\int_{\mathcal{N}_t}\lvert D\psi \rvert ^2\mathrm{dv}  +(p-q-1)\int_{\mathcal{N}_t}2\operatorname{Re}\left( \overline{D\psi} D\psi_1 \right)\mathrm{dv}  \nonumber \\
    &-\text{angular}(\alpha). \nonumber
\end{align*}
Applying the Euclidean divergence theorem on the first integral in each equation yields
\begin{align}\label{intermediateEstimate0}
    0 &= \int_{\mathcal{S}_t} \left( (1+t)\lvert D\psi_{2} \rvert ^2 +(1-t)\lvert D\psi_{1} \rvert^2 +\frac{1}{4}(1-t)\lvert D\psi \rvert^2 +(1-t)\operatorname{Re}\left( \overline{D\psi} D\psi_1 \right) \right) \d\rho \, \d\mu \nonumber \\
    &- \int_{\mathcal{S}_*} \left( (1+\tau_*)\lvert D\psi_{2} \rvert ^2 +(1-\tau_*)\lvert D\psi_{1} \rvert^2 +\frac{1}{4}(1-\tau_*)\lvert D\psi \rvert^2 +(1-\tau_*)\operatorname{Re}\left( \overline{D\psi} D\psi_1 \right) \right) \d\rho \, \d\mu \nonumber \\
    &+ \int_{\mathcal{I}_t} \rho\left( - \lvert D\psi_{2} \vert^2+ \lvert D\psi_{1}\rvert^2+\frac{1}{4}\lvert D\psi \rvert ^2+\operatorname{Re}\left( \overline{D\psi} D\psi_1 \right) \right) \d\tau \, \d\mu \\
    &+ \int_{\mathcal{B}_t} \binom{(1+\tau)\lvert D\psi_{2} \rvert ^2+(1-\tau) \lvert D\psi_{1} \rvert^2+\frac{1}{4}(1-\tau)\lvert D\psi \rvert^2+(1-\tau)\operatorname{Re}\left( \overline{D\psi} D\psi_1 \right)}{-\rho \lvert D\psi_{2} \vert^2+\rho \lvert D\psi_{1}\rvert^2+\frac{1}{4}\rho\lvert D\psi \rvert ^2+\rho\operatorname{Re}\left( \overline{D\psi} D\psi_1 \right)}\cdot\nu\binom{\rho}{1+\tau} \d\mathcal{B} \nonumber \\
    &-2(p-q)\int_{\mathcal{N}_t}\lvert D\psi_{2} \rvert ^2\mathrm{dv}  +2(p-q-1)\int_{\mathcal{N}_t}\lvert D\psi_{1} \rvert ^2\mathrm{dv}  \nonumber \\
    &+\frac{1}{2}(p-q-1)\int_{\mathcal{N}_t}\lvert D\psi \rvert ^2\mathrm{dv}  +(p-q-1)\int_{\mathcal{N}_t}2\operatorname{Re}\left( \overline{D\psi} D\psi_1 \right)\mathrm{dv}  \nonumber \\
    &-\text{angular}(\alpha), \nonumber
\end{align}
and
\begin{align}\label{intermediateEstimate1}
    0 &= \int_{\mathcal{S}_t} \left( (1+t)\lvert D\psi_{1} \rvert ^2+(1-t) \lvert D\psi_{0} \rvert^2+\frac{1}{4}(1+t)\lvert D\psi \rvert^2-(1+t)\operatorname{Re}\left( \overline{D\psi} D\psi_1 \right) \right) \d\rho \, \d\mu \nonumber \\
    &- \int_{\mathcal{S}_*} \left( (1+\tau_*)\lvert D\psi_{1} \rvert ^2+(1-\tau_*) \lvert D\psi_{0} \rvert^2+\frac{1}{4}(1+\tau_*)\lvert D\psi \rvert^2-(1+\tau_*)\operatorname{Re}\left( \overline{D\psi} D\psi_1 \right) \right) \d\rho \, \d\mu \nonumber \\
    &+ \int_{\mathcal{I}_t} \rho\left( - \lvert D\psi_{1} \vert^2+ \lvert D\psi_{0}\rvert^2-\frac{1}{4}\lvert D\psi \rvert ^2+\operatorname{Re}\left( \overline{D\psi} D\psi_1 \right) \right) \d\tau \, \d\mu \\
    &+ \int_{\mathcal{B}_t} \binom{(1+\tau)\lvert D\psi_{1} \rvert ^2+(1-\tau) \lvert D\psi_{0} \rvert^2+\frac{1}{4}(1+\tau)\lvert D\psi \rvert^2-(1+\tau)\operatorname{Re}\left( \overline{D\psi} D\psi_1 \right)}{-\rho \lvert D\psi_{1} \vert^2+\rho \lvert D\psi_{0}\rvert^2-\frac{1}{4}\rho\lvert D\psi \rvert ^2+\rho\operatorname{Re}\left( \overline{D\psi} D\psi_1 \right)}\cdot\nu\binom{\rho}{1+\tau} \d\mathcal{B} \nonumber \\
    &-2(p-q-1)\int_{\mathcal{N}_t}\lvert D\psi_{1} \rvert ^2\mathrm{dv}  +2(p-q)\int_{\mathcal{N}_t}\lvert D\psi_{0} \rvert ^2\mathrm{dv}  \nonumber \\
    &-\frac{1}{2}(p-q-1)\int_{\mathcal{N}_t}\lvert D\psi \rvert ^2\mathrm{dv}  +(p-q-1)\int_{\mathcal{N}_t}2\operatorname{Re}\left( \overline{D\psi} D\psi_1 \right)\mathrm{dv}  \nonumber \\
    &-\text{angular}(\alpha), \nonumber
\end{align}
where $\d \mathcal{B}$ is the induced measure on $\mathcal{B}_t$ and $\nu\equiv(\rho^2+(1+\tau)^2)^{-\frac{1}{2}}$ is a normalisation factor for the outward unit normal to $\mathcal{B}_t$. In the above expressions, we have that both integrals on the cylinder $\mathcal{I}_t$ vanish due to the factor of $\rho$.

\subsection{Estimates for the $\psi_2$ component}

To obtain the estimate \eqref{upperResult} for $\psi_2$, we return to equation (\ref{intermediateEstimate0}) and observe that the integral over $\mathcal{B}_t$ simplifies. One has that 
\begin{align}\label{step2}
    &\int_{\mathcal{S}_t} \left( (1+t)\lvert D\psi_{2} \rvert ^2 +(1-t)\lvert D\psi_{1} \rvert^2 +\frac{1}{4}(1-t)\lvert D\psi \rvert^2 +(1-t)\operatorname{Re}\left( \overline{D\psi} D\psi_1 \right) \right) \d\rho \, \d\mu \nonumber \\
    &+ \int_{\mathcal{B}_t} 2\nu\rho\left( \lvert D\psi_{1} \rvert^2 +\frac{1}{4}\lvert D\psi \rvert^2 +\operatorname{Re}\left( \overline{D\psi} D\psi_1 \right) \right) \d\mathcal{B} \nonumber \\
    &+2(p-q-1)\int_{\mathcal{N}_t}\left(\lvert D\psi_{1} \rvert ^2   +\frac{1}{4}\lvert D\psi \rvert ^2\operatorname{Re}\left( \overline{D\psi} D\psi_1 \right) \right)\mathrm{dv}  \\
    &= \int_{\mathcal{S}_\star} \left( (1+\tau_\star)\lvert D\psi_{2} \rvert ^2 +(1-\tau_\star)\lvert D\psi_{1} \rvert^2 +\frac{1}{4}(1-\tau_\star)\lvert D\psi \rvert^2 +(1-\tau_\star)\operatorname{Re}\left( \overline{D\psi} D\psi_1 \right) \right) \d\rho \, \d\mu \nonumber \\
    &+2(p-q)\int_{\mathcal{N}_t}\lvert D\psi_{2} \rvert ^2 \, \mathrm{dv}  \nonumber \\
    &+\text{angular}(\alpha). \nonumber
\end{align}
Defining, as in \Cref{rmk:linear_combinations_in_system}), 
\begin{align}\label{PSI}
    \mathbf{\Psi} \equiv \psi_1 + \frac{1}{2} \psi,
\end{align}
so that $\lvert D \bm{\Psi} \rvert ^2 \equiv \lvert D\psi_{1} \rvert^2 +\frac{1}{4}\lvert D\psi \rvert^2 +\operatorname{Re}\left( \overline{D\psi} D\psi_1 \right) \geq 0$, we obtain the inequality
\begin{align}
    &\int_{\mathcal{S}_t} (1+t)\lvert D\psi_{2} \rvert ^2 \, \d\rho \, \d\mu +2(p-q-1)\int_{\mathcal{N}_t}\lvert D \bm{\Psi} \rvert ^2 \, \mathrm{dv}  \nonumber\\
    &\leq \int_{\mathcal{S}_\star} \left( (1+\tau_\star)\lvert D\psi_{2} \rvert ^2 + (1-\tau_\star)\lvert D \bm{\Psi} \rvert ^2\right) \d\rho \, \d\mu \label{stepAbove} \\
    & +2(p-q)\int_{\mathcal{N}_t}\lvert D\psi_{2} \rvert^2 \, \mathrm{dv} \nonumber\\
    & + \text{angular}(\alpha).\nonumber
\end{align}

We now perform a re-labelling of the indices
\[
p \longrightarrow p'+p, \quad q \longrightarrow q',
\]
so that the operator $D$ transforms as
\begin{align*}
    D \longrightarrow D'\partial_{\rho}^{p},
\end{align*}
where $D' \equiv D^{(p',q',\alpha)}$. Summing both sides of \eqref{stepAbove} over
$\mho\equiv \{ (p',q',\alpha) \, | \, p'+q'+\lvert\alpha\rvert \leq m \}$ and using \Cref{FTL} to see that $\sum_{\mho}\text{angular}(\alpha) = 0$, we obtain
\begin{align*}
    &\sum_{\mho}\int_{\mathcal{S}_t} (1+t)\lvert D'(\partial_{\rho}^{p}\psi_{2}) \rvert ^2 \, \d\rho \, \d\mu +2\sum_{\mho}(p'+p-q'-1)\int_{\mathcal{N}_t}\lvert D'(\partial_{\rho}^{p}\bm{\Psi}) \rvert^2 \, \mathrm{dv}  \\
    &\leq \sum_{\mho}\int_{\mathcal{S}_\star} \left( (1+\tau_\star)\lvert D'(\partial_{\rho}^{p}\psi_{2}) \rvert ^2 +(1-\tau_\star)\lvert D'(\partial_{\rho}^{p}\bm{\Psi}) \rvert ^2 \right) \d\rho \, \d\mu \\
    &+2\sum_{\mho}(p'+p-q')\int_{\mathcal{N}_t}\lvert D'(\partial_{\rho}^{p}\psi_{2}) \rvert^2 \, \mathrm{dv} .
\end{align*}
We now choose $p$ and $m$ such that
\begin{align*}
    p > m + 1,
\end{align*}
and commute $\Sigma_{\mho}$ through terms involving factors of $(p'-q')$ to obtain
\begin{align}\label{recoverPSI}
    (1+t) & \sum_{\mho}\int_{\mathcal{S}_t} \lvert D'(\partial_{\rho}^{p}\psi_{2}) \rvert^2 \, \d\rho \, \d\mu +2(p-m-1)\sum_{\mho}\int_{\mathcal{N}_t}\lvert D'(\partial_{\rho}^{p}\bm{\Psi}) \rvert^2 \, \mathrm{dv}  \nonumber \\
    &\leq (1+\tau_\star)\sum_{\mho}\int_{\mathcal{S}_*} \lvert D'(\partial_{\rho}^{p}\psi_{2}) \rvert ^2 \, \d\rho \, \d\mu +(1-\tau_*)\sum_{\mho}\int_{\mathcal{S}_\star} \lvert D'(\partial_{\rho}^{p}\bm{\Psi}) \rvert ^2 \, \d\rho \, \d\mu \\
    &+2(p+m+1)\sum_{\mho}\int_{\mathcal{N}_t}\lvert D'(\partial_{\rho}^{p}\psi_{2}) \rvert^2 \, \mathrm{dv} . \nonumber
\end{align}
The bulk integral on the left-hand side is then positive. Dropping this term gives
\begin{align*}
    \int_{\mathcal{S}_t}\sum_{\mho} \lvert D'(\partial_{\rho}^{p}\psi_{2}) \rvert ^2 \, \d\rho \, \d\mu & \lesssim \int_{\mathcal{S}_\star}\sum_{\mho} \left( \lvert D'(\partial_{\rho}^{p}\psi_{2}) \rvert ^2 +\lvert D'(\partial_{\rho}^{p}\bm{\Psi}) \rvert ^2\right) \d\rho \, \d\mu \\
    &+2(p+m+1) \int_{\tau_\star}^t \left( \int_{\mathcal{S}_\tau}\sum_{\mho}\lvert D'(\partial_{\rho}^{p}\psi_{2}) \rvert ^2 \d\rho \, \d\mu \right) \d\tau,
\end{align*}
to which we apply a Gr\"onwall argument to obtain
\begin{align*}
    \int_{\mathcal{S}_t}\sum_{\mho} \lvert D'(\partial_{\rho}^{p}\psi_{2}) \rvert ^2 \, \d\rho \, \d\mu \lesssim_{p,m} \int_{\mathcal{S}_\star}\sum_{\mho} \left( \lvert D'(\partial_{\rho}^{p}\psi_{2}) \rvert ^2+\lvert D'(\partial_{\rho}^{p}\bm{\Psi}) \rvert ^2 \right) \d\rho \, \d\mu.
\end{align*}
Upon integration, we are left with
\begin{align}\label{psi_2_spacetime_estimate_upper}
    \begin{split}
     \lVert \partial_{\rho}^{p} \psi_2 \rVert^2_{H^m(\mathcal{N}_t)}  & \lesssim_{p,m} \int_{\mathcal{S}_\star}\sum_{\mho} \left( \lvert D'(\partial_{\rho}^{p}\psi_{2}) \rvert ^2 + \lvert D'(\partial_{\rho}^{p}\bm{\Psi}) \rvert ^2\right) \d\rho \, \d\mu \\
     & \lesssim_{p,m} \int_{\mathcal{S}_\star} \sum_{\mho} \left( \sum_{i \in \{1,2\}} |D'(\partial_\rho^p \psi_i)|^2 + |D'(\partial_\rho^p \psi)|^2 \right) \d \rho \, \d \mu .
     \end{split}
\end{align}

\subsection{Estimates for the remaining components}
We now proceed to derive estimates for the remaining components. Returning to equation \eqref{intermediateEstimate1}, as before, the integral on the cylinder $\mathcal{I}_t$ vanishes due to the factor of $\rho$, while the integral over $\mathcal{B}_t$ simplifies to $\int_{\mathcal{B}_t} 2\nu\rho\lvert D\psi_{0} \rvert ^2\d\mathcal{B}$. As the latter is non-negative, we obtain the inequality
\begin{align*}
    \int_{\mathcal{S}_t} &\left( (1+t)\lvert D\psi_{1} \rvert ^2+(1-t) \lvert D\psi_{0} \rvert^2+\frac{1}{4}(1+t)\lvert D\psi \rvert^2-(1+t)\operatorname{Re}\left( \overline{D\psi} D\psi_1 \right) \right) \d\rho \, \d\mu \\
    &+2(p-q)\int_{\mathcal{N}_t}\lvert D\psi_{0} \rvert^2 \, \mathrm{dv}  +(p-q-1)\int_{\mathcal{N}_t}2\operatorname{Re}\left( \overline{D\psi} D\psi_1 \right)\mathrm{dv}  \\
    &\leq \int_{\mathcal{S}_\star} \left( (1+\tau_\star)\lvert D\psi_{1} \rvert ^2+(1-\tau_\star) \lvert D\psi_{0} \rvert^2+\frac{1}{4}(1+\tau_\star)\lvert D\psi \rvert^2-(1+\tau_\star)\operatorname{Re}\left( \overline{D\psi} D\psi_1 \right) \right) \d\rho \, \d\mu \\
    &+2(p-q-1)\int_{\mathcal{N}_t}\lvert D\psi_{1} \rvert^2 \, \mathrm{dv}  +\frac{1}{2}(p-q-1)\int_{\mathcal{N}_t}\lvert D\psi \rvert^2 \, \mathrm{dv}  \\
    &+\text{angular}(\alpha).
\end{align*}
Defining, {as in \Cref{rmk:linear_combinations_in_system}},
\[
\bm{\Phi} \equiv \psi_1 - \frac{1}{2}\psi
\]
so that $\lvert D \bm{\Phi} \rvert ^2 \equiv \lvert D\psi_{1} \rvert^2 +\frac{1}{4}\lvert D\psi \rvert^2 -\operatorname{Re}\left( \overline{D\psi} D\psi_1 \right) \geq 0$, we obtain
\begin{align}
\label{step1}
\begin{split}
    \int_{\mathcal{S}_t} & \left( (1+t)  \lvert D\bm{\Phi}\rvert^2 +(1-t) \lvert D\psi_{0} \rvert^2 \right) \d\rho \, \d\mu +2(p-q)\int_{\mathcal{N}_t}\lvert D\psi_{0} \rvert^2 \, \mathrm{dv}  \\
    & \leq \int_{\mathcal{S}_\star} \left((1+\tau_\star) \lvert D\bm{\Phi}\rvert ^2+(1-\tau_\star) \lvert D\psi_{0} \rvert^2 \right) \d\rho \, \d\mu +2(p-q-1)\int_{\mathcal{N}_t}\lvert D\bm{\Phi}\rvert^2 \, \mathrm{dv}  \\
    & +\text{angular}(\alpha).
\end{split}
\end{align}
Now, as before, we perform the shifts 
\[
p \longrightarrow p'+p, \qquad q \longrightarrow q',
\]
and sum both sides of inequality \eqref{step1} over $\mho$ to get
\begin{align*}
    \sum_{\mho}\int_{\mathcal{S}_t} & \left( (1+t)\lvert D'(\partial_{\rho}^{p}\bm{\Phi})\rvert ^2+(1-t) \lvert D'(\partial_{\rho}^{p}\psi_{0}) \rvert^2 \right) \d\rho \, \d\mu +2 \sum_{\mho}(p'+p-q')\int_{\mathcal{N}_t}\lvert D'(\partial_{\rho}^{p}\psi_{0}) \rvert ^2\mathrm{dv} \\
    &\leq  \sum_{\mho}\int_{\mathcal{S}_\star} \left((1+\tau_\star) \lvert D'(\partial_{\rho}^{p}\bm{\Phi})\rvert ^2+(1-\tau_\star) \lvert D'(\partial_{\rho}^{p}\psi_{0}) \rvert^2 \right) \d\rho \, \d\mu \\
    & + 2 \sum_{\mho}(p'+p-q'-1)\int_{\mathcal{N}_t}\lvert D'(\partial_{\rho}^{p}\bm{\Phi})\rvert ^2\mathrm{dv} . 
\end{align*}
Using that $p > m+1$, we further conclude that 
\begin{align}\label{forPsi2}
\begin{split}
    \sum_{\mho}\int_{\mathcal{S}_t} & \left( (1+t)\lvert D'(\partial_{\rho}^{p}\bm{\Phi})\rvert ^2+(1-t) \lvert D'(\partial_{\rho}^{p}\psi_{0}) \rvert^2 \right) \d\rho \, \d\mu \\
    & +2(p-m-1)\sum_{\mho}\int_{\mathcal{N}_t}\lvert D'(\partial_{\rho}^{p}\psi_{0}) \rvert^2 \, \mathrm{dv} \\
    & \leq  \sum_{\mho}\int_{\mathcal{S}_\star} \left((1+\tau_\star) \lvert D'(\partial_{\rho}^{p}\bm{\Phi})\rvert ^2+(1-\tau_\star) \lvert D'(\partial_{\rho}^{p}\psi_{0}) \rvert^2 \right) \d\rho \, \d\mu \\
    &+2(p+m) \sum_{\mho}\int_{\mathcal{N}_t}\lvert D'(\partial_{\rho}^{p}\bm{\Phi})\rvert^2 \, \mathrm{dv} . 
\end{split}
\end{align}
Dropping the positive bulk term on the left-hand side, we obtain 
\[
\int_{\mathcal{S}_t}\sum_{\mho} \lvert D'(\partial_{\rho}^{p}\bm{\Phi})\rvert^2  \, \d\rho \, \d\mu \lesssim \int_{\mathcal{S}_\star}\sum_{\mho} \lvert D'(\partial_{\rho}^{p}\bm{\Phi})\rvert^2 \, \d\rho \, \d\mu +3(p+m)\int_{\tau_\star}^{t}\left( \int_{\mathcal{S}_\tau}\sum_{\mho}\lvert D'(\partial_{\rho}^{p}\bm{\Phi})\rvert^2 \, \d\rho \, \d\mu \right) \! \d\tau.
\]
Using Gr\"onwall's inequality, we obtain for all $t \in [\tau_*, 1]$ the estimate 
\[ \int_{\mathcal{S}_t} \sum_{\mho} |D'(\partial_\rho^p \mathbf{\Phi})|^2 \, \d\rho \, \d \mu \lesssim_{p,m} \int_{\mathcal{S}_{\star}} \sum_{\mho} |D'(\partial_\rho^p \mathbf{\Phi})|^2 \, \d\rho \, \d \mu, \]
which gives upon integration
\begin{align}
\label{Phi_spacetime_estimate_top_domain}
\begin{split}
\| \partial_\rho^p \mathbf{\Phi} \|^2_{H^m(\mathcal{N}_t)} &\lesssim_{p,m} \int_{\mathcal{S}_{\star}} \sum_{\mho} |D'(\partial_\rho^p \mathbf{\Phi})|^2 \, \d\rho \, \d \mu \\
& \lesssim_{p,m} \int_{\mathcal{S}_\star} \sum_{\mho} \left(|D'(\partial_\rho^p \psi_1)|^2 + |D'(\partial_\rho^p \psi)|^2 \right) \d \rho \, \d \mu .
\end{split}
\end{align}
But also from \eqref{recoverPSI} we have, using \eqref{psi_2_spacetime_estimate_upper},
\begin{align}\label{Psi_spacetime_estimate_top_domain}
\begin{split}
    \lVert \partial_{\rho}^{p}\bm{\Psi} \rVert ^{2}_{H^{m}(\mathcal{N}_t)} &\lesssim_{p,m} \sum_{\mho}\int_{\mathcal{S}_*} \lvert  D'(\partial_{\rho}^{p}\psi_{2}) \rvert^2 \, \d\rho \, \d\mu +\sum_{\mho}\int_{\mathcal{S}_\star} \lvert D'(\partial_{\rho}^{p}\bm{\Psi}) \rvert^2 \, \d\rho \, \d\mu \\
    &+ \sum_{\mho}\int_{\mathcal{N}_t}\lvert D'(\partial_{\rho}^{p}\psi_{2}) \rvert ^2\mathrm{dv} \\
    & \lesssim_{p,m} \sum_{\mho}\int_{\mathcal{S}_*} \lvert  D'(\partial_{\rho}^{p}\psi_{2}) \rvert^2 \, \d\rho \, \d\mu +\sum_{\mho}\int_{\mathcal{S}_\star} \lvert D'(\partial_{\rho}^{p}\bm{\Psi}) \rvert^2 \, \d\rho \, \d\mu \\
    & \lesssim_{p,m} \int_{\mathcal{S}_\star} \sum_{\mho} \left( \sum_{i \in \{1,2\}}|D'(\partial_\rho^p \psi_i)|^2 + |D'(\partial_\rho^p \psi)|^2 \right) \d \rho \, \d \mu .
\end{split}
\end{align}
Combining the estimates for $\partial_\rho^p \mathbf{\Phi}$ and $\partial_\rho^p \mathbf{\Psi}$, estimates \eqref{Phi_spacetime_estimate_top_domain} and \eqref{Psi_spacetime_estimate_top_domain} above, we obtain
\begin{align}\label{psi_1_psi_spacetime_estimate_top_domain}
\begin{split}
    \lVert \partial_{\rho}^{p}\bm{\Phi} \rVert ^{2}_{H^{m}(\mathcal{N}_t)} + \lVert \partial_{\rho}^{p}\bm{\Psi} \rVert ^{2}_{H^{m}(\mathcal{N}_t)} &\equiv 2 \lVert \partial_{\rho}^{p}\psi_1 \rVert ^{2}_{H^{m}(\mathcal{N}_t)} + \frac{1}{2} \lVert \partial_{\rho}^{p}\psi \rVert ^{2}_{H^{m}(\mathcal{N}_t)} \\
    &\lesssim_{p,m} \int_{\mathcal{S}_\star} \sum_{\mho} \left( \sum_{i \in \{1,2\}}|D'(\partial_\rho^p \psi_i)|^2 + |D'(\partial_\rho^p \psi)|^2 \right) \d \rho \, \d \mu.
\end{split}
\end{align}
It remains to estimate $\partial_\rho^p \psi_0$. Using the fact that $p> m+1$, from \eqref{forPsi2} we obtain, using \eqref{Phi_spacetime_estimate_top_domain},
\begin{align}
    \label{psi_0_spacetime_estimate_top_domain}
    \begin{split}
    \| \partial_\rho^p \psi_0 \|^2_{H^m(\mathcal{N}_t)} & \lesssim_{p,m} \int_{\mathcal{S}_\star} \sum_{\mho} \left( \sum_{i\in \{0,1\}} |D'(\partial_\rho^p \psi_i )|^2 + |D'(\partial_\rho^p \psi)|^2 \right) \d \rho \, \d \mu.
    \end{split}
\end{align}
Combining estimates \eqref{psi_2_spacetime_estimate_upper}, \eqref{psi_1_psi_spacetime_estimate_top_domain} and \eqref{psi_0_spacetime_estimate_top_domain} completes the proof of \Cref{Proposition1}.
\end{proof}

\subsection{Asymptotic expansions near {$\mathcal{I}$}}
\label{section:asymptotic_expansions_near_future_infinity}

{Suppose the bound \eqref{Proposition1_hypothesis} is satisfied. Then} \Cref{Proposition1} controls the components $\psi$, $\psi_k$, $k \in \{0, 1, 2\}$, in $H^m(\mathcal{N}_t)$ uniformly in $t \in [\tau_\star, 1]$ in terms of {the norm $N_m$} of Cauchy initial data on $\mathcal{S}_\star$. {Moreover, using \Cref{rmk:time_reversibility}, we deduce the same control on $t \in [-1, \tau_\star]$, i.e. the components $\{\psi, \psi_0, \psi_1, \psi_2 \}$ are in fact in $H^m(\mathscr{D})$, where $\mathscr{D} \equiv \underline{\mathcal{N}}_\varepsilon \cup \mathcal{N}_1$.}

Since the boundary of the $5$-dimensional domain {$\mathscr{D}$} is Lipschitz, by Sobolev embedding we have, for $k \in \{ 0, 1, 2\}$ and $p>m+1$,
\begin{align*}
    \partial_\rho^p \psi, \ \partial_\rho^p \psi_k \in H^{m}({\mathscr{D}}) \hookrightarrow C^{r,\alpha}({\mathscr{D}}),
\end{align*}
for $r$ a positive integer and $\alpha \in (0,1)$ satisfying $r+\alpha \leq m - \frac{5}{2}$ and $m \geq 3$. Equivalently, $m \geq r+\alpha+\frac{5}{2}$. Restricting to $\alpha \leq \frac{1}{2}$, we have 
\begin{align*}
    \partial_\rho^p \psi, \ \partial_\rho^p \psi_k \in H^{r+3}({\mathscr{D}}) \hookrightarrow C^{r,\alpha}({\mathscr{D}})
\end{align*}
whenever $p > m+ 1 \geq r + 4$.

Now if $f$ is any of the components $\psi$, $\psi_k$, $k \in \{0,1,2\}$, by successively integrating in $\rho$ one obtains an expansion of the form
\begin{align}
    f = \sum_{p'=0}^{p-1}\frac{1}{p'!}f^{(p')}\rho^{p'} + J^p(\partial_{\rho}^pf),
\label{Expansionf}
\end{align}
where $J$ is the operator
\begin{align*}
    J: g \longmapsto J(g) = \int_{0}^{\rho}g\left(\tau,\rho',t\indices{^\bmA_\bmB}\right)\d\rho'.
\end{align*}
It follows that the remainder has the regularity
\begin{align*}
    J^p(\partial_{\rho}^pf) = f - \sum_{p'=0}^{p-1}\frac{1}{p'!}f^{(p')}\rho^{p'} \in C^{r,\alpha}({\mathscr{D}})
\end{align*}
for $p > r + 4$, $\alpha \leq \frac{1}{2}$, and any $r \in \mathbb{N}$. Moreover, since the cylinder $\mathcal{I} = \{ \rho = 0 \}$ is a total characteristic of the system \eqref{transport1}--\eqref{transport4}, the coefficients $f^{(p')}(\tau, \tensor{t}{^{\bmA}_{\bmB}}) = \partial_{\rho}^{p'}f \lvert_{\mathcal{I}}$ are in principle explicitly computable---{see \Cref{Appendix:InteriorEquations}}.

\section{Estimates near $\scri^-$}
\label{Section:EstimatesLowerDomain}

In this section we construct estimates on the lower domain $\mathcal{\underline{N}}_\varepsilon$ in terms of initial data on $\scri^-$. (see \Cref{fig:FDOE}). The local existence of the solution with  data on $\scri^-$ and $\underline{\mathcal{B}}_\varepsilon$ follows from an application of the theorem of Rendall \cite{Ren90} (cf. \cite{Luk12}), which can then be extended to all of $\underline{\mathcal{N}}_\varepsilon \cup \mathcal{N}_1$ using our estimates and a last slice argument; we refer to \S5.4 in \cite{TauVal23} for details.

Given $\rho_\star >0$, we consider the following hypersurfaces:
\begin{align*}
\scri^-_{\rho_\star} &\equiv \scri^- \cap \{ 0 \leq \rho \leq \rho_\star \} , \\
\underline{\mathcal{B}}_\varepsilon &\equiv  \left\{ (\tau, \rho, t^\bmA{}_{\bmB}) \, | \, -1 \leq \tau \leq -1 + \varepsilon, ~ \rho = \frac{2\rho_\star}{1-\tau}, ~ t^\bmA{}_{\bmB} \in \mathrm{SU}(2) \right\}, \\
\mathcal{S}_{-1+\varepsilon} &\equiv  \left\{ (\tau, \rho, t^\bmA{}_{\bmB}) \, | \, \tau = -1 + \varepsilon, ~ 0 \leq \rho \leq \frac{2\rho_\star}{2-\varepsilon}, ~ t^\bmA{}_\bmB \in \mathrm{SU}(2) \right\},
\\
\mathcal{I}_\varepsilon &\equiv \bigg\{ (\tau, \rho, t^\bmA{}_{\bmB}) \, | \, -1 \leq \tau \leq -1 + \varepsilon, ~ \rho = 0, ~ t^\bmA{}_\bmB \in \mathrm{SU}(2) \bigg\}.
\end{align*}
The set $\underline{\mathcal{B}}_\varepsilon$ is a short incoming null hypersurface intersecting $\scri^-$ at $\rho = \rho_\star$. As depicted in \Cref{fig:FDOE}, let $\underline{\mathcal{N}}_\varepsilon$ denote the spacetime slab bounded by the hypersurfaces $\underline{\mathcal{B}}_\varepsilon$, $\mathscr{I}^-_{\rho_\star}$, $\mathcal{S}_{-1+\varepsilon}$ and $\mathcal{I}_\varepsilon$. As in the previous section, given $m\geq 0$ we let
\[
\mho \equiv \{ (q', p', \alpha) \in \mathbb{N} \times \mathbb{N} \times \mathbb{N}^3 \, : \, q' + p' + |\alpha| \leq m \}.
\]

\subsection{Estimating $\psi_2$}

The component $\psi_2$ is the most troublesome to estimate. Here, in \Cref{Proposition2} we obtain higher-order estimates analogous to those in the previous section for the upper domain $\mathcal{N}_1$. These estimates require a bootstrap assumption on the remaining components $\psi_0$, $\psi_1$ and $\psi$, which we prove in \Cref{Proposition3}.

\begin{proposition}\label{Proposition2}
Let $0 < \epsilon \ll 1$. Suppose that for $m, \, q, \, p \in \mathbb{N}$ and $(q', p', \alpha) \in \mathbb{N} \times \mathbb{N} \times \mathbb{N}^3$ satisfying
\[ q' + p' + |\alpha| \leq m \quad \text{and} \quad m + p < q  \]
we have the following bounds on the characteristic initial data for $f\in\{\psi,\,\psi_0,\,\psi_1,\,\psi_2\}$:
\begin{equation}
    \label{initialDataBound}
    \sum_\mho \int_{\scri^-_{\rho_\star}} |D'( \partial_\tau^q \partial^p_\rho f)|^2 \, \mathrm{d} \rho \, \mathrm{d}\mu + \sum_\mho \int_{\underline{\mathcal{B}}_\varepsilon} |D'(\partial_\tau^q \partial^p_\rho f)|^2 \, \mathrm{d} \underline{\mathcal{B}} \leq \Omega_*
\end{equation}
for some $\Omega_* > 0$, where $\d \underline{\mathcal{B}}$ is the measure on $\underline{\mathcal{B}}_\varepsilon$ induced by $\mathrm{dv}$ {and $D' = \partial_\rho^{p'}\partial_\tau^{q'}\bmZ^\alpha$}. Additionally, assume the bootstrap bound
    \begin{align}\label{bootstrap}
        \lVert\partial_{\rho}^{p}\partial_\tau^{q}g \rVert^{2}_{H^{m}(\mathcal{\underline{N}}_\varepsilon)}  \leq \Omega_*
    \end{align}
for $g\in\{\psi,\,\psi_0,\,\psi_1  \}$. Then there exists a constant $C>0$ depending on $p$, $m$ and $q$ such that 
    \[
\lVert\partial_{\rho}^{p}\partial_\tau^{q}\psi_{2}\rVert^{2}_{H^{m}(\mathcal{\underline{N}}_\varepsilon)}  \leq C \Omega_* .
    \]
\end{proposition}

\begin{proof}
    As before, we consider the differential operators $D \equiv D^{p,q,\alpha} \equiv \partial_\rho^p \partial_\tau^q \bmZ^\alpha$
    and commute them into equations \eqref{transport1}--\eqref{transport4} to construct identity \eqref{expandedEstimate0}. Integrating \eqref{expandedEstimate0} over  $\mathcal{\underline{N}}_\varepsilon$ yields
    \begin{align*}
        0 &= \int_{\mathcal{\underline{N}}_\varepsilon} \binom{\partial_\tau}{\partial_\rho}\cdot\binom{(1+\tau)\lvert D\psi_{2} \rvert ^2+(1-\tau) \lvert D\psi_{1} \rvert^2+\frac{1}{4}(1-\tau)\lvert D\psi \rvert^2+(1-\tau)\operatorname{Re}\left( \overline{D\psi} D\psi_1 \right)}{-\rho \lvert D\psi_{2} \vert^2+\rho \lvert D\psi_{1}\rvert^2+\frac{1}{4}\rho\lvert D\psi \rvert ^2+\rho\operatorname{Re}\left( \overline{D\psi} D\psi_1 \right)} \, \mathrm{dv}  \\
        &-2(p-q)\int_{\mathcal{\underline{N}}_\varepsilon}\lvert D\psi_{2} \rvert^2 \, \mathrm{dv}  +2(p-q-1)\int_{\mathcal{\underline{N}}_\varepsilon}\lvert D\psi_{1} \rvert ^2\mathrm{dv}  \\
        &+\frac{1}{2}(p-q-1)\int_{\mathcal{\underline{N}}_\varepsilon}\lvert D\psi \rvert^2 \, \mathrm{dv}  +(p-q-1)\int_{\mathcal{\underline{N}}_\varepsilon}2\operatorname{Re}\left( \overline{D\psi} D\psi_1 \right)\mathrm{dv}  \\
        &+\text{angular}(\alpha),
    \end{align*}
    where we have again written the terms with angular dependence collectively. Using the divergence theorem, 
    \begin{align*}
        &\int_{\mathcal{S}_{-1+\varepsilon}} \left( \varepsilon \lvert D\psi_{2} \rvert ^2 +(2-\varepsilon)\lvert D\bm{\Psi} \rvert ^2  \right)\d\rho \, \d\mu-\int_{\scri^{-}_{\rho_*}} 2 \lvert D\bm{\Psi} \rvert ^2 \, \d\rho \, \d\mu -\int_{\mathcal{I}_{\varepsilon}} \rho\left( - \lvert D\psi_{2} \vert^2+ \lvert D\bm{\Psi} \rvert ^2 \right) \d\tau \, \d\mu \\
        &+\int_{\mathcal{\underline{B}}_\varepsilon} \binom{(1+\tau)\lvert D\psi_{2} \rvert ^2+(1-\tau) \lvert D\psi_{1} \rvert^2+\frac{1}{4}(1-\tau)\lvert D\psi \rvert^2+(1-\tau)\operatorname{Re}\left( \overline{D\psi} D\psi_1 \right)}{-\rho \lvert D\psi_{2} \vert^2+\rho \lvert D\psi_{1}\rvert^2+\frac{1}{4}\rho\lvert D\psi \rvert ^2+\rho\operatorname{Re}\left( \overline{D\psi} D\psi_1 \right)}\cdot\nu\binom{-\rho}{1-\tau} \d\mathcal{\underline{B}} \\
        &-2(p-q)\int_{\mathcal{\underline{N}}_\varepsilon}\lvert D\psi_{2} \rvert^2 \, \mathrm{dv}  +2(p-q-1)\int_{\mathcal{\underline{N}}_\varepsilon}\lvert D\psi_{1} \rvert^2 \, \mathrm{dv}  \\
        &+\frac{1}{2}(p-q-1)\int_{\mathcal{\underline{N}}_\varepsilon}\lvert D\psi \rvert^2 \, \mathrm{dv}  +(p-q-1)\int_{\mathcal{\underline{N}}_\varepsilon}2\operatorname{Re}\left( \overline{D\psi} D\psi_1 \right)\mathrm{dv} \\
        &+\text{angular}(\alpha) = 0,
    \end{align*}
 where $\nu \equiv (\rho^2+(1-\tau)^2)^{-\frac{1}{2}}$ is a normalisation factor for the outward pointing unit normal to $\mathcal{\underline{B}}_\varepsilon$ and $\d\mathcal{\underline{B}}$ is the induced measure on $\mathcal{\underline{B}}_\varepsilon$. Observe that the integral over $\mathcal{I}_{\varepsilon}$ vanishes due to the factor of $\rho$ and the integral over $\mathcal{\underline{B}}_\varepsilon$ simplifies to $\int_{\mathcal{\underline{B}}_\varepsilon} -2\nu\rho\lvert D\psi_{2} \rvert ^2\d\mathcal{\underline{B}}$. Hence
    \begin{align}\label{twoUses}
    \begin{split}
        \int_{\mathcal{S}_{-1+\varepsilon}} &\left( \varepsilon \lvert D\psi_{2} \rvert ^2 + (2-\varepsilon)\lvert D\bm{\Psi} \rvert ^2  \right) \d\rho \, \d\mu + \text{angular}(\alpha) \\
        &= \int_{\mathcal{\underline{B}}_\varepsilon} 2\nu\rho\lvert D\psi_{2} \rvert^2 \, \d\mathcal{\underline{B}} + \int_{\scri^{-}_{\rho_*}} 2 \lvert D\bm{\Psi} \rvert^2 \, \d\rho \, \d\mu \\
        &+2(p-q)\int_{\mathcal{\underline{N}}_\varepsilon}\lvert D\psi_{2} \rvert^2 \, \mathrm{dv}  +2(q-p+1)\int_{\mathcal{\underline{N}}_\varepsilon} \lvert D\bm{\Psi} \rvert^2 \, \mathrm{dv} . 
    \end{split}
    \end{align}
We now perform the relabelling
    \begin{align*}
        p \longrightarrow p' + p, \qquad q \longrightarrow q' + q,
    \end{align*}
    so that $D \longrightarrow D'\partial_\rho^{p} \partial_\tau^{q} \equiv \partial_\rho^{p'+p} \partial_\tau^{q'+q}\bmZ^\alpha$ and sum \eqref{twoUses} over $\mho$. The integrals over the angular terms vanish using \Cref{FTL}, and we obtain the estimate
    \begin{align*}
        \varepsilon\sum_{\mho} &\int_{\mathcal{S}_{-1+\varepsilon}} \lvert D'(\partial_\rho^{p} \partial_\tau^{q}\psi_{2}) \rvert^2 \, \d\rho \, \d\mu \leq 2\sum_{\mho}\int_{\mathcal{\underline{B}}_\varepsilon} \nu\rho\lvert D'(\partial_\rho^{p} \partial_\tau^{q}\psi_{2}) \rvert^2 \, \d\mathcal{\underline{B}} \\
        &+2\sum_{\mho}\int_{\scri^{-}_{\rho_*}} \lvert D'(\partial_\rho^{p} \partial_\tau^{q}\bm{\Psi}) \rvert^2 \, \d\rho \, \d\mu \\
        &+2\sum_{\mho}(p'+p-q'-q)\int_{\mathcal{\underline{N}}_\varepsilon}\lvert D'(\partial_\rho^{p} \partial_\tau^{q}\psi_{2}) \rvert^2 \, \mathrm{dv}  \\
        &+2\sum_{\mho}(q'+q-p'-p+1)\int_{\mathcal{\underline{N}}_\varepsilon} \lvert D'(\partial_\rho^{p} \partial_\tau^{q}\bm{\Psi}) \rvert^2 \, \mathrm{dv} ,
    \end{align*}
from where it follows, observing that $  m + p < q$ and using the bound \eqref{initialDataBound},  that 
    \begin{align*}
        \varepsilon\sum_{\mho} \int_{\mathcal{S}_{-1+\varepsilon}}& \lvert D'(\partial_\rho^{p} \partial_\tau^{q}\psi_{2}) \rvert^2 \, \d\rho \, \d\mu \\
        &\leq C^{(0)} \Omega_* + C^{(1)}_{p,m, q}\lVert\partial_{\rho}^{p}\partial_\tau^{q}\psi_{2}\rVert^{2}_{H^{m}(\mathcal{\underline{N}}_\varepsilon)} +2C^{(2)}_{p,m,q}\lVert\partial_{\rho}^{p}\partial_\tau^{q}\bm{\Psi}\rVert^{2}_{H^{m}(\mathcal{\underline{N}}_\varepsilon)}, 
    \end{align*}
    where $C^{(0)}$ is a numerical constant independent of $p$, $m$, $q$, and $\varepsilon$, $C^{(1)}_{p,m, q} \equiv 2( m-q+p) <0$, and $C^{(2)}_{p,m,q} = 2(q-p+1+m) > 2(2m+1) \geq 2$. Using the bootstrap bound (\ref{bootstrap}), this then yields the estimate
    \begin{align*}
        \varepsilon\sum_{\mho} \int_{\mathcal{S}_{-1+\varepsilon}}& \lvert D'(\partial_\rho^{p} \partial_\tau^{q}\psi_{2}) \rvert^2 \, \d\rho \, \d\mu \leq C  \Omega_* +C^{(1)}_{p,m, q}\lVert\partial_{\rho}^{p}\partial_\tau^{q}\psi_{2}\rVert^{2}_{H^{m}(\mathcal{\underline{N}}_\varepsilon)} ,
    \end{align*}
    where $C$ depends on $p$, $m$ and $q$, or equivalently
    \begin{align*}
        \varepsilon f'(\varepsilon) \leq C \Omega_* +C^{(1)}_{p,m, q}f(\varepsilon) ,
    \end{align*}
    where $f(\varepsilon)\equiv\lVert\partial_{\rho}^{p}\partial_\tau^{q}\psi_{2}\rVert^{2}_{H^{m}(\mathcal{\underline{N}}_\varepsilon)}$, since the hypersurfaces $\mathcal{S}_{-1+t}$ with $ t \in (0,\varepsilon)$ foliate the domain $\mathcal{\underline{N}}_\varepsilon$. Using an integrating factor, this inequality can be rewritten as
    \begin{align*}
        \frac{\d}{\d \varepsilon}(\varepsilon^{-CC^{(1)}_{p,m, q}} f(\varepsilon)) \leq C \Omega_* \varepsilon^{-C^{(1)}_{p,m,q}-1},
    \end{align*}
    which is integrable near $\varepsilon=0$ provided that $C^{(1)}_{p,m,q}<0$. This is true by the choice $q > m + p$, and upon integration, we obtain
    \begin{align*}
        f(\varepsilon) \leq \frac{\Omega_*}{-C^{(1)}_{p,m,q}} \lesssim_{p,m,q} \Omega_*.
    \end{align*}
\end{proof}

\subsection{The bootstrap bound}
To show that the bootstrap bound \eqref{bootstrap} holds, we make use of the identity \eqref{twoUses} derived above ---{note that \eqref{twoUses} does \emph{not} depend on the bootstrap bound}. Observe that it is sufficient to prove that the bootstrap bound holds up to some constant (possibly dependent on $p$, $m$ and $q$).
\begin{proposition}\label{Proposition3}
    Let $0 < \varepsilon < 1$. Suppose that for $m,p,q \in \mathbb{N}$ and $(p',q',\alpha) \in \mathbb{N} \times \mathbb{N} \times \mathbb{N}^3$ satisfying
    \begin{align*}
        p'+q'+ \lvert \alpha \rvert \leq m \quad \text{and} \quad m + p < q
    \end{align*}
    we have the bound \eqref{initialDataBound} on the characteristic initial data for the components $f \in \{ \psi, \psi_0, \psi_1, \psi_2\}$. Then there exists a constant $C>0$ depending on $p$, $m$, and $q$, but independent of $\varepsilon$, such that
    \begin{align*}
    \sum_{k \in \{0,1\}} \lVert\partial_{\rho}^{p}\partial_\tau^{q}\psi_{k}\rVert^{2}_{H^{m}(\mathcal{\underline{N}}_\varepsilon)}+ \lVert\partial_{\rho}^{p}\partial_\tau^{q}\psi\rVert^{2}_{H^{m}(\mathcal{\underline{N}}_\varepsilon)} \leq C\Omega_*.
    \end{align*}
\end{proposition}

\begin{proof}
    Performing the same relabelling $p \longrightarrow p'+p$ and $q \longrightarrow q'+q$ as in \Cref{Proposition2} and summing both sides of (\ref{twoUses}) over $\mho$ gives
    \begin{align*}
        \sum_{\mho}\int_{\mathcal{S}_{-1+\varepsilon}} \lvert D'(\partial_\rho^{p} \partial_\tau^{q}\bm{\Psi}) \rvert^2 \,\d\rho \, \d\mu &\leq C^{(0)}\Omega_* +C^{(1)}_{p,m,q}\lVert\partial_{\rho}^{p}\partial_\tau^{q}\psi_{2}\rVert^{2}_{H^{m}(\mathcal{\underline{N}}_\varepsilon)} \\
        &+C^{(2)}_{p,m,q}\sum_{\mho}\int_{\mathcal{\underline{N}}_\varepsilon} \lvert D'(\partial_\rho^{p} \partial_\tau^{q}\bm{\Psi}) \rvert^2 \,\mathrm{dv} ,
    \end{align*}
    where $C^{(0)}$ is a numerical constant, $C^{(1)}_{p,m,q} = 2( m-q+p )<0$, and $C^{(2)}_{p,m,q} = 2(q-p+1 +m)$. Writing
    \begin{align*}
        g(t) \equiv \sum_{\mho}\int_{\mathcal{S}_{-1+t}} \lvert D'(\partial_\rho^{p} \partial_\tau^{q}\bm{\Psi}) \rvert^2 \, \d\rho \, \d\mu,
    \end{align*}
    the inequality implies, dropping the negative term on the right-hand side, that 
    \begin{align*}
        g(t) \leq C^{(0)} \Omega_* +C^{(2)}_{p,m,q}\int_{0}^{t} g(s) \, \d s.
    \end{align*}
    Applying Gr\"onwall's lemma yields, since $ 0 \leq t \leq \varepsilon \leq 1$,
    \[ g(t) \leq C^{(0)} \Omega_* e^{C^{(2)}_{p,m,q}t} \lesssim_{p,m,q} \Omega_*, \]
    and integrating gives
    \begin{align} \label{towardsRecovery}
        \lVert\partial_{\rho}^{p}\partial_\tau^{q}\bm{\Psi}\rVert^{2}_{H^{m}(\mathcal{\underline{N}}_\varepsilon)} = \int_{0}^{\varepsilon} g(t) \, \d t \lesssim_{p,m,q} \Omega_* \varepsilon \lesssim_{p,m,q} \Omega_* .
    \end{align}

We now use this to show that $\| \partial_\rho^p \partial_\tau^q \psi_0 \|^2_{H^m(\underline{\mathcal{N}}_\varepsilon)}$ is also bounded by $\Omega_*$. This follows from a similar calculation using the current (\ref{current2}) involving $\psi_0$. Following similar steps as for the other components, one arrives at
    \begin{align*}
        &\int_{\mathcal{S}_{-1+\varepsilon}} \left( \varepsilon \lvert D\bm{\Phi} \rvert^2 +(2-\varepsilon)\lvert D\psi_0\rvert^2 \right) \d\rho \, \d\mu +\text{angular}(\alpha) \\
        &= \int_{\scri^{-}_{\rho_*}} 2\lvert D\psi_0\rvert^2 \, \d\rho \, \d\mu + \int_{\mathcal{\underline{B}}_\varepsilon} 2\nu\rho \lvert D\bm{\Phi} \rvert^2 \, \d\mathcal{\underline{B}} \\
        &+2(p-q-1)\int_{\mathcal{\underline{N}}_\varepsilon}\lvert D\bm{\Phi} \rvert^2 \, \mathrm{dv} +2(q-p)\int_{\mathcal{\underline{N}}_\varepsilon}\lvert D\psi_0\rvert^2 \, \mathrm{dv} ,
    \end{align*}
    where $\nu \equiv (\rho^2 +(1-\tau)^2)^{-\frac{1}{2}}$.  As with the other components, we shift $p \longrightarrow p' + p$, $q \longrightarrow q' + q$ and sum over $\mho$ to obtain the estimate
    \begin{align}\label{recoveryForComponents}
        \begin{split}
        (2-\varepsilon) & \sum_{\mho}\int_{\mathcal{S}_{-1+\varepsilon}} \lvert D'(\partial_\rho^{p} \partial_\tau^{q}\psi_0)\rvert^2 \, \d\rho \, \d\mu \\
        &\leq \sum_{\mho}\int_{\scri^{-}_{\rho_*}} 2\lvert D'(\partial_\rho^{p} \partial_\tau^{q}\psi_0)\rvert^2 \, \d\rho \, \d\mu + \sum_{\mho}\int_{\mathcal{\underline{B}}_\varepsilon} 2\nu\rho \lvert D'(\partial_\rho^{p} \partial_\tau^{q}\bm{\Phi}) \rvert^2 \, \d\mathcal{\underline{B}} \\
        &{+C^{(1)}_{p,m, q}}\sum_{\mho}\int_{\mathcal{\underline{N}}_\varepsilon}\lvert D'(\partial_\rho^{p} \partial_\tau^{q}\bm{\Phi}) \rvert^2 \, \mathrm{dv} {+C^{(2)}_{p,m, q}}\sum_{\mho}\int_{\mathcal{\underline{N}}_\varepsilon}\lvert D'(\partial_\rho^{p} \partial_\tau^{q}\psi_0)\rvert^2 \, \mathrm{dv} ,
    \end{split}
    \end{align}
where the constants $C^{(1)}_{p,m, q} < 0$ and $C^{(2)}_{p,m,q} > 1$ take the same form as in \Cref{Proposition2}. Hence, using the bound \eqref{initialDataBound} on the characteristic data and dropping the negative term on the right-hand side, we obtain
\begin{align*}  &\int_{\mathcal{S}_{-1+\varepsilon}} \sum_{\mho} \lvert D'(\partial_\rho^{p} \partial_\tau^{q}\psi_0)\rvert^2 \, \d\rho \, \d\mu \lesssim C^{(0)} \Omega_* + {C^{(2)}_{p,m,q}} \int_{0}^{\varepsilon}\left(\int_{\mathcal{S}_{-1+t}} \sum_{\mho} \lvert D'(\partial_\rho^{p} \partial_\tau^{q}\psi_0)\rvert^2 \, \d\rho \, \d\mu\right) \! \d t,
\end{align*}
where $C^{(0)}$ is a numerical constant independent of $p$, $q$, $m$, or $\varepsilon$. Gr\"onwall's lemma and integration in $t$ then gives
    \begin{equation}
    \label{psi_0_spacetime_bound_bootstrap}
        \lVert \partial_{\rho}^p\partial_{\tau}^q\psi_0 \rVert ^2_{H^{m}(\mathcal{\underline{N}}_\varepsilon)} \lesssim_{p,m,q} \Omega_*.
    \end{equation}
Returning to \eqref{recoveryForComponents}, we are now in a position to also estimate $\mathbf{\Phi}$. Moving the bulk $\mathbf{\Phi}$ term to the left-hand side and using \eqref{initialDataBound} and \eqref{psi_0_spacetime_bound_bootstrap}, we have
\begin{align*}
    \lVert \partial_{\rho}^p\partial_{\tau}^q\bm{\Phi} \rVert ^2_{H^{m}(\mathcal{\underline{N}}_\varepsilon)}
        &\leq \sum_{\mho}\int_{\scri^{-}_{\rho_*}} 2\lvert D'(\partial_\rho^{p} \partial_\tau^{q}\psi_0)\rvert^2 \, \d\rho \, \d\mu + \sum_{\mho}\int_{\mathcal{\underline{B}}_\varepsilon} 2\nu\rho \lvert D'(\partial_\rho^{p} \partial_\tau^{q}\bm{\Phi}) \rvert^2 \, \d\mathcal{\underline{B}} \\
        &+|C^{(1)}_{p,m,q}| \lVert \partial_{\rho}^p\partial_{\tau}^q\psi_0 \rVert ^2_{H^{m}(\mathcal{\underline{N}}_\varepsilon)} \\
        & \lesssim_{p,m,q} \Omega_*.
\end{align*}
Recalling that $\psi = \mathbf{\Psi} - \mathbf{\Phi}$ and $\psi_1 = \frac{1}{2} (\mathbf{\Psi} + \mathbf{\Phi})$, together with \eqref{towardsRecovery} and \eqref{psi_0_spacetime_bound_bootstrap}, this yields the claimed bootstrap estimate.
\end{proof}

\subsection{Asymptotic expansions near $\scri^-$}
\label{Section:expansions_near_past_infinity}

{Suppose the bound \eqref{initialDataBound} holds.} The estimates obtained in \Cref{Proposition2,Proposition3} control the solution to the wave equation in the domain $\underline{\mathcal{N}}_\varepsilon$, and in particular permit the construction of expansions analogous to those outlined in \Cref{section:asymptotic_expansions_near_future_infinity}. Near $\scri^+$, we saw that better regularity was obtained for higher $\rho$-derivatives, while near $\scri^-$, on the other hand, better regularity is obtained for higher $\tau$-derivatives. Indeed, setting $p=0$ under the assumptions of \Cref{Proposition2,Proposition3}, one observes that for $q > m + 3$, $k \in \{ 0, 1, 2\}$,
\begin{align*}
\partial_{\tau}^q\psi, \ \partial_{\tau}^q\psi_k \in H^{m+3}(\mathcal{\underline{N}}_\varepsilon) \hookrightarrow C^{m,\alpha}(\underline{\mathcal{N}}_\varepsilon)
\end{align*}
for $\alpha \in (0,\frac{1}{2}]$. Integrating $q$ times in $\tau$ {(shrinking the domain $\underline{\mathcal{N}}_\varepsilon$ slightly)}, each $f\in \{ \psi, \psi_0, \psi_1, \psi_2 \}$ then admits the expansion
\begin{align*}
    f = \sum_{q'=0}^{q-1}\frac{1}{q'!}(\partial_{\tau}^{q'}f)\rvert_{\scri^-}(\tau+1)^{q'} + I^{q}(\partial_{\tau}^{q}f),
\end{align*}
where $I^{q}(\partial_{\tau}^{q}f)\in C^{m,\alpha}(\mathcal{\underline{N}}_\varepsilon)$ and $I$ is the operator defined by
\begin{align*}
    f \longmapsto I(f) = \int_{-1}^{\tau}f(\tau',\rho,t\indices{^{\bmA}_{\bmB}}) \, \d\tau'.
\end{align*}

\section{Characteristic initial data}\label{Section:AsymptoticCharactreistic}

\subsection{{Minimal set of characteristic data}}

In this subsection we briefly review the construction of the initial data for the components $\psi$, $\psi_0$, $\psi_1$ and $\psi_2$ on $\mathscr{I}^-\cup \underline{\mathcal{B}}_\varepsilon$ from a reduced set of data. Recall that
\begin{align*} 
    \psi & = \sqrt{2} \partial_\tau \phi, \qquad \qquad \qquad \quad~~ \psi_0 = \frac{1}{\sqrt{2}}\bmX_- \phi, \\
    \psi_1 &= - \frac{1}{\sqrt{2}} (\tau \partial_\tau \phi + \rho \partial_\rho \phi), \qquad \psi_2 = -\frac{1}{\sqrt{2}} \bmX_+ \phi.
\end{align*}
{Suppose we are g}iven the past radiation field $\phi|_{\scri^-}$ {on $\scri^-$ and the data $\phi|_{\underline{\mathcal{B}}_\varepsilon}$ on $\underline{\mathcal{B}}_\varepsilon$.} {Then the components} $\psi_0|_{\scri^-}$ and $\psi_2|_{\scri^-}$, {as well as $\psi_0|_{\underline{\mathcal{B}}_\varepsilon}$ and $\psi_2|_{\underline{\mathcal{B}}_\varepsilon}$,} can be determined directly as $\bmX_\pm$ are tangent to the spheres. To compute $\psi|_{\scri^-}$ and $\psi_1|_{\scri^-}$, one needs to compute the time derivative $\partial_\tau \phi|_{\scri^-}$, {which is transverse to $\scri^-$}. Restricting the wave equation \eqref{spatInfWaveEqn} to $\scri^-$, one obtains
\begin{align*}
    -2\rho\partial_\rho \partial_\tau\phi|_{\scri^-} +2\partial_\tau \phi|_{\scri^-} -\rho^2\partial^2_\rho \phi |_{\scri^-}-\slashed{\Delta}_{\mathbb{S}^2}\phi|_{\scri^-} = 0,
\end{align*}
or equivalently
\begin{align*}
    \rho\partial_\rho \psi|_{\scri^-} -\psi|_{\scri^-} = -\frac{1}{\sqrt{2}}(\rho^2\partial^2_\rho \phi |_{\scri^-}+\slashed{\Delta}_{\mathbb{S}^2}\phi|_{\scri^-}).
\end{align*}
The right-hand side involves only tangential derivatives to $\scri^-$, so is directly computable from $\phi|_{\scri^-}$. This is therefore a transport equation along the generators of $\scri^-$ for the component $\psi|_{\scri^-}$, and can be solved for $\psi|_{\scri^-}$ given the initial value $\psi_\star$ at the cut $\mathcal{C}_\star = \{ \rho = \rho_\star \} \cap \mathcal{\scri^-} = \scri^- \cap \underline{\mathcal{B}}_\varepsilon$. {But $\psi_\star$ can be computed using the radiation fields $\phi|_{\scri^-}$ and $\phi|_{\underline{\mathcal{B}}_\varepsilon}$: from the knowledge of $\phi|_{\underline{\mathcal{B}}_\varepsilon}$, one can compute the derivative tangent to $\underline{\mathcal{B}}_\varepsilon$, $\sqrt{2} \bme_{\bmzero\bmzero'}(\phi)|_{\underline{\mathcal{B}}_\varepsilon}$, and therefore find $\sqrt{2} \bme_{\bmzero\bmzero'}(\phi)|_{\mathcal{C}_\star} = \sqrt{2} \psi|_{\mathcal{C}_\star} + \rho \partial_\rho \phi|_{\mathcal{C}_\star}$. The second term $\rho \partial_\rho \phi|_{\mathcal{C}_\star}$ can be computed by evaluating $\rho \partial_\rho \phi|_{\scri^-}$ at $\mathcal{C}_\star$ ($\partial_\rho$ being tangent to $\scri^-$).}

The data on the incoming null hypersurface $\underline{\mathcal{B}}_\varepsilon$ is constructed similarly:  
{one can rewrite the wave equation \eqref{spatInfWaveEqn} in terms of the double null frame $(\bme_{\bmzero \bmzero'}, \bme_{\bmone\bmone'}, \bmX_{\pm})$ as
\[ \Box \phi = 2 \bme_{\bmzero \bmzero'} (\bme_{\bmone\bmone'}(\phi)) - \sqrt{2} \bme_{\bmone \bmone'}(\phi) - \slashed{\Delta}_{\mathbb{S}^2} \phi = 0, \]
where $\bme_{\bmzero\bmzero'}$ is tangent to $\underline{\mathcal{B}}_\varepsilon$ and $\bme_{\bmone\bmone'}$ is transverse, giving a transport equation for the transverse derivative $\sqrt{2}\bme_{\bmone\bmone'}(\phi)|_{\underline{\mathcal{B}}_\varepsilon}$ along $\underline{\mathcal{B}}_\varepsilon$. To complete this construction, one needs to give data for $\sqrt{2}\bme_{\bmone\bmone'}(\phi)|_{\underline{\mathcal{B}}_\varepsilon}$ at $\mathcal{C}_\star$, i.e., using \eqref{FramesMinkowski}, the values $-\rho\partial_\rho \phi |_{\mathcal{C}_\star}$. As before, this can be computed using $\phi|_{\scri^-}$.
}

\medskip
We therefore have the following {standard} lemma:

\begin{lemma}
\label{reducedInitialData}
The initial data on $\scri^- \cup \underline{\mathcal{B}}_\varepsilon$ for the components $\psi$, $\psi_0$, $\psi_1$ and $\psi_2$ can be computed given the radiation field $\phi|_{\scri^-}$ and the data $\phi|_{\underline{\mathcal{B}}_\varepsilon}$.  We call the {pair} $(\phi|_{\scri^-}, \phi|_{\underline{\mathcal{B}}_\varepsilon})$ the \emph{reduced set of data}.
\end{lemma}


\subsection{Higher order transversal derivatives}

Using the equation and the reduced set of data, one can also recursively solve for all higher order transversal derivatives on both $\scri^-$ and $\underline{\mathcal{B}}_\epsilon$.

\subsubsection{Past null infinity}
Applying $\partial_\tau^{q-1}$ ($q \in \mathbb{N}$) to the wave equation \eqref{spatInfWaveEqn} and restricting to $\tau = -1$, one obtains
\begin{equation}
    \label{transversal_derivatives_equation_past_null_infinity}
    \rho \partial_\rho \varphi^{(q)} - q \varphi^{(q)} = S^{\scri^-}_{q-1}[\varphi^{(q-1)}],
\end{equation}
where $\varphi^{(q)} = (\partial_\tau^q \phi)|_{\scri^-}$ and
\begin{equation} 
    \label{S_l_operator_past_null_infinity}
    S^{\scri^-}_{\ell}[\varphi^{(\ell)}] = -\frac{1}{2} \rho^2 \partial_\rho^2 \varphi^{(\ell)} + \ell \rho \partial_\rho \varphi^{(\ell)} - \frac{1}{2} \slashed{\Delta}_{\mathbb{S}^2} \varphi^{(\ell)} - \frac{1}{2} (\ell + 1)\ell \varphi^{(\ell)}
\end{equation}
for $\ell \in \mathbb{N}_0$. The equation \eqref{transversal_derivatives_equation_past_null_infinity} is a transport equation for $\varphi^{(q)}$, with solution
\begin{equation} 
    \label{past_null_infinity_transverse_derivative_solution}
    \varphi^{(q)}(\rho, t^\bmA{}_\bmB) = \left( \frac{\rho}{\rho_\star} \right)^q \varphi^{(q)}(\rho_\star, t^\bmA{}_\bmB) + \rho^q\int_{\rho_\star}^\rho \sigma^{-q-1} S^{\scri^-}_{q-1}[\varphi^{(q-1)}](\sigma,t^\bmA{}_\bmB) \, \d \sigma.
\end{equation}
Since $S_{q-1}[\varphi^{(q-1)}]$ depends only on $\varphi^{(q-1)}$, the formula \eqref{past_null_infinity_transverse_derivative_solution} provides a recipe to iteratively compute all transverse derivatives $\varphi^{(q)}$, $q \in \mathbb{N}$, provided the data $\varphi^{(q)}(\rho_\star, t^\bmA{}_\bmB) = \partial_\tau^q \phi|_{\mathcal{C}_\star}$ are known for each $q \in \mathbb{N}$. This data can be computed using the data $\phi|_{\underline{\mathcal{B}}_\epsilon}$, by differentiating $q$ times along $\bme_{\bmzero \bmzero'}$, evaluating at $\tau = -1$, and using the values of $\partial_\rho$ derivatives of $\phi|_{\scri^-}$ at $\rho = \rho_\star$.

\subsubsection{The hypersurface $\underline{\mathcal{B}}_\varepsilon$}

Here we briefly explain how to compute all $\partial_\rho$ derivatives of $\phi|_{\underline{\mathcal{B}}_\varepsilon}$. Since $\bme_{\bmzero\bmzero'} = 2^{-1/2} ((1-\tau)\partial_\tau + \rho \partial_\rho)$ is tangent to $\underline{\mathcal{B}}_\varepsilon$, computing $(\bme_{\bmzero\bmzero'})^p \phi|_{\underline{\mathcal{B}}_\varepsilon}$ for any $p \in \mathbb{N}$ is immediate. Now repeatedly applying $\bme_{\bmone\bmone'}$ to \eqref{spatInfWaveEqn} and evaluating at $\rho = 2\rho_\star /(1-\tau)$ (on $\underline{\mathcal{B}}_\varepsilon$) gives
\begin{equation}
    \label{transverse_derivative_equation_incoming_hypersurface}
    (2 \bme_{\bmzero\bmzero'} - p \sqrt{2}) \psi^{(p)} = S^{\underline{\mathcal{B}}_\varepsilon}_{p-1}[\psi^{(p-1)}],
\end{equation}
where $\psi^{(p)} = ((\bme_{\bmone\bmone'})^p\phi)|_{\underline{\mathcal{B}}_\varepsilon}$ and 
\[ S^{\underline{\mathcal{B}}_\varepsilon}_{\ell}[\psi^{(\ell)} ] = \slashed{\Delta}_{\mathbb{S}^2} \psi^{(\ell)} - 2 \sum_{k=0}^{\ell-1} \left( \begin{array}{c} \ell \\ k \end{array} \right) (-2^{-1/2})^{\ell-k} \bme_{\bmzero \bmzero'} \psi^{(k+1)} - 2 \sum_{k=0}^{\ell-2} \left( \begin{array}{c} \ell \\ k \end{array} \right) (-2^{-1/2})^{\ell-k} \psi^{(k+2)}. 
\]

The equation \eqref{transverse_derivative_equation_incoming_hypersurface} is a transport equation for $\psi^{(p)}$ along $\underline{\mathcal{B}}_\varepsilon$, and can be solved (iteratively) given initial data at the cut $\mathcal{C}_\star$; this can, as before, be computed by differentiating the radiation field along $\scri^-$ and differentiating $\phi_{\underline{\mathcal{B}}_\varepsilon}$ along $\underline{\mathcal{B}}_\varepsilon$, and combining these derivatives appropriately. Given now the knowledge of the derivatives $\psi^{(p)} = ((\bme_{\bmone\bmone'})^p \phi)|_{\underline{\mathcal{B}}_\varepsilon}$ and $(\bme_{\bmzero\bmzero'})^p \phi_{\underline{\mathcal{B}}_\varepsilon}$ for any $p\in \mathbb{N}$, the transverse derivatives $(\partial_\rho^p \phi)|_{\underline{\mathcal{B}}_\varepsilon}$ can be obtained from appropriate linear combinations of the former.

\section{Controlling the solution from $\scri^-$ to $\scri^+$}
\label{Section:Stitching}

We now turn to combining the estimates in the lower and upper domains $\underline{\mathcal{N}}_\varepsilon$ and $\mathcal{N}_1$ in order to control the solution near $\scri^+$ in terms of asymptotic characteristic data on $\scri^- {\cup \underline{\mathcal{B}}_\varepsilon}$. Below we use the subscript $-$ to indicate that a particular quantity is associated to the lower domain, and the subscript $+$ for quantities associated to the upper domain. We denote by $f$ any element of $\{ \psi, \psi_0, \psi_1, \psi_2 \}$. We have obtained the following:

\begin{enumerate}[(i)]
    \item Control of the solution in H\"older spaces on $\underline{\mathcal{N}}_\varepsilon$ in terms of the characteristic data in Sobolev spaces on $\scri^-$ {on $\underline{\mathcal{B}}_\varepsilon$} (\Cref{Proposition2,Proposition3}, \Cref{Section:expansions_near_past_infinity}). By the continuity of the solution obtained, we therefore have quantitative control {of the solution in the same H\"older space} on $\mathcal{S}_{-1+\varepsilon}$, by restriction.
    \item Control of H\"older norms, and in particular asymptotic expansions, of the solution up to $\scri^+$ in terms of Sobolev norms of the initial data on the spacelike hypersurface $\mathcal{S}_{-1+\varepsilon}$ (\Cref{Proposition1}, \Cref{section:asymptotic_expansions_near_future_infinity}).
\end{enumerate}
We point out that {although} the control of the solution is quantitative, it is certainly not sharp\footnote{By \emph{sharp} here we mean no loss of derivatives in evolution from $\scri^-$ to $\scri^+$.}, not least because Sobolev embeddings necessarily lose a small amount of regularity. More severely, in order to make contact across $\mathcal{S}_{-1+\varepsilon}$ our method requires an overshoot of regularity at $\scri^-$.

\subsection{From $\scri^-$ to $\mathcal{S}_{-1+\varepsilon}$}
\label{sec:from_past_to_Cauchy}

In order to translate the requirements of \Cref{Proposition2,Proposition3} in terms of intrinsically defined Sobolev spaces on $\scri^-_{\rho_\star}$ and $\underline{\mathcal{B}}_\varepsilon$, we must treat the two null surfaces separately. The requirement on past null infinity is immediate since $\bmpartial_\rho$ is tangent to $\scri^-$ and $\bmpartial_\tau$ is uniformly transverse to $\scri^-$; setting $p_- = 0$ and $m_- + 3 < q_-$ as discussed {in \Cref{Section:expansions_near_past_infinity}}, we deduce that the data $f \in \{ \psi, \psi_0, \psi_1, \psi_2 \}$ on $\scri^-_{\rho_\star}$ having the regularity
\begin{align}
    \left(  f, \partial_\tau f, \dots ,\partial^{q_-+m_-+3}_\tau f \right) \in H^{m_- +q_- +3}(\scri^-_{\rho_\star}) \times H^{m_- +q_- +2}(\scri^-_{\rho_\star}) \times \dotsc \times L^2(\scri^-_{\rho_\star})
    \label{RegularityConditionsScriMinus}
\end{align}
is sufficient to satisfy the conditions of \Cref{Proposition2,Proposition3} on $\scri^-$ {(i.e. the bound on the first term in condition \eqref{initialDataBound} holds)}. 

{To treat $\underline{\mathcal{B}}_\varepsilon$, recall that:} $\bme_{\bmzero \bmzero'} = (1-\tau) \bmpartial_\tau + \rho \bmpartial_\rho$ is a null generator of $\underline{\mathcal{B}}_\varepsilon$, the coefficients $(1-\tau)$ and $\rho$ are uniformly equivalent to $1$ on $\underline{\mathcal{B}}_\varepsilon$, and the vector $\bmpartial_\rho$ is uniformly transverse to $\underline{\mathcal{B}}_\varepsilon$. We therefore have for any $q$ and $m$
\begin{align*}
    \int_{\underline{\mathcal{B}}_\varepsilon} \sum_{q' + p' + |\alpha| \leq m} | \partial_\tau^{q+q'} \bmZ^\alpha \partial_\rho^{p'} f|^2 \, \d \underline{\mathcal{B}} & \lesssim \int_{\underline{\mathcal{B}}_\varepsilon} \sum_{q' + p' + |\alpha| \leq m} | (\bme_{\bmzero \bmzero'})^{q+q'} \bmZ^\alpha \partial_\rho^{p'} f|^2 \, \d \underline{\mathcal{B}} \\
    & + \int_{\underline{\mathcal{B}}_\varepsilon} \sum_{q' + p' + |\alpha| \leq m} | \bmZ^\alpha \partial_\rho^{q+ q' + p'} f|^2 \, \d \underline{\mathcal{B}} \\
    & \lesssim \sum_{j=0}^m \| \partial_\rho^j f \|^2_{H^{m+q-j}(\underline{\mathcal{B}}_\varepsilon)} + \sum_{j=q}^{q+m} \| \partial_\rho^j f \|^2_{H^{q+m-j}(\underline{\mathcal{B}}_\varepsilon)} \\
    & \lesssim \sum_{j=0}^{m+q} \| \partial_\rho^j f \|^2_{H^{m+q-j}(\underline{\mathcal{B}}_\varepsilon)},
\end{align*}
where in the first inequality we used the fact that $(\bme_{\bmzero \bmzero'})^s f$ is uniformly equivalent in $L^2(\underline{\mathcal{B}}_\varepsilon)$ to the sum of norms of $\partial_\tau^s f $ and $\partial_\rho^s f$ for any $s \in \mathbb{N}$. Then, {setting $m=m_- + 3$, the data $f \in \{ \psi, \psi_0, \psi_1, \psi_2 \}$ on $\underline{\mathcal{B}}_\varepsilon$ having the same regularity requirement as \eqref{RegularityConditionsScriMinus}, with $\partial_\tau$ replaced with $\partial_\rho$ and $\scri^-_{\rho_\star}$ replaced with $\underline{\mathcal{B}}_\varepsilon$, is sufficient to ensure that the bound on the second term in condition \eqref{initialDataBound} holds.} {In summary, if the data $f \in \{\psi, \psi_0, \psi_1, \psi_2\}$ on $\scri^-_{\rho_\star} \cup \underline{\mathcal{B}}_\varepsilon$ satisfies
\[ (f, \partial_{\nparallel} f, \ldots, \partial_{\nparallel}^{q_-+m_- +3} f) \in H^{m_-+q_-+3} \times H^{m_- + q_- + 2} \times \ldots \times L^2, \]
with $\partial_{\nparallel} = \partial_\tau$ on $\scri^-_{\rho_\star}$ and $\partial_\nparallel = \partial_\rho$ on $\underline{\mathcal{B}}_\varepsilon$, then condition \eqref{initialDataBound} is satisfied, and by
} \Cref{Proposition2,Proposition3}, the solution in particular satisfies $\partial^{q_-}_\tau f \in H^{m_-+3}(\mathcal{\underline{N}}_\varepsilon)$ and, for $0<\alpha\leq \frac{1}{2}$, admits the expansion 
\begin{equation}
\label{Taylor_expansion_in_lower_domain}
    f = \sum_{q'=0}^{q_--1}\frac{1}{q'!}(\partial_{\tau}^{q'}f)\rvert_{\scri^-}(\tau+1)^{q'} + C^{m_-,\alpha}(\mathcal{\underline{N}}_\varepsilon),
\end{equation}
{as shown in \Cref{Section:expansions_near_past_infinity}.}

\subsection{From $\mathcal{S}_{-1+\varepsilon}$ to $\scri^+$}
Similarly, an application of \Cref{Proposition1} with $p = p_+$ and $m = m_+ + 3$ shows\footnote{It is easily seen that the norm on the initial data $\int_{\mathcal{S}_\star} \sum_{p'+q'+|\alpha| \leq m} |\partial_\tau^{q'} \partial_\rho^{p'+p} \bmZ^\alpha f|^2 \, \d \rho \, \d \mu$ is controlled by $\|\partial^m_\tau f\|^2_{H^p(\mathcal{S}_\star)} + \ldots + \|f \|^2_{H^{p+m}({\mathcal{S}_\star})}$.} that if the data for $f \in \{\psi, \psi_0, \psi_1, \psi_2\}$ on $\mathcal{S}_{-1+\varepsilon}$ has the regularity
\[ ( f, \partial_\tau f, \ldots, \partial_\tau^{m_+ + 3} f ) \in H^{p_+ + m_+ + 3}(\mathcal{S}_{-1+\varepsilon}) \times H^{p_+ + m_+ +2}(\mathcal{S}_{-1+\varepsilon}) \times \ldots \times H^{p_+}(\mathcal{S}_{-1+\varepsilon}) \]
for $p_+ > m_+ + 4$, then the solution has the regularity $\partial^{p_+}_{\rho} f \in H^{m_++3}({\mathscr{D}})$, and in particular admits the expansion
\begin{equation}
    \label{Taylor_expansion_upper_domain}
    f = \sum_{p'=0}^{p_+-1}\frac{1}{p'!}f^{(p')}\rho^{p'} + C^{m_+,\alpha}({\mathscr{D}}),
\end{equation}
{as shown in \Cref{section:asymptotic_expansions_near_future_infinity}}.

\medskip

\subsection{Obtaining the desired regularity at $\scri^+$}

Noting the regularity prescribed at $\scri^-$ in \Cref{sec:from_past_to_Cauchy} and the form of the expansion \eqref{Taylor_expansion_in_lower_domain}, we have that on $\mathcal{S}_{-1+\varepsilon}$ the $\tau$-derivatives of the components $f$ have the regularity
\begin{align*}
\partial_\tau^s f|_{\mathcal{S}_{-1+\varepsilon}} &= \sum_{q' = s}^{q_- -1} \frac{1}{(q' - s)!}(\partial_\tau^{q'} f)|_{\scri^-} \varepsilon^{q'-s} + C^{m_- -s}(\mathcal{S}_{-1+\varepsilon}) \\
& \in H^{m_- + 2}(\mathcal{S}_{-1+\varepsilon}) + H^{m_- -s}(\mathcal{S}_{-1+\varepsilon}) \\
& \subset H^{m_- -s}(\mathcal{S}_{-1+\varepsilon})
\end{align*}
for $0 \leq s \leq m_-$ (for $s \geq q_-$ the sum being empty), where in the second line we used the fact that a $C^k(\mathcal{S}_\tau)$ function is in $H^k(\mathcal{S}_\tau)$ as the domain of integration is compact\footnote{Note that the metric distance with respect to $\bmeta$ to $\mathcal{I} = \{ \rho = 0 \}$ is infinite, however.} in $\rho$. Hence to satisfy the above requirements for propagation up to $\scri^+$ with expansion \eqref{Taylor_expansion_upper_domain} we need $m_- - (m_+ + 3) \geq p_+ > m_+ + 4$, as well as $m_- +3 < q_- $, i.e. altogether we require
\[ q_- \geq 2m_+ +12. \]

\subsection{Main result}
\label{Subsection:MainResult}

By setting $p_+ = m_+ + 5$, $q_- = m_- + 4$, and $m_- = p_+ + m_+ + 3 = 2m_+ + 8$, we obtain our main theorem in the following form.

\begin{theorem}
    \label{thm:main_theorem_transverse_derivatives}
    Let $\rho_{*}>0$, $0<\varepsilon \ll 1$ be real numbers and $m_+ \in \mathbb{N}$ be an integer.  Suppose the asymptotic characteristic data\footnote{Cf. \Cref{reducedInitialData}.} for $\square  \phi = 0$ for the components $f \in \{\psi, \psi_0, \psi_1, \psi_2 \}$ on $\scri^-_{\rho_\star} \cup \underline{\mathcal{B}}_\varepsilon$ has the regularity
    \begin{equation*}
    (f, \partial_\nparallel f, \ldots, \partial_\nparallel^{4m_+ + 23} f) \in H^{4m_+ + 23} \times H^{4m_+ + 22} \times \ldots \times L^2,
    \end{equation*}
    where $\partial_\nparallel$ denotes a transverse derivative to $\scri^-$ or $\underline{\mathcal{B}}_\varepsilon$, i.e. $\partial_\nparallel = \partial_\tau$ on $\scri^-$ and $\partial_\nparallel = \partial_\rho$ on $\underline{\mathcal{B}}_\varepsilon$. Additionally, suppose that
    \[ \phi|_{\scri^-_{\rho_\star}} \in H^{4m_+ + 24}(\scri^-_{\rho_\star}) \quad \text{and} \quad \phi|_{\underline{\mathcal{B}}_\varepsilon} \in H^{4m_+ + 24}(\underline{\mathcal{B}}_\varepsilon). \]
   Then in the domain $\mathscr{D}\equiv \mathcal{\underline{N}}_\varepsilon \cup \mathcal{N}_1$ this data gives rise to a unique solution to the wave equation which {in $\mathscr{D}$} admits the expansion
    \begin{equation}
        \label{expansion_main_theorem_transverse_derivatives}
        \phi(\tau, \rho, t^\bmA{}_\bmB) = \sum_{p' = 0}^{{m_+}} \frac{1}{p'!} \phi^{(p')}(\tau, t^\bmA{}_\bmB)\rho^{p'} + C^{m_+,\alpha}({\mathscr{D}}).
    \end{equation}
  Moreover, the leading term
    \begin{equation}
    \label{main_theorem_analytic_term}
    \sum_{p' = 0}^{m_+} \frac{1}{p'!} \phi^{(p')}(\tau, t^\bmA{}_\bmB) \rho^{p'}
    \end{equation}
    in the expansion \eqref{expansion_main_theorem_transverse_derivatives} {has analytic projections to each fixed spherical harmonic} (in $\tau$ for $\tau \in [-1,1]$).
    
\end{theorem}

\begin{remark}
    \label{rmk:main_theorem_remark_on_expansion}
  {\em A remark is in order regarding the way we have chosen to state \Cref{thm:main_theorem_transverse_derivatives}, by separating the ``leading order'' terms \eqref{main_theorem_analytic_term} and the remainder in the expansion \eqref{expansion_main_theorem_transverse_derivatives}. Because the terms in \eqref{main_theorem_analytic_term} are {smoother}, they can be absorbed into the remainder term without any apparent loss of information; indeed, the splitting into the two parts in expression \eqref{expansion_main_theorem_transverse_derivatives}, as written, appears at first sight to be arbitrary, since the remainder may contain terms that might ruin the structure of the ``leading order'' terms. The emphasis here, however, is on keeping the same structure of the expansion as first found by Friedrich \cite{Fri03b}, where the leading order terms are generically less regular than the remainder. That is, expansion \eqref{expansion_main_theorem_transverse_derivatives} holds \emph{without} assuming the conditions of \Cref{thm:main_theorem_transverse_derivatives} at $\scri^-_{\rho_\star} \cup \underline{\mathcal{B}}_\varepsilon$: it is sufficient to assume the regularity requirements \eqref{Proposition1_hypothesis} on a Cauchy slice, which are weaker than the conditions of \Cref{thm:main_theorem_transverse_derivatives}. The content of the theorem, then, is the statement that the potential polyhomogeneous terms in \eqref{main_theorem_analytic_term} \emph{do not arise} under the stated assumptions in the asymptotic past.}
\end{remark}

\begin{proof}
    We first show that expansion \eqref{Taylor_expansion_upper_domain} implies an expansion of the form claimed in the theorem. {We begin by showing that we have an expansion up to order $\rho^{m_+ + 4}$, and then show that the extra terms are no more regular than the remainder, and absorb them into the $C^{m_+,\alpha}(\mathscr{D})$ term.} Using \eqref{Taylor_expansion_upper_domain} for the component $f=\psi = \sqrt{2} \partial_\tau \phi$ and integrating in $\tau$ from $-1$ to $\tau$, we obtain, recalling that $p_+ = m_+ + 5$,
    \[ \phi(\tau, \rho, t^\bmA{}_\bmB) - \phi(-1,\rho,t^\bmA{}_\bmB) = \sum_{p'=0}^{m_+ + 4}\frac{1}{p'!} \phi^{(p')}(\tau, t^\bmA{}_\bmB) \rho^{p'} - \sum_{p'=0}^{m_+ + 4}\frac{1}{p'!} \phi^{(p')}(-1, t^\bmA{}_\bmB) \rho^{p'} + C^{m_+,\alpha}(\mathscr{D}). \]
    By assumption, $\phi|_{\scri^-_{\rho_\star}} \in H^{4m_+ + 24}(\scri^-_{\rho_\star}) \hookrightarrow C^{4m_+ + 21,\alpha}(\scri^-_{\rho_\star})$, so, considered as a function on $\mathscr{D}$ which is constant in $\tau$, the second term on the left-hand side of the above can be absorbed into the $C^{m_+, \alpha}(\mathscr{D})$ remainder. Similarly, taking the trace on $\{ \rho = 0 \}$ at the cost of losing one derivative,
    \[ \phi^{(p')}(-1,t^\bmA{}_\bmB) = (\partial_\rho^{p'} \phi |_{\scri^-_{\rho_\star}})|_{\rho = 0} \in H^{4m_++24 - p' - 1}(\mathbb{S}^3) \hookrightarrow C^{3m_+ + 16,\alpha}(\mathbb{S}^3) \]
    for all $0 \leq p' \leq m_+ + 4$ and $\alpha \leq \frac{1}{2}$. So the second term on the right-hand side can also be absorbed into the remainder (observe that polynomials in $\rho$ are smooth and bounded on $\mathscr{D}$, since $\rho$ is bounded on $\mathscr{D}$). {This gives us an expansion of the form \eqref{expansion_main_theorem_transverse_derivatives}, with the upper limit $m_+$ in the explicit part replaced with $m_+ + 4$:
    \begin{equation}
        \label{main_thm_expansion_to_higher_order}
        \phi(\tau, \rho, t^\bmA{}_\bmB) = \sum_{p' = 0}^{m_+ + 4} \frac{1}{p'!} \phi^{(p')}(\tau, t^\bmA{}_\bmB)\rho^{p'} + C^{m_+,\alpha}(\mathscr{D}).
    \end{equation}
    }

{It remains to show that the explicit term \eqref{main_theorem_analytic_term} {has} analytic in $\tau$ {projections to each spherical harmonic}, and handle the terms of order $\rho^{p'}$ for $m_+ + 1 \leq p' \leq m_+ + 4$. In addition to the expansion \eqref{main_thm_expansion_to_higher_order}, here we also employ expansion \eqref{Taylor_expansion_in_lower_domain}, which is valid only in $\underline{\mathcal{N}}_\varepsilon$. Integrating this in $\tau$ from $-1$ to $\tau$ for $f= \psi$, we obtain
    \[ \phi (\tau,\rho, t^\bmA{}_\bmB) - \phi(-1,\rho,t^\bmA{}_\bmB) = \sum_{q'=0}^{q_- - 1} \frac{1}{(q'+1)!} (\partial_\tau^{q' +1}\phi)|_{\scri^-}(\rho,t^\bmA{}_\bmB) (\tau+1)^{q'+1} + C^{m_-, \alpha}(\underline{\mathcal{N}}_\varepsilon), \]
    where $m_- = 2m_+ + 8$, $q_- = 2m_+ + 12$, and we have already seen that $\phi(-1,\rho,t^\bmA{}_\bmB)$ has the regularity $C^{4m_+ + 21, \alpha}$, and so can be absorbed into the remainder. Equating this with \eqref{main_thm_expansion_to_higher_order}, we have
    \[ \sum_{p'=0}^{m_+ + 4} \frac{1}{p'!} \phi^{(p')}(\tau, t^\bmA{}_\bmB) \rho^{p'} + C^{m_+, \alpha}(\mathscr{D}) = \sum_{q' = 0}^{2m_+ + 11} \frac{1}{(q'+1)!} (\partial_\tau^{q'+1}\phi)|_{\scri^-}(\rho, t^\bmA{}_\bmB)(\tau + 1)^{q' + 1} + C^{2m_+ + 8,\alpha}(\underline{\mathcal{N}}_\varepsilon). \]
    In particular, this implies that, as functions of $\tau$, $\phi^{(p')}(\tau, t^\bmA{}_\bmB)$ are of class $C^{m_+,\alpha}$ in $\underline{\mathcal{N}}_\varepsilon$ for $0 \leq p' \leq m_+ + 4$, as the entire expression on the right-hand side is more regular. But the general expressions for the functions $\phi^{(p')}$ are known and are given in \Cref{Appendix:InteriorEquations}. In particular, they are given by expansions \eqref{ExpansionHarmonics} in terms of spherical harmonics, whose coefficients are given by \eqref{JacobiLogarithmicSolution}, and whose regularity for a given $p'$ is generically at most $C^{p'-1}$ in $\tau$ near $\tau = - 1$. If $\phi^{(p')}(\tau, t^\bmA{}_\bmB)$ is $C^{p'}$ in $\tau$ near $\tau = -1$, then the constants $\mathfrak{d}_{p',p'+m,j}$ in \eqref{JacobiLogarithmicSolution} vanish for all $0 \leq j \leq 2(p'+m)$, $m \geq 0$, and in particular $\phi^{(p')}$ {has analytic projections to each fixed spherical harmonic}. We therefore deduce {the claimed regularity for} $\phi^{(p')}(\tau, t^\bmA{}_\bmB)$ in $\tau$ near $\tau = -1$ for $m_+ \geq p'$. The terms in \eqref{main_thm_expansion_to_higher_order} for which we do not deduce analyticity, i.e. $\phi^{(p')}$ for $m_+ + 1 \leq p' \leq m_+ + 4$, may contain logarithmic terms of the form $(1 \pm \tau)^{p'} \log(1\pm \tau)$ (and are otherwise analytic). But these logarithmic terms are at least $C^{m_+, \alpha}$ in $\tau$ (for any $\alpha < 1$), and hence the terms $\phi^{(p')}$ for $m_+ + 1 \leq p' \leq m_+ + 4$ can be absorbed into the remainder in \eqref{main_thm_expansion_to_higher_order}, giving \eqref{expansion_main_theorem_transverse_derivatives}. Finally, the expansion \eqref{expansion_main_theorem_transverse_derivatives} is valid in the entire domain $\mathscr{D}$, not just $\underline{\mathcal{N}}_\varepsilon$, so the analyticity of {the projections to spherical harmonics of} \eqref{main_theorem_analytic_term} extends to the full range $\tau \in [-1,1]$, by the symmetry of logarithmic terms in \eqref{JacobiLogarithmicSolution}.
    }
\end{proof} 

\begin{remark}
{\em   {In physical applications, the notion that there should be no radiation entering the spacetime through past null infinity may be encoded in the} no-incoming radiation condition, {which} for the wave equation {reads} \cite{Sommerfeld1912,Mad70,Som92}
\begin{equation}
\partial_v (r \tilde{\phi})|_{\mathscr{I}^-}=0.
\label{NoIncomingRadiationCondition}
\end{equation}
 In conformal coordinates \eqref{NoIncomingRadiationCondition} takes the form $\partial_\rho \phi|_{\scri^-} = 0$. {Combined with \eqref{NoIncomingRadiationCondition} (for simplicity, one can set $\phi$ together with all its transverse derivatives to vanish on $\scri^-_{\rho_\star}$, for instance), \Cref{thm:main_theorem_transverse_derivatives} may be seen as providing conditions on data on $\underline{\mathcal{B}}_\varepsilon$ which ensure that the solution possesses a desired amount of conformal regularity near spatial infinity.}
}
\end{remark}

\section{Physical coordinates}
\label{Section:Interpretation}

In this section we provide a characterisation of the regularity conditions in the statement of \Cref{thm:main_theorem_transverse_derivatives} in terms of physical coordinates. {For simplicity, we assume in this section that the functions are smooth {and uniformly bounded} in the angular coordinates, and treat only regularity and decay in the radial coordinates $r$ and $v$.}

\subsection{{Regularity in terms of physical coordinates}}

\subsubsection{Past null infinity}
Using the expressions \eqref{Friedrich_vectors_past_scri} and \eqref{Friedrich_form_past_scri}, the condition that
\[ \phi|_{\scri^-} \in H^l(\scri^-_{\rho_\star}, {\d \rho \, \d \mu}) \iff \sum_{p + |\alpha| \leq l} \int_{\mathrm{SU}(2)} \int_0^{\rho_*} | \partial_\rho^p \bmZ^\alpha \phi|_{\scri^-}|^2 \, \d \rho \, \d \mu < \infty \]
{ may be stated as} ({recall that} $\phi|_{\scri^-}$ {is} smooth {and uniformly bounded} in the angular coordinates),
\begin{equation}
    \label{Sobolev_norm_on_past_null_infinity_in_physical_coordinates}
    \int_{v_\star}^\infty \left| (v^2 \partial_v)^p (r \tilde{\phi})|_{\scri^-}\right|^2 \frac{\d v}{v^2} < \infty \quad \text{for all } ~ 0\leq p \leq l,
\end{equation}
where $v_\star = 1/2\rho_\star$, and $(r\tilde{\phi})|_{\scri^-}(v,\omega) = \lim_{r \to \infty} r\tilde{\phi}$, where $\omega$ denotes the angular coordinates and the limit is taken along lines of constant $v$. {For the transverse derivatives, we have $\partial_\tau \phi |_{\scri^-} = \frac{v}{2} \partial_v \phi |_{\scri^-} - \frac{r^2}{v}\partial_r \phi|_{\scri^-}$, so $\partial_\tau \phi|_{\scri^-}$ is guaranteed to be in $L^2(\scri^-_{\rho_\star}, \d \rho \, \d \mu)$ provided $v \partial_v \phi|_{\scri^-}, \, \frac{r^2}{v} \partial_r \phi|_{\scri^-} \in L^2(\scri^-_{\rho_\star}, \frac{\d v}{v^2})$. Similarly, higher $(\partial_\tau)^l$ derivatives will be in $L^2(\scri^-_{\rho_\star}, \d \rho \, \d \mu)$ provided $(v \partial_v)^q \phi |_{\scri^-} \in L^2(\scri^-_{\rho_\star}, \frac{\d v}{v^2})$ and $(\frac{r^2}{v} \partial_r)^q \phi|_{\scri^-} \in L^2(\scri^-_{\rho_\star}, \frac{\d v}{v^2})$ for $0 \leq q \leq l$. In order to satisfy the assumptions of \Cref{thm:main_theorem_transverse_derivatives} on $\scri^-_{\rho_\star}$, the following conditions are therefore sufficient:
\begin{align}
    \label{Sobolev_norm_past_null_infinity_transverse_derivative_physical_coords_first_part}
    & \int_{v_\star}^\infty \left| (v^2 \partial_v)^p (v \partial_v)^q (r \tilde{\phi})|_{\scri^-} \right|^2 \frac{\d v}{v^2} < \infty, \\
    \label{Sobolev_norm_past_null_infinity_transverse_derivative_physical_coords_second_part}
    & \int_{v_\star}^\infty \left| (v^2 \partial_v)^p \left( \frac{r^2}{v} \partial_r \right)^q \! (r \tilde{\phi})|_{\scri^-} \right|^2 \frac{\d v}{v^2} < \infty
\end{align}
for all $0 \leq p + q \leq 4m_+ + 24$ ($p \geq 0$, $q \geq 0$). Here the limit at $\scri^-$ is understood as being taken after applying the $r^2 \partial_r$ derivatives. Note that condition \eqref{Sobolev_norm_past_null_infinity_transverse_derivative_physical_coords_first_part} is superfluous, and can be replaced with the same statement with $q = 0$, since $v^2 \partial_v$ imposes stronger decay than $v\partial_v$.
}

\subsubsection{The incoming hypersurface $\underline{\mathcal{B}}_\varepsilon$}

We compute, in $(v,R)$ coordinates,
\[ \bme_{\bmzero \bmzero'} = \frac{1}{\sqrt{2}}\left(\frac{2-Rv}{v} \right)\partial_R \quad \text{and} \quad \bme_{\bmone \bmone'} = \frac{1}{\sqrt{2}} \left( \frac{v}{2-Rv} \right) \left( 2 \partial_v - R^2 \partial_R \right), \]
so that
\[ \d \underline{\mathcal{B}}|_{\underline{\mathcal{B}}_\varepsilon} = \frac{\sqrt{2}}{(2-Rv)^2} \d R \wedge \d \mu. \]
Note that in $(v,R)$ coordinates the rescaled metric is given by 
\[
\bmeta = R^2 \d v^2 +  (\d v \otimes\d R+ \d R\otimes\d v)  - \bmsigma,
\]
so that $\bme_{\bmzero \bmzero'}$ and $\bme_{\bmone \bmone'} $ are clearly null, and $\bmeta(\bme_{\bmzero \bmzero'}, \bme_{\bmone \bmone'}) = 1$, as expected. Write $\underline{\phi} = \phi|_{\underline{\mathcal{B}}_\varepsilon}$. The condition
\[ \underline{\phi} \in H^l(\underline{\mathcal{B}}_\varepsilon, {\d  \underline{\mathcal{B}}}) \iff \sum_{q+|\alpha| \leq l} \int_{\mathrm{SU}(2)}\int_0^{R_\varepsilon} \left| (\bme_{\bmzero \bmzero'})^q \bmZ^\alpha \underline{\phi} \right|^2 \d R \, \d \mu <\infty, \]
where $R_\varepsilon = \varepsilon/v_\star$ then implies, again for $\underline{\phi}$ smooth in the angular coordinates for simplicity, and using that $(2-Rv)$ is uniformly equivalent to $1$ on $\underline{\mathcal{B}}_\varepsilon$,
\[ \int^\infty_{r_\star} \left| (r^2 \partial_r )^q (r \tilde{\underline{\phi}}) \right|^2 \frac{\d r}{r^2} < \infty \quad \text{for all} ~ 0 \leq q \leq l, \]
where $r_\star = {v_\star}/{\varepsilon}$ and $\underline{\tilde{\phi}} = r^{-1} \underline{\phi}$. {The transverse derivative $\partial_\rho$ to $\underline{\mathcal{B}}_\varepsilon$ has the expression $\partial_\rho = v(2-Rv)(-v \partial_v + R\partial_R)$, from which we deduce, similarly to \eqref{Sobolev_norm_past_null_infinity_transverse_derivative_physical_coords_first_part}--\eqref{Sobolev_norm_past_null_infinity_transverse_derivative_physical_coords_second_part}, that the assumptions of \Cref{thm:main_theorem_transverse_derivatives} on $\underline{\mathcal{B}}_\varepsilon$ are satisfied provided
\begin{align}
    & \int_{r_\star}^\infty | (r^2 \partial_r)^q (r \tilde{\underline{\phi}})|^2 \frac{\d r}{r^2} < \infty, \\
    & \int_{r_\star}^\infty |(r^2\partial_r)^q (v^2 \partial_v)^p (r\tilde{\underline{\phi}})|^2 \frac{\d r}{r^2} < \infty
\end{align}
for all $0 \leq p + q \leq 4m_+ + 24$, $p, q \geq 0$, where it is understood that the restriction to $\underline{\mathcal{B}}_\varepsilon$ (denoted by underlining $\tilde{\phi}$) is taken after applying any transverse derivatives.
}

\subsection{Asymptotic expansions in physical coordinates}
\label{Subsection:Peeling}

{Finally, we briefly recap how the conclusions of \Cref{thm:main_theorem_transverse_derivatives} lead to an asymptotic expansion in the physical spacetime. When we write $\scri^+$ here, we mean the part of $\scri^+$ contained in $\partial \mathscr{D}$, i.e. for $u \ll -1$. Suppose that the conditions of \Cref{thm:main_theorem_transverse_derivatives} are satisfied; then $\phi \in C^{m_+,\alpha}(\mathscr{D})$. In particular, $\phi$ admits a Taylor expansion around $\tau = 1$,
\begin{equation}
    \label{Taylor_expansion_around_scri}
    \phi(\tau, \rho, \omega) = \sum_{k=0}^{m_+} \frac{1}{k!} \phi^{(k)}(\rho, \omega) (\tau - 1)^k + o((\tau - 1)^{m_+}), 
\end{equation}
where $\phi^{(k)} \in C^{m_+ -k, \alpha}(\scri^+)$. Using $\tau - 1 = \frac{u}{r}$ and $\phi = r\tilde{\phi}$, we can rewrite \eqref{Taylor_expansion_around_scri} as, for $r \gg 1$,
\begin{equation}
    \label{asymptotic_expansion}
    \tilde{\phi}(r,u,\omega) = \frac{1}{r} \sum_{k=0}^{m_+} \tilde{\phi}^{(k)}(u,\omega) \frac{1}{r^{k}} + o\left(r^{-(m_+ + 1)}\right)
\end{equation}
for some functions $\tilde{\phi}^{(k)}(u,\omega) \in C^{m_+ - k,\alpha}(\scri^+)$.
}

\medskip

\section*{Acknowledgements}
JAVK thanks the hospitality of the Erwin Schr\"odinger Institute for Mathematics and Physics of the University of Vienna where part of this research was carried out as part of the workshop \emph{Carrollian Physics and Holography}.
We thank the anonymous referees for comments which have improved the presentation of this article.
\appendix

\section*{Use of AI}
No AI has been used in the derivation of the equations, construction of the proofs or writing of the text in this article. ChatGTP 5.6 (Sol) has been used to check the overall consistency of the text and equations.

\section{A lemma of Friedrich}
\label{Appendix:FLemma}

For completeness here we record a technical lemma due to Friedrich \cite{Fri03b} to elegantly deal with angular derivatives when computing the estimates in the main text.

\begin{lemma}\label{FTL}
    Given a multi-index $\alpha = (\alpha_1,\alpha_2,\alpha_3)$ of non-negative integers, set
    \begin{align*}
        \bmZ^\alpha = {\sqrt{\frac{\alpha_1 ! \alpha_2 ! \alpha_3 !}{\alpha !}}} \bmZ_{1}^{\alpha_1}\bmZ_{2}^{\alpha_2}\bmZ_{3}^{\alpha_3},
    \end{align*}
    where $\bmZ_k$ denotes the left-invariant vector fields on $\mathrm{SU}(2)$, with the convention that $\bmZ^\alpha$ reduces to the identity operator if $\lvert\alpha\rvert = \alpha_1 + \alpha_2 + \alpha_3 = 0$. Then for any smooth $\mathbb{C}$-valued functions $f$, $g$ on $\mathrm{SU}(2)$, any $k\in\{1,2,3\}$ and $m \in \mathbb{N}$, one has
    \begin{align*}
        \sum_{\lvert\alpha\rvert=m}\left( \bmZ^\alpha\bmZ_k f \bmZ^\alpha g+\bmZ^\alpha f \bmZ^\alpha\bmZ_k g\right) = \sum_{\lvert\alpha\rvert=m}\bmZ_k\left( \bmZ^\alpha f \bmZ^\alpha g \right),
    \end{align*}
    from which it follows immediately that
    \begin{align*}
        \sum_{\lvert\alpha\rvert=m}\int_{\mathrm{SU}(2)}\left( \bmZ^\alpha\bmX_{\pm} f \bmZ^\alpha g+\bmZ^\alpha f \bmZ^\alpha\bmX_{\pm} g \right) \mathrm{d} \mu = 0.
    \end{align*}
\end{lemma}

\section{Interior equations on \texorpdfstring{$\mathcal{I}$}{I}}
\label{Appendix:InteriorEquations}

The total characteristic nature of the cylinder at spatial infinity for the wave equation is reflected in the fact that the wave operator $\square$ reduces, upon evaluation on $\mathcal{I}$, to the interior operator $\mathring{\square}\equiv\square|_{\mathcal{I}}$, which acts on a scalar $\phi$ by
\[
\mathring{\square} \phi \equiv (1-\tau^2) \ddot\phi
-2\tau \dot\phi -\frac{1}{2}\{ \bmX_+, \, \bmX_- \} \phi,
\]
where $\dot{\phi} = \partial_\tau \phi$ and $\{\bmX_+, \, \bmX_- \}$, denoting the anticommutator of $\bmX_+$ and $\bmX_-$, is proportional to the Laplacian on the unit $2$-sphere. Computing the commutator
\[ 
[ \partial_\rho^p, \, \square] \phi =
2\tau p \partial_\rho^p\dot{\phi} -2 p\rho \partial_\rho^p \phi' - p(p-1)\partial_\rho^p\phi \]
 we obtain
\[
\square \partial^p_\rho \phi + 2\tau p \partial^p_\rho\dot{\phi}-2 \rho p \partial^p_\rho \phi' -p(p-1)\partial_\rho^p\phi =0. 
\]
 Evaluating on $\mathcal{I}$, one therefore finds the following set of interior equations on $\mathcal{I}$ for all $p \geq 0$,
\begin{equation}
\label{InteriorEquationsOnCylinder}
\mathring{\square} \phi^{(p)} + 2\tau p \dot{\phi}^{(p)} -p(p-1)\phi^{(p)} =0. 
\end{equation}
The equations \eqref{InteriorEquationsOnCylinder} are solvable explicitly in terms of a basis of functions on $\mathrm{SU}(2)$, as we explain below.

\subsection{Expansions in terms of Friedrich's basis $T_m{}^j{}_k$}
\label{Section:Expansions_In_Harmonics}

According to the Peter--Weyl theorem, the functions $\sqrt{m+1} T_m{}^j{}_k$, $m \in \mathbb{N} \cup \{ 0 \}$, $j, \, k \in \{0, 1, \ldots, m\}$, form an orthonormal basis for $L^2(\mathrm{SU}(2),\d\mu)$ (see e.g. Appendix A.3 of \cite{TauVal23}, \cite{BeyDouFraWha12} or \cite{Fri98a}), so any $L^2$ function on $\mathrm{SU}(2)$ can be expanded in terms of this basis. Using the Hopf map, one can relate the functions $T_m{}^j{}_k$ to the more usual spin-weighted spherical harmonics ${}_s Y_{\ell m}$:
\begin{equation} \label{spin_weighted_harmonics}
{}_sY_{\ell m} = (-1)^{s+m} \sqrt{\frac{2\ell+1}{4\pi}} T_{2\ell}{}^{\ell-m}{}_{\ell-s}
\end{equation}
for $\ell \in \mathbb{N} \cup \{0 \}$ and $m, \,s \in \{ -\ell, -\ell+1, \ldots,\ell-1, \ell \}$. Further details can be found in  \cite{FriKan00,Val04d} (the reader unfamiliar with the functions $T_m{}^j{}_k$ may take \eqref{spin_weighted_harmonics} as a definition). Since the spin weight of $\phi$ is $s =0$, we expand the functions $\phi^{(p)}$ in \eqref{InteriorEquationsOnCylinder} as $\sum_{\ell,m} \phi^{(p)}_{\ell m} {}_0 Y_{\ell m}$, or equivalently
\begin{equation}
\phi^{(p)}= \sum_{\ell=0}^\infty \sum_{j=0}^{2\ell} a_{p;\ell,j} T_{2\ell}{}^j{}_{\ell},
\label{ExpansionHarmonics}
\end{equation}
where the coefficients $a_{p;\ell,j}$ are functions of $\tau$. 

From the expansion \eqref{ExpansionHarmonics} and equation \eqref{InteriorEquationsOnCylinder}, it follows then that the coefficient $a_{p;\ell,j}$ satisfies the ODE
\begin{equation}
    (1-\tau^2)\ddot{a}_{p;\ell,j} +2 (p-1)\tau  \dot{a}_{p;\ell,j} + \big( \ell^2+ \ell -p^2+p  \big)a_{p;\ell,j} =0,
\label{JacobiEqn:Raw}
\end{equation}
where the integers $(p,\ell,j)$ are such that  $0\leq j\leq 2\ell$. Equation \eqref{JacobiEqn:Raw} is an example of a \emph{Jacobi ordinary differential equation}. Jacobi equations are usually parametrised in the form
\begin{equation}
    D_{(n,\alpha,\beta)}a \equiv (1-\tau^2) \ddot{a} +\big(\beta-\alpha -(\alpha+\beta+2)\tau \big)\dot{a} + n(n+\alpha+\beta+1) a =0.
    \label{JacobiEqn:Model}
\end{equation}
A direct comparison between equations \eqref{JacobiEqn:Raw} and \eqref{JacobiEqn:Model} gives
\begin{subequations}
\begin{eqnarray}
&& \alpha = -p, \label{Jacobi:alpha}\\
&& \beta = -p, \label{Jacobi:beta} \\
&& n= n_1\equiv p+ \ell, \qquad \mbox{or} \qquad n=n_2\equiv p- \ell -1. \label{Jacobi:n}
\end{eqnarray}
\end{subequations}

The qualitative nature of the solutions to equation \eqref{JacobiEqn:Raw} differs depending on whether $\ell <p$ or $\ell \geq p$, as we explain below.

\subsubsection{Properties of the Jacobi ODE}
\label{Section:JacobiODE}

An extensive discussion of the solutions of the Jacobi equation can be found in the monograph \cite{Sze78} from which we borrow a number of identities. The solutions to  \eqref{JacobiEqn:Model} are given by the \emph{Jacobi polynomials} $P_n^{(\alpha, \beta)}(\tau)$, of degree $n$, defined by
\[
P_n^{(\alpha,\beta)}(\tau) \equiv \sum_{l=0}^n
\binom{n+\alpha}{l}\binom{n+\beta}{n-l}\left( \frac{\tau-1}{2}\right)^{n-l}\left(\frac{\tau+1}{2}\right)^l.
\]
In particular, one has $P_0^{(\alpha,\beta)}(\tau)=1$, and $P_n^{(\alpha,\beta)}(-\tau) =(-1)^n P_n^{(\beta,\alpha)}(\tau)$. The differential operator defined by \eqref{JacobiEqn:Model} exhibits the symmetries
\begin{subequations}
\begin{eqnarray}
&&\hspace{-1.5cm} D_{(n,\alpha,\beta)} \left( \left( \frac{1-\tau}{2}
   \right)^{-\alpha}a(\tau)  \right) = \left( \frac{1-\tau}{2}
   \right)^{-\alpha} D_{(n+\alpha,-\alpha,\beta)} a(\tau), \label{JacobiEqnIdentity1}\\
&&\hspace{-1.5cm} D_{(n,\alpha,\beta)} \left( \left( \frac{1+\tau}{2}
    \right)^{-\beta}a(\tau)  \right) = \left( \frac{1+\tau}{2}
    \right)^{-\beta} D_{(n+\beta,\alpha,-\beta)} a(\tau), \label{JacobiEqnIdentity2}\\
&&\hspace{-1.5cm} D_{(n,\alpha,\beta)} \left( \left( \frac{1-\tau}{2}
   \right)^{-\alpha}\left( \frac{1+\tau}{2}
    \right)^{-\beta}a(\tau)  \right) = \left( \frac{1-\tau}{2}
   \right)^{-\alpha}\left( \frac{1+\tau}{2}
    \right)^{-\beta} D_{(n+\alpha+\beta,-\alpha,-\beta)} a(\tau), \label{JacobiEqnIdentity3}
\end{eqnarray}
\end{subequations}
which hold for $|\tau|<1$, arbitrary $C^2$-functions $a(\tau)$, and
arbitrary real values of the parameters $\alpha$, $\beta$, $n$. An alternative definition of the Jacobi polynomials, convenient for
verifying when the functions vanish identically, is given by
\[
P_n^{(\alpha,\beta)}(\tau) = \frac{1}{n!}\sum_{k=0}^n c_k \left( \frac{\tau-1}{2} \right)^k,
\]
with
\begin{eqnarray*}
&& \hspace{-1cm}c_0\equiv (\alpha+1)(\alpha+2)\cdots(\alpha +n), \\
&& \hspace{2cm}\vdots \\
&& \hspace{-1cm} c_k \equiv \frac{n!}{k!(n-k)!}(\alpha+k+1)(\alpha+k+2)\cdots
   (\alpha+n) (n+1 +\alpha +\beta) (n+2+\alpha+\beta)\cdots
   (n+k+\alpha+\beta),\\
&& \hspace{2cm}\vdots \\
&& \hspace{-1cm}c_n\equiv (n+1+\alpha+\beta)(n+2+\alpha+\beta) \cdots (2n+\alpha+\beta).
\end{eqnarray*}

Of particular interest is the case when the two solutions are polynomial. This situation is described in the following classical lemma.

\begin{lemma}[\cite{Sze78}, Theorem 4.23.2]
\label{LemmaSzegoA3}
For integers $n, \alpha, \beta$ with $n\geq 0$ the only cases for which the general solution to \eqref{JacobiEqn:Model} is polynomial are:
\begin{itemize}
    \item [(a)] $\alpha, \beta <0$ with $\alpha \geq -n$, $\beta\geq -n$, $\alpha+\beta \leq -n-1$, $n\geq 1$;
    \item[(b)] $\alpha, \beta <0$, $\alpha<-n$, $\beta<-n$, $n\geq 0$.
\end{itemize}
\end{lemma}

\subsubsection{Solutions for $0\leq \ell <p$}

When $n=n_1$, we are in case (a) of \Cref{LemmaSzegoA3},  and when $n=n_2$, we are in case (b), therefore when $0 \leq \ell< p$ one has two linearly independent polynomial solutions. By direct inspection of the formulae above, the polynomial $P_{n_1}^{(\alpha,\beta)}(\tau)=P_{n_1}^{(-p,-p)}(\tau)$ with $(n_1,\,\alpha,\,\beta)=(p+ \ell,-p,-p)$ vanishes identically, while 
\[
Q_2(\tau) \equiv P_{n_2}^{(-p,-p)}(\tau)
\]
gives a polynomial of degree $n_2=p-\ell-1$. A further non-trivial solution can be written down using the identity \eqref{JacobiEqnIdentity1}; one finds a polynomial of degree $n_1 = p + \ell$ given by
\[
Q_1(\tau) \equiv \left(\frac{1-\tau}{2}  \right)^{p}P_{\ell}^{(p,-p)}(\tau).
\]
Since $n_2<n_1$, the solutions $Q_1$ and $Q_2$ are linearly independent. Yet another solution can be obtained using identity \eqref{JacobiEqnIdentity2}, namely
\[
Q_3(\tau)\equiv \left( \frac{1+\tau}{2} \right)^{p}P_{\ell}^{(-p,p)}(\tau),
\]
which, again, is a polynomial of degree $n_1$. It can be verified that $Q_1$ and $Q_3$ are also linearly independent. Making use of these solutions one can write down the general solution to equation \eqref{JacobiEqn:Raw}, for $0\leq \ell <p$, in the symmetric form
\begin{equation}
\label{a_coefficient_q_less_p}
a_{p;\ell,j}(\tau) = \mathfrak{c}_{p;\ell,j}\left(\frac{1-\tau}{2}  \right)^{p}P_{\ell}^{(p,-p)}(\tau) + \mathfrak{d}_{p;\ell,j}\left( \frac{1+\tau}{2} \right)^{p}P_{\ell}^{(-p,p)}(\tau),
\end{equation}
with $\mathfrak{c}_{_{p;\ell,j}}$ and $\mathfrak{d}_{_{p;\ell,j}}$ constants to be determined from the initial conditions. In particular, we have the following lemma:

\begin{lemma}
For $0\leq \ell <p$, the solutions to the Jacobi equation \eqref{JacobiEqn:Raw} are analytic at $\tau=\pm 1$.
\end{lemma}

\begin{remark}
It is clear that the formula \eqref{a_coefficient_q_less_p} does not give non-vanishing solutions if $p=\ell$.
\end{remark}

 \subsubsection{Solutions for $p=\ell$}
This case is no longer covered by Lemma \ref{LemmaSzegoA3}, therefore there is at most one polynomial solution. In this case we make use of identity \eqref{JacobiEqnIdentity3} with $n=n_1$, and look for solutions of the form
 \[
 a_{p;p,j}(\tau)= \left(\frac{1-\tau}{2}\right)^{p}\left(\frac{1+\tau}{2} \right)^{p} b(\tau),
 \]
 with $b(\tau)$ satisfying the equation
 \[
 D_{(0,p,p)}b(\tau) = (1-\tau^2) \ddot{b}(\tau) - 2\big( (p+1)\tau \big)\dot{b}(\tau)=0.
 \]
This can be integrated to give
 \begin{equation}
 b(\tau) = \mathfrak{c}_{p;p,j} + \mathfrak{d}_{p;p,j} \int_0^\tau \frac{\mathrm{d}\varsigma}{(1+\varsigma)^{p+1}(1-\varsigma)^{p+1}},
 \label{Solution:b}
 \end{equation}
 with $\mathfrak{c}_{_{p;p,j}}$ and $\mathfrak{d}_{p;p,j}$ constants. 
 Thus, the general solution to \eqref{JacobiEqn:Raw} for $p=\ell$ can be written as 
 \begin{equation}
 a_{p;p,j}(\tau)=\left( \frac{1-\tau}{2}\right)^{p}\left(\frac{1+\tau}{2} \right)^{p} \left(\mathfrak{c}_{p;p,j} + \mathfrak{d}_{p;p,j} \int_0^\tau \frac{\mathrm{d}\varsigma}{(1+\varsigma)^{p+1}(1-\varsigma)^{p+1}}\right).
 \label{JacobiSolution:Logarithmic}
 \end{equation}
Expanding the integrand of \eqref{Solution:b} in partial fractions, one sees that $ a_{p;p,j}(\tau)$ contains terms of the form
 \[
 (1-\tau)^{p}\log(1-\tau) \quad \text{and} \quad (1+\tau)^{p}\log(1+\tau),
 \]
 which are, respectively, of class $C^{p-1}$ and $C^{p-1}$ at $\tau=\pm 1$. These are the only singular terms in the solution \eqref{JacobiSolution:Logarithmic}. The rest of the solution is polynomial in $\tau$, and thus analytic at $\tau=\pm 1$. The solutions in the case $p=\ell$ are therefore not smooth at the critical sets $\mathcal{I}^\pm$ (cf. \cite{MMV22}).

\subsubsection{Solutions for $\ell > p$}
As in the case $p=\ell$, the conditions of \Cref{LemmaSzegoA3} are not satisfied and we have one polynomial solution and one with logarithmic divergences. Here, however, the structure of the singular solution is more complicated than in the case $p= \ell$. Set
\[
\ell=p+m, \qquad m\in \mathbb{N}.
\]
It follows then that
\[
\ell^2+\ell -p^2+p =(p+\ell)(\ell-p+1) = (2p+m)(m+1).
\]
Thus, one obtains the following reduced version of \eqref{JacobiEqn:Raw}:
\begin{equation}
(1-\tau^2)\ddot{a} +2\big(p-1)\tau \dot{a} +(2p+m)(m+1)a =0,
\label{JacobiEqn:Raw2}
\end{equation}
where here and below we have suppressed the indices of the function $a_{p;p+m,j}$ to ease readability. 
In particular, one has that
\[
n_1 =2p+m, \qquad \alpha = \beta=-p.
\]
Using the identity \eqref{JacobiEqnIdentity3} we then look for solutions of the form
\begin{equation}
a(\tau) = \left( \frac{1-\tau}{2}\right)^{p}\left( \frac{1+\tau}{2}\right)^{p}\tilde{a}(\tau),
\label{AnsatzSolution}
\end{equation}
where $\tilde{a}$ satisfies
\[
D_{(n',\alpha',\beta')}\tilde{a}=0,
\]
with
\begin{eqnarray*}
&& n'\equiv n_1+\alpha+\beta = m, \\
&& \alpha' \equiv -\alpha = p, \\
&& \beta' \equiv -\beta = p.
\end{eqnarray*}
Explicitly, one has that
\begin{equation}
D_{(n',\alpha',\beta')}\tilde{a} = (1-\tau^2) \ddot{\tilde{a}}
 -(2p+2)\tau\dot{\tilde{a}}+m(m+2p+1)\tilde{a}.
 \label{JacobiEqn:ReducedLog}
 \end{equation}
It follows then that the polynomial of degree $m$
\[
P_{n'}^{(\alpha',\beta')}(\tau) = P_m^{(p,p)}(\tau)
\]
is the only polynomial solution to equation \eqref{JacobiEqn:ReducedLog}. The second linearly independent solution is given by the Jacobi function of the second kind $Q_m^{(p,p)}(\tau)$. Following equation (4.62.1) in \cite{Sze78}, it admits the representation
\begin{equation}
Q_m^{(p,p)}(\tau) = \frac{1}{2(1-\tau^2)^p}\int_{-1}^1 (1-s^2)^p F(\tau,s)\, \mathrm{d}s + P_m^{(p,p)}(\tau) Q_0^{(p,p)}(\tau),
\label{JacobiSecondKindIntegralRepresentation}
\end{equation}
where
\[
F(\tau,s) \equiv \frac{P_m^{(p,p)}(\tau)-P_m^{(p,p)}(s)}{\tau-s}.
\]
Observe that for fixed $\tau\in[-1,1]$ the function $F(\tau,s)$
is actually a polynomial: both the numerator and denominator have a zero at $s=\tau$ so that the nominally rational expression actually reduces to a polynomial. In particular, the integral in the first term of \eqref{JacobiSecondKindIntegralRepresentation} is well-defined. One then has that 
\[
Q_m^{(p,p)}(\tau) = \frac{\mathcal{P}(\tau)}{(1-\tau^2)^p}+ P_m^{(p,p)}(\tau) Q_0^{(p,p)}(\tau),
\]
where $\mathcal{P}(\tau)$ is a polynomial satisfying, generically, that $\mathcal{P}(\pm 1)\neq 0$. 
Its specific form will not be relevant for the subsequent discussion. In view of equation \eqref{AnsatzSolution},  the singular behaviour in $a(\tau)$ at $\tau=\pm 1$ is contained in $Q_0^{(p,p)}(\tau)$. Now one has from equation (4.62.2) in \cite{Sze78} that, for $\tau > 0$,
\[
Q_0^{(p,p)}(\tau) = \frac{(-1)^{p+1}}{2}\log(1-\tau) + (1-\tau)^{-p}M( 1-\tau ), 
\]
where $M(s)$ is a power series convergent for $|s|<1$ with $M(0)\neq 0$, and in particular that 
\[
Q_0^{(0,0)}(\tau) =\frac{1}{2}\log\left( \frac{1+\tau}{1-\tau} \right).
\]
A similar behaviour can be deduced for $\tau < 0$ near $\tau=-1$. It follows then that the singular behaviour of $a(\tau)$  near $\tau=1$ is of the 
form
\[
(1-\tau)^{p}\log(1-\tau).
\]
That is, the general solution to the ODE \eqref{JacobiEqn:Raw} is of class $C^{p-1}$ at $\tau=1$. The behaviour at $\tau=-1$ can be obtained \emph{mutatis mutandi}: the singular behaviour at $\tau=-1$ is of the form $(1+\tau)^{p}\log(1+\tau)$. In summary, the generic solution to \eqref{JacobiEqn:Raw} is of class $C^{p-1}$, {and not $C^p$}, at $\tau=\pm 1$. 

To conclude, we record that the general solution of the Jacobi ODE \eqref{JacobiEqn:Raw2} with $m\geq 0$ ($\ell \geq p$) can be written (restoring indices) as   
\begin{equation}
a_{p;p+m,j}(\tau) = \left(\frac{1-\tau}{2} \right)^{p}\left( \frac{1+\tau}{2} \right)^{p} P_{m}^{(p,p)}(\tau) \left( \mathfrak{c}_{p;p+m,j}  + \mathfrak{d}_{p;p+m,j} Q_{0}^{(p,p)}(\tau) \right) + \mathfrak{d}_{p;p+m,j}\mathcal{P}_{p;m}(\tau) \label{JacobiLogarithmicSolution}
\end{equation}
with $\mathcal{P}_{p;m}(\tau)=0$ if $p=m=0$ but non-zero otherwise. Crucially, $\mathcal{P}_{p;m}(-1)\neq 0$ as long as $p\neq 0$ and $m\neq 0$. 

\begin{remark}
{\em A property of the solution \eqref{JacobiLogarithmicSolution} which plays a central role in the main text is that in the absence of logarithmic solutions, i.e. when $\mathfrak{d}_{p;p+m,j}=0$ the function $a_{p;p+m,j}(\tau)$ has zeros of order $p$ at both $\tau=\pm 1$. 

}
\end{remark}

\begin{remark}
{\em The discussion in this appendix in particular shows that for the wave equation with our regularity assumptions, there is complete symmetry in the presence of logarithmic divergences at the critical sets $\mathcal{I}^\pm$. This situation is in contrast with the situation for solutions to spin-$s$ fields with $s\neq 0$, where in the analogous analysis the spin $s \in \frac{1}{2} \mathbb{N}$ feeds into the power of $(1 \pm \tau)$ multiplying $\log (1\pm \tau )$ asymmetrically \cite{TauVal23}.}
\end{remark}

\begin{remark}
{\em The above analysis of the intrinsic transport equations on $\mathcal{I}$ can also be performed in terms of so-called ultraspherical Gegenbauer polynomials, see \cite{FueHen24}. Gegenbauer polynomials are special cases of Jacobi polynomials.}
\end{remark}

\begin{remark}
{\em Making use of the transformation formulae of \Cref{Section:TransformationCoordinates} the expressions for the coefficients $a_{p;\ell,j}(\tau)$ imply, in the particular case for which logarithmic terms are absent, a formal expansion of the form 
\[
\tilde{\phi} = \frac{1}{r}\left( \tilde{\varphi}^{(0)} + \sum_{p=1}^{m+4}\frac{1}{p!} \left(\frac{1}{2|u|} \right)^{p}\varphi^{(p)} + \mathcal{O}\left( \frac{1}{|u|^{m+5}} \right) \right) + \mathcal{O}\left( \frac{1}{r^2}\right),
\]
where the coefficients $\tilde{\varphi}^{(0)}$ and $\varphi^{(p)}$ are $u$-independent smooth functions on $\mathbb{S}^2$, and the index $m$ measures the regularity on $\underline{\mathcal{B}}_\varepsilon \cup \scri^-$. Here, while $\tilde{\varphi}^{(0)}$ can be made of any combination of spherical harmonics, the coefficient $\varphi^{(p)}$ is restricted to harmonics with $\ell < p$. This formal expansion was first reported in \cite{FueHen24}. We do not claim such an expansion for solutions arising in \Cref{thm:main_theorem_transverse_derivatives}, since this would require controlling the multipolar structure of the remainder in expansion \eqref{expansion_main_theorem_transverse_derivatives}.}
\end{remark}

\printbibliography

\end{document}